\documentclass[letterpaper,aps,pra,twocolumn,tightenlines,preprintnumbers,nofootinbib,showkeys,superscriptaddress,longbibliography]{revtex4-2}
\usepackage[utf8]{inputenc}
\usepackage{XCharter}
\usepackage{mathptmx}
\usepackage[T1]{fontenc}
\usepackage{fullpage}
\usepackage{graphicx}
\usepackage[dvipsnames]{xcolor}
\usepackage{hyperref}
\hypersetup{colorlinks=true,citecolor=Red,linkcolor=NavyBlue,urlcolor=Brown}
\usepackage[caption=false]{subfig}
 
\usepackage{natbib}
\usepackage{relsize}
\usepackage[left=1.65cm,right=1.65cm,top=2cm,bottom=2cm]{geometry}

\usepackage{array}
\usepackage{longtable}
\usepackage{dcolumn}
\usepackage{bm}
\usepackage{tabularx}
\usepackage{amsmath}  
\usepackage{amsfonts}
\usepackage{amssymb}
\usepackage{float}
\usepackage{amsthm}
\usepackage{thmtools}
\usepackage{thm-restate}
\usepackage{braket}
\usepackage{fullpage}
\usepackage[linesnumbered, ruled, vlined]{algorithm2e}
\usepackage{nicefrac}
\usepackage{bbm}
\usepackage{commath}
\usepackage{mathtools}
\usepackage{varwidth}
\usepackage[most]{tcolorbox}
\usepackage[thinc]{esdiff}

\def \eps {\varepsilon}

\def\Pr{\mathrm{Pr}}

\newcommand{\paren}[1]{\left( #1 \right)}
\def\Tr{\mathrm{Tr}}

\newcommand{\polylog}[1]{\mathrm{polylog}\paren{#1}}

\newcommand{\bigO}{\mathcal{O}}
\newcommand{\cmnt}[1]{}
\newtheorem{theorem}{Theorem} %[section]
\newtheorem{lemma}{Lemma}%[theorem]{Lemma}

\begin{document}

\title{Simulating quantum collision models with Hamiltonian simulations using early fault-tolerant quantum computers}

\author{Kushagra Garg}  
\email{kushagra.garg@research.iiit.ac.in}
\affiliation{Center for Computational Natural Sciences and Bioinformatics, International Institute of Information Technology, Hyderabad 500 032, India}
\affiliation{Center for Quantum Science and Technology, International Institute of Information Technology, Hyderabad 500 032, India}

\author{Zeeshan Ahmed}  
\email{zeeshan.ahmed@research.iiit.ac.in}
\affiliation{Center for Quantum Science and Technology, International Institute of Information Technology, Hyderabad 500 032, India}
\affiliation{Center for Security, Theory and Algorithmic Research, International Institute of Information Technology, Hyderabad 500 032, India}

\author{Subhadip Mitra}
\email{subhadip.mitra@iiit.ac.in}
\affiliation{Center for Computational Natural Sciences and Bioinformatics, International Institute of Information Technology, Hyderabad 500 032, India}
\affiliation{Center for Quantum Science and Technology, International Institute of Information Technology, Hyderabad 500 032, India}

\author{Shantanav Chakraborty}  
\email{shchakra@iiit.ac.in}
\affiliation{Center for Quantum Science and Technology, International Institute of Information Technology, Hyderabad 500 032, India}
\affiliation{Center for Security, Theory and Algorithmic Research, International Institute of Information Technology, Hyderabad 500 032, India}

\date{\today}% It is always \today, today,

\begin{abstract}
We develop randomized quantum algorithms to simulate quantum collision models, also known as repeated interaction schemes, which provide a rich framework to model various open-system dynamics. The underlying technique involves composing time evolutions of the total (system, bath, and interaction) Hamiltonian and intermittent tracing out of the environment degrees of freedom. This results in a unified framework where any near-term Hamiltonian simulation algorithm can be incorporated to implement an arbitrary number of such collisions on early fault-tolerant quantum computers: we do not assume access to specialized oracles such as block encodings and minimize the number of ancilla qubits needed. In particular, using the correspondence between Lindbladian evolution and completely positive trace-preserving maps arising out of memoryless collisions, we provide an end-to-end quantum algorithm for simulating Lindbladian dynamics. For a system of $n$-qubits, we exhaustively compare the circuit depth needed to estimate the expectation value of an observable with respect to the reduced state of the system after time $t$ while employing different near-term Hamiltonian simulation techniques, requiring at most $n+2$ qubits in all. We compare the CNOT gate counts of the various approaches for estimating the Transverse Field Magnetization of a $10$-qubit XX-Heisenberg spin chain under amplitude damping. Finally, we also develop a framework to efficiently simulate an arbitrary number of memory-retaining collisions, i.e.,  where environments interact, leading to non-Markovian dynamics. Overall, our methods can leverage quantum collision models for both Markovian and non-Markovian dynamics on early fault-tolerant quantum computers, shedding light on the advantages and limitations of simulating open systems dynamics using this framework.
\end{abstract}
\maketitle

%\tableofcontents

\section{\label{sec: Introduction}Introduction}

Simulating the dynamics of closed quantum systems, also known as Hamiltonian simulation, is widely considered to be one of the foremost applications of a quantum computer~\cite{lloyd1996universal, berry2014exponential, berry2015simulating, campbell2019random,low2019hamiltonian}. A natural extension is the problem of simulating open systems dynamics, wherein the underlying quantum system interacts with an environment, resulting in non-unitary dynamics. Over the years, several techniques have been put forth that capture the effects of environmental interactions on the underlying quantum system~\cite{breuer2002theory, rivas2012open}. Quantum master equations derived in a wide variety of settings accurately describe the system dynamics under both coherent evolution and environmental interactions. Of these, the most widely analyzed is the Gorini-Kossakowski-Sudarshan-Lindblad (GKSL) master equation (generally referred to as a Lindblad master equation), which provides a general model to describe Markovian quantum dynamics, i.e., when the system and the environment remain uncorrelated throughout the evolution~\cite{Gorini_1976, Lindblad1976}. This equation generates dynamics of completely positive trace-preserving (CPTP) maps, capturing effects such as spontaneous emission, dephasing, and dissipation~\cite{breuer2002theory, rivas2012open}. Besides being fundamental to our understanding of environmental effects on quantum systems, Lindbladian dynamics has found diverse applications in quantum computation, ranging from quantum error mitigation strategies~\cite{endo2018practical,kandala2019error} to recent quantum algorithms for preparing ground~\cite{ding2024single} and thermal states of Hamiltonians~\cite{chen2023quantum}. Thus, the problem of simulating open systems dynamics is of considerable importance. 

One would expect quantum computers to provide a significant advantage over their classical counterparts for this problem, owing to the favorable scaling of the Hilbert space dimension with the number of qubits. However, since most current methods require significant resources in the form of access to specialized oracles such as block-encodings~\cite{low2019hamiltonian, chakraborty2019power}, several ancilla qubits, and sophistical controlled operations~\cite{cleve2017efficient, childs2017efficient,li2023simulating, patel2023wave1, patel2023wave2}, they are ill-suited for the early fault-tolerant era, where quantum computers have limited ancilla space, and complicated control logic is absent~\cite{katabarwa2024early}. This naturally leads us to the question of whether open systems dynamics can be efficiently simulated on early fault-tolerant quantum computers. 

We consider quantum collision models, also known as repeated interaction schemes, which have emerged as a powerful framework to describe the dynamics of open quantum systems~\cite{bruneau2014repeated, ciccarello2022collisionreview, Ciccarello2013PRA, cattaneo2021collision, strasberg2017quantum}. In such models, the environment comprises several individual subsystems, each interacting with the system for a fixed time interval, one after another. The underlying idea is that environmental effects can arise from the repeated sequential interactions between the underlying system and the individual environment subsystems (referred to as \textit{sub-environments} throughout this article). Despite their simplicity, collision models can effectively describe a variety of physical quantum systems, since it is possible to engineer a broad range of dynamics through a sequence of relatively simple discrete interactions in this framework. For instance, the basic framework of a microscopic maser can be described using collision models \cite{filipowicz1986maser}. In quantum thermodynamics, collision models have been employed in the study of quantum batteries \cite{barra2019dissipative, seah2021battery}, Landauer's principle \cite{lorenzo2015landauer}, and quantum thermalization \cite{Scarani2002, Manatuly2019, strasberg2017quantum, Rodrigues2019, Leitch2022, Grimmer2018, Hammam_2022}. Other applications include the study of quantum optical systems \cite{ciccarello2017optics, Grimsmo2015, Whalen2017, Pichler2016, Fischer2018, Fischer_2018_JOPC, ciluffo2020collisional}, modelling of continuous measurements \cite{Gross_2018}, and quantum metrology or thermometry \cite{seah2019collisional, shu2020surpassing}. Indeed, Markovian dynamics naturally emerge from memory-less collisions, i.e., when the environment comprises many non-interacting sub-environments interacting exactly once with the system. On the other hand, in scenarios where an interaction between the system and an environment is followed by an interaction between two or more sub-environments, memory effects arise, leading to non-Markovian dynamics~\cite{Pellegrini2009-eh, Man2019, Lorenzo2017-je, Lorenzo2016, Ciccarello2013PRA}.    

There are clear advantages in simulating open systems dynamics using quantum collision models. The interaction between the system and an environment subsystem is simply the time evolution of the corresponding total Hamiltonian (sum of the system, sub-environment, and interaction Hamiltonian), and repeated interactions boil down to implementing a composition of Hamiltonian evolutions. Thus, these models provide an avenue to simulate various open quantum systems dynamics using only Hamiltonian simulation algorithms (which have seen remarkable progress) as core subroutines. Indeed, this includes both state-of-the-art Hamiltonian simulation algorithms with near-optimal complexities~\cite{berry2015hamiltonian, berry2015simulating, low2017optimal, low2019hamiltonian} as well as methods suitable for near-term implementation, such as qDRIFT and (low-order) Trotter methods, that are simple and provide better performance in practice~\cite{lloyd1996universal, childs2018toward, campbell2019random, childs2021theory}. Ref.~\cite{ding2024simulating} has recently explored the possibility of simulating Lindbladian dynamics using Hamiltonian simulations. It unravels Lindblad dynamics as a stochastic Schr\"{o}dinger equation and constructs an effective Hamiltonian from the Kraus operator representation obtained from the discretization of these unraveling equations. While this approach achieves near-optimal query complexity, it still requires several ancilla qubits and block-encoding access to the effective Hamiltonian to achieve optimal complexity.

In this paper, we build on the results of Ref.~\cite{pocrnic2023quantum}, where the authors provide explicit bounds for approximating Lindbladian dynamics through Markovian maps from quantum collision models. They compute the query complexity of this problem using Hamiltonian simulation by qubitization~\cite{low2019hamiltonian} and higher-order Trotter methods~\cite{childs2021theory}. In contrast, we explore end-to-end complexities of simulating open systems dynamics using quantum collision models restricted by a lack of ancilla qubits and no access to specialized oracles such as block encodings. We also consider near-term simulation methods such as low- and high-order Trotter methods~\cite{childs2021theory}, qDRIFT~\cite{campbell2019random}, and, in particular, incorporate the randomized Hamiltonian simulation technique of Refs.~\cite{wan2022randomized, wang2023qubit, chakraborty2023implementing}, which is a particular instance of a Linear Combination of Unitaries (LCU) that can be implemented efficiently on early fault-tolerant quantum computers.

We first consider a quantum collision model that naturally gives rise to Markovian dynamics. We assume the environment comprises a discrete set of subsystems, each interacting with the system for a time $\Delta t$, one by one, before being traced out. This process is repeated $K$ times, leading to a Markovian $K$-collision map. Given the total Hamiltonian $\overline{H}$ for each collision (the system, the corresponding sub-environment, and their interaction), we develop a randomized algorithm that outputs an $\varepsilon$-additive estimate of the expectation value of any observable $O$ with respect to the reduced state of the system after undergoing $K$ memoryless collisions (i.e., a Markovian $K$-collision map). This is achieved by composing a Hamiltonian simulation that implements $e^{-i\overline{H}\Delta t}$ to a precision of $\eps'\approx\frac{\eps}{K\|O\|}$, $K$ times. Using different near-term Hamiltonian simulation techniques, we exhaustively compare the end-to-end cost of estimating the desired expectation value.

Subsequently, we develop a randomized quantum algorithm for simulating Lindbladian dynamics, which can be shown as a particular instance of the  Markovian $K$-collision map. Our procedure (i) does not require access to block encodings of $H$, and (ii) is qubit efficient since it only employs near-term Hamiltonian simulation procedures. We exhaustively compare the cost (circuit depth, qubit count, and classical repetitions) of simulating Lindbladian dynamics using different Hamiltonian simulation techniques. For instance, we find that in terms of the evolution time and precision, a Hamiltonian simulation by the Single-Ancilla LCU (SA-LCU) method requires a circuit depth of $\widetilde{\bigO}(t^3/\eps)$.\footnote{Throughout this article, we follow standard complexity-theoretic notations. See Sec.~\ref{sec: notations}} This method outperforms both the first-order Trotter method and qDRIFT by a factor of $1/\eps$. On the other hand, second-order Trotterization has a circuit depth of $\widetilde{\bigO}(t^{9/4}/\eps^{5/4})$: compared to the Hamiltonian simulation by SA-LCU, it has a better dependence on $t$ but a worse dependence on the precision. Interestingly, for $2k$-order, with $k\to \infty$, the circuit depth is better only by a factor of~$t/\eps$ compared to first-order Trotter and qDRIFT. In fact, the circuit depth of higher-order Trotter is $\widetilde{\bigO}(t^2/\eps)$, which is optimal for simulating Linbladian dynamics using quantum collision models~\cite{cleve2017efficient} up to log factors. 

We also numerically benchmark the performance of simulating Lindbladian dynamics using collision models. For this, we consider a $10$-qubit transverse-field Ising Hamiltonian under amplitude damping. Our goal is to estimate the average transverse magnetization of the reduced quantum state of the system up to an additive accuracy of $\varepsilon$. If we demand high precision, say $\varepsilon=10^{-4}$, the Hamiltonian simulation by SA-LCU has a CNOT gate count (per coherent run) $200$ times lower than that of the second-order Trotter method and $2000$ times lower than qDRIFT. On the other hand, for low precision and long evolution time $t$, the higher-order Trotter methods outperform other near-term techniques.

Finally, we also explore the possibility of simulating non-Markovian collisions on early fault-tolerant quantum computers. We extend the notion of Markovian $K$-collision maps to include interactions between the environment subsystems. We consider the framework of Ciccarello et al.~\cite{Ciccarello2013PRA}, where a collision between the system and an environment subsystem is followed by a collision between two consecutive environment subsystems. More precisely, the $j$-${\mathrm{th}}$ iteration comprises the following two interactions: first, the system interacts with the sub-environment $j$ and then there is an interaction between the sub-environments $j$ and $j+1$. We simulate a total of $K$ iterations (non-Markovian $K$-collision map) of these interaction sequences using near-term Hamiltonian simulation procedures, requiring circuit depths that scale similarly to the Markovian case. Thus, our overall framework allows for the possibility of simulating a broad range of open systems dynamics through quantum collision models.

The remainder of this paper is organized as follows. Sec.~\ref{sec: notations}, we formalize the notations we use throughout the article. We present a framework for simulating memoryless collisions in Sec.~\ref{sec:Modified Collision Model}, and develop a randomized quantum algorithm to simulate a Markovian $K$-collision map. Sec.~\ref{sec:Lindbladian Dynamics Simulation} applies this to simulate Lindbladian dynamics and comprises detailed comparisons of the complexities. We also benchmark the gate counts for the estimation of the overall magnetization of the transverse-field Ising model under amplitude damping. In Sec.~\ref{sec:Non-markovian dynamics}, we extend our approach to simulate non-Markovian collisions and define a non-Markovian $K$-collision map. We conclude by summarizing our results and discussing possible open problems in Sec.~\ref{sec:Discussion and Outlook}.

\section{Notations}
\label{sec: notations}
We use $g(n)= \bigO(f(n))$  to imply that $g$ is upper bounded by $f$, i.e., there exist constants $k_1$ and $k_2$ such that $\forall n>k_1$, $g(n)\leq k_2\cdot f(n)$. We also follow the standard convention of using \textit{tilde} ($\sim$) to hide polylogarithmic factors. For instance, $\widetilde{\bigO}(f(n))=\bigO(f(n)\mathrm{polylog}(f(n)))$. The trace of an operator $A$ is denoted by $\Tr[A]$, while the expectation value of the operator will be denoted by $\mathbb{E}[A]$. The probability of an event $X$ occurring will be denoted by $\Pr[X]$. 

We use operator as well as superoperator norms. The Schatten $p$-norm of the operator $X$ is defined as 
\begin{equation*}
\norm{X}_p =\bigg(\sum_{j} \sigma^{p}_j(X)\bigg)^{1/p},    
\end{equation*}
where $\sigma_j(X)$ is the $j$-${\mathrm{th}}$ singular value of $X$. So if $\sigma_{\max}(X)$ denotes the maximum singular value of $X$, we have 
\begin{equation*}
\lim_{p\rightarrow\infty} \norm{X}_p = \sigma_{\max} \cdot \lim_{p\rightarrow\infty} \bigg(\sum_{j} \dfrac{\sigma^{p}_j(X)}{\sigma^{p}_{\max}(X)}\bigg)^{1/p}=\ \sigma_{\max},
\end{equation*}
which is the spectral norm of the operator $X$. We will denote this as simply $\norm{X}$. For instance, for any density operator $\rho$, we have $\norm{\rho}_1=1$, while for any unitary $U$, the spectral norm $\norm{U}=1$. For a superoperator $\mathcal{M}$ which maps operators to operators, we define the induced $1$-norm of $\mathcal{M}$ as follows: 
\begin{equation*}
\|\mathcal{M}\|_{1 \to 1} = \sup_{\rho \neq 0} \frac{\|\mathcal{M}[\rho]\|_1}{\|\rho\|_1}.
\end{equation*}

We provide a comprehensive list of all mathematical symbols used in this paper, along with their definitions in Appendix~\ref{sec:symbols} (see Table \ref{table:notation-table}).

\section{Quantum Collision Models}\label{sec:Modified Collision Model}
Collision models are a versatile framework for simulating the dynamics of open quantum systems, where the system interacts with its environment through a sequence of discrete interactions or `collisions.' Unlike continuous-time approaches, collision models treat the environment as a set of discrete, independent subsystems (referred to as sub-environments throughout this article) that interact with the system one at a time. This setup provides valuable physical insights and a constructive approach to simulating complex open systems dynamics. For instance, repeated interactions between the system and each environment subsystem may be as follows: each sub-environment (initially uncorrelated with both the system and other sub-environments) interacts with the system for a short duration before being traced out, dissipating information and energy while resulting in a CPTP map. The overall dynamics of these collisions is described by a composition of the CPTP maps, allowing for the modeling of Markovian (Lindbladian) dynamics. In this section, we discuss the possibility of simulating an arbitrary number of such collisions using early fault-tolerant quantum computers. %\textcolor{purple}{Beyond this, quantum collision models can also model more complex behavior. Memory effects arise if correlations persist between successive interactions, leading to non-Markovian dynamics. Concretely, this manifests via system-environment feedback or inter-environmental correlations, each offering unique insights into non-Markovian quantum processes \cite{Ciccarello2013PRA}, which we discuss in Sec.~\ref{sec:Non-markovian dynamics}.}

\subsection{Markovian $K$-collision map}
\label{subsec:K-collision-approx}
Let us consider a quantum system in an $n$-qubit Hilbert space $\mathcal{H}_S$ that is coupled to a quantum environment belonging to the space $\mathcal{H}_E$, formed out of the tensor product of $m$ environment subsystems: $\mathcal{H}_E = \mathcal{H}_{E_1} \otimes \mathcal{H}_{E_1} \dots \otimes \mathcal{H}_{E_{m}}$. We denote the system Hamiltonian by $H_S$. The environment is a discrete sum of environment subsystems, with the Hamiltonian of the sub-environment $j$ denoted by $H_{E_j}$. The $j$-${\mathrm{th}}$ collision corresponding to the interaction between the system $S$ and the $j$-th sub-environment, $E_j$, is denoted by the interaction Hamiltonian $H_{I_j}$. Thus, the total Hamiltonian is given by 
\begin{equation}
\label{eq:total-Hamiltonian}
H=H_S+\sum_{j=1}^{m}\left(H_{I_j}+H_{E_j}\right).\end{equation}
We assume that the total Hamiltonian corresponding to the $j$-${\mathrm{th}}$ collision can be expressed as a Linear combination of Pauli operators, i.e., 
\begin{align}
\label{eq:j-th collision Hamiltonian}
    H_j &=  H_S + H_{E_j} + H_{I_j} = \sum_{i=1}^{L_j}h_{i,j}P_{i,j},
\end{align}

where $P_{i, j} \in \pm\{I, \sigma^x, \sigma^y, \sigma^z\}^{\otimes~n}$ and, without loss of generality, each $h_{i,j} \in \mathbb{R}^{+}$. Also, let $L=\max_j L_j$. Finally, for $1\leq j \leq K$, we denote the normalized Hamiltonians $\overline{H}_j = H_j/\beta_j$, where $\beta_j= \sum_{i=1}^{L_j} h_{i,j}$. Note that this is also without loss of generality, as we can always rescale $\Delta t$ by multiplying it with $\beta_j$.

Now, we are in a position to describe the dynamics induced by the $j$-${\mathrm{th}}$ collision: the system, in state $\rho_S$, interacts with the $j$-${\mathrm{th}}$ sub-environment, initialized in the state $\rho_{E_j}$, for a duration of $\Delta t$. The dynamics induced by the $j$-th collision is then given by the time evolution operator
\begin{equation}
\label{eq: Collision unitary}
    U_j = e^{-i\beta_j \overline{H}_j\Delta t}
\end{equation}
% \cmnt{To capture non-Markovian effects, we allow interactions between pairs of sub-environments. We represent these interactions by Unitaries $U_{i, j}$. The subscript indicate the specific sub-environment involved. For simplicity of the analysis, we only consider interaction between two consecutive sub-environments, thus choosing unitaries of type $U_{j, j+1}$. In each collision, we first apply the unitary $U_j$, evolving the system and interacting it with the $j$-th sub-environment. Subsequently, a new environment state $\rho_{E_{j+1}}$ is initialized. The new environment and the previous one then interact through the unitary $U_{j, j+1}$, after which the initial environment is traced out. We formalize this though the definition of a collision map as follows:}
Following this, the sub-environment is traced out. Formally, we define a Markovian collision map: 
\begin{restatable}[Markovian collision map]{definition}{}
\label{def: collision map}
Let $j\in [1, K]$, $\rho_{E_{j}}$ represent the initial state of the $j$-th sub-environment, and $U_j$ be the unitary representing the interaction between the system and the $j$-th sub-environment. The $j$-th collision map $\Phi_j$ is defined as:
    \begin{equation}
        \Phi_j[.] \equiv \Tr_{E_{j}}\left[U_j\left(. \otimes \rho_{E_{j}}\right)U_j^{\dagger}\right]      
        % \Phi_j[.] \equiv \Tr_{E_{j}}\left[U_{j, j+1}U_j\left(. \otimes \rho_{E_{j+1}}\right)U_j^{\dagger}  U_{j, j+1}^{\dagger}\right].
    \end{equation}
Here, $\Tr_{E_j}$ denotes the partial trace over the $j$-th sub-environment, resulting in a reduced state for the system.
\end{restatable}
\noindent
For multiple collisions, this process is repeated. The resulting iterative sequence for $K$ collisions can be described by composing $K$ such collision maps. Formally, we define a Markovian $K$-collision map as follows:

\begin{restatable}[Markovian $K$-collision map]{definition}{}
\label{def: K-collision map}
Let $\Phi_1$, $\Phi_2$, $\dots$,  $\Phi_K$ be the collision maps from Definition \ref{def: collision map}. The $K$-collision map $\mathcal{M}_K$ is the composition of these maps, defined as:  
    \begin{align}
        \mathcal{M}_{K}[\cdot] \equiv  \left[\bigcirc_{j=1}^{K} \Phi_{j} \left[ \cdot\right]\right].  
    \end{align}    
\end{restatable}
\noindent
Thus, the $K$ collisions correspond to composing the collision map (evolving under $H$ for time $t=\Delta t$ followed by tracing out the corresponding sub-environment) $K$ times. The Hamiltonian simulation is the key subroutine in simulating a Markovian $K$-collision map on a quantum computer. Our goal is to develop a randomized quantum algorithm that estimates the expectation value of an observable $O$ for a system (initialized in $\rho_S$) after $K$ collisions, up to an additive precision $\eps$. That is, the algorithm outputs $\mu$ such that
\begin{equation}
\label{eq: error between experiment and the collision map}
\left|\mu -\Tr\left[O\mathcal{M}_K\left[\rho_{S}\right]\right]\right| \leq \eps.
\end{equation}
In what follows, we estimate the precision required of any Hamiltonian simulation procedure to estimate $\mu$ with $\varepsilon$-additive accuracy. Let us denote by $\widetilde{U}_j$, the circuit corresponding to the Hamiltonian simulation procedure implementing $e^{-i\Delta t\overline{H}_j}$ to some accuracy (to be determined later). Then, it is possible to define approximate versions of the collision maps from Definitions \ref{def: collision map} and \ref{def: K-collision map}, respectively, where the exact time evolution operator $U_j$ is now replaced by $\bar{U}_j$: 
\begin{equation}
\label{eq: approximate collision map}
        \widetilde{\Phi}_j[\cdot] \equiv \Tr_{E_{j}}\left[\widetilde{U}_j\left(\cdot \otimes \rho_{E_{j}}\right)\widetilde{U}_j^{\dagger}\right],
    \end{equation}
and subsequently, the approximate $K$-Collision map is defined as   
\begin{align}
\label{eq:approx-K-collision-map}
        \widetilde{\mathcal{M}}_{K}[\cdot] \equiv  \left[\bigcirc_{j=1}^{K} \widetilde{\Phi}_{j} \left[ \cdot \right]\right].
\end{align}    

We now prove that if the Hamiltonian simulation procedures in the construction of the approximate Markovian $K$-collision map are implemented with a precision of $\varepsilon/(3 K \norm{O})$, then the expectation value of observable $O$ with respect to the state of the system with respect to the map $\widetilde{M}_K$ is $\varepsilon$-close to the desired expectation value. We do so via the following lemma: 
\begin{restatable}[Bounds on the Markovian collision map]{lemma}{}\label{lemma: approximate K-collision map distance bound}
    Let $O$ be an observable and $\widetilde{\mathcal{M}}_K$ represent the approximate Markovian $K$-collision map in Eq.~\eqref{eq:approx-K-collision-map}, 
    where
    \begin{equation}
        \max_{1\leq j \leq K}\norm{U_j - \widetilde{U}_j} \leq \frac{\eps}{3K\norm{O}}.
    \end{equation}
    Then, the expectation value of $O$ over a state evolved under the approximate map satisfies,
    \begin{equation}
        \left|\Tr\left[O\mathcal{M}_K[\rho]\right] - \Tr\left[O \widetilde{\mathcal{M}}_K[\rho]\right]\right| \leq \eps.
    \end{equation}
\end{restatable}
\begin{proof}
We consider the error between the operations performed on the state $\rho$ under $U_j$ and $\widetilde{U}_j$. More precisely, let $\norm{U_j-\widetilde{U}_j}\leq \xi_j$, where we denote the maximum error in any of the Hamiltonian simulation procedures in the definition of the approximate $K$-collision map by $\xi_{\max}$, i.e., $\xi_{\max}=\max_{1\leq j \leq K} \xi_{j}$. Then, using Theorem \ref{thm:distance-expectation}, we have
    \begin{align}
        &\big\| U_j \rho  U_j^{\dagger} - \widetilde{U}_j \rho \widetilde{U}_j^{\dagger}\big\|_1 
        \leq 3 \xi_j,
    \end{align}
for any quantum state $\rho$. Since partial trace is a CPTP operation, it is contractive under the trace norm. Therefore, we obtain,
    \begin{align}
        &\bigg\| \Tr_{E_j} \left[U_j \rho U_j^{\dagger} \right]- \Tr_{E_j} \left[ \widetilde{U}_j \rho \widetilde{U}_j^{\dagger} \right] \bigg\|_1 \leq 3 \xi_j.
    \end{align}
This inequality implies that the trace distance between the exact map $\Phi_j[\rho]$ and the approximate map $\widetilde{\Phi}_j[\rho]$ is bounded by $3 \xi_j$:
    \begin{align}
        \left\| \Phi_j[\rho] - \widetilde{\Phi}_j[\rho] \right\|_1 \leq 3 \xi_j.
    \end{align}
So,
    \begin{align}
        \max_j \norm{\Phi_j[\rho] - \widetilde{\Phi}_j[\rho]}_1 \leq 3 \xi_{\max}.
    \end{align}
Now, from Lemma \ref{thm: Bounds on the distance between the composition of maps}, we obtain,
    \begin{align}
        \left\| \bigcirc_{j=1}^{K} \Phi_j[\rho] - \bigcirc_{j=1}^{K} \widetilde{\Phi}_j[\rho] \right\|_1 \leq 3 K \xi_{\max}.
    \end{align}
Thus, we have
    \begin{equation}
        \| \mathcal{M}_K[\rho] - \widetilde{\mathcal{M}}_K[\rho] \|_1 \leq 3 K \xi_{\max}.
    \end{equation}
Finally, we apply Theorem \ref{thm:distance-expectation} once again to bound the difference in the expectation values of the observable $O$ under $\mathcal{M}_K$ and $\widetilde{\mathcal{M}}_K$:
   \begin{align}
        \left| \Tr\left[O \mathcal{M}_K[\rho] \right] -\Tr\left[O\widetilde{\mathcal{M}}_K[\rho] \right] \right| \leq 3 K \|O\| \xi_{\max}.
    \end{align}
Therefore, by choosing $\xi_{\max} = \varepsilon/(3K\norm{O})$, we can ensure that
    \begin{equation}
        \left| \Tr\left[O \mathcal{M}_K[\rho] \right] - \Tr\left[O \widetilde{\mathcal{M}}_K[\rho] \right] \right| \leq \eps.
    \end{equation}
\end{proof}

Now, we can compare the complexities of using different Hamiltonian simulation algorithms to implement a Markovian $K$-collision map on a quantum computer. As mentioned earlier, we will primarily focus on near-term Hamiltonian simulation techniques, i.e., ones that do not require multiple ancilla qubits. We will consider Trotter methods, qDRIFT, and also the SA-LCU method of Ref.~\cite{chakraborty2023implementing}. We need to estimate the cost of composing $K$ Hamiltonian simulation algorithms, each implementing $e^{-i\overline{H}\Delta t}$ to a precision $\varepsilon/(3K\norm{O})$, from Lemma~\ref{lemma: approximate K-collision map distance bound}. So, if a Hamiltonian simulation algorithm requires a circuit depth of
$$
\tau\left(\Delta t,~\eps/(K\|O\|)\right),
$$
to implement $U_j$, up to an additive precision of $\bigO(\eps/(K\|O\|))$, the $K$-collision map can be implemented within a circuit depth of
$$
\tau_d=\bigO\left(K \tau\left(\Delta t,~\eps/(K\|O\|)\right) +K\tau_{\rho_E}\right),
$$
where $\tau_{\rho_E}$ is the maximum of the circuit depths of the unitaries preparing the sub-environments, i.e., if for $j\in [1,K]$, if $\tau_{E_j}$ is the circuit depth of preparing $\rho_{E_j}$, then $\tau_{\rho_E}=\max_{j\in [1, K]}\tau_{\rho_{E_j}}$. While any near-term Hamiltonian simulation procedure can be used to implement a Markovian $K$-collision map, the SA-LCU algorithm~\cite{chakraborty2023implementing} demands some attention. This method, described in the next section, expresses each $U_j$ as a linear combination of strings of Pauli operators (say $\widetilde{U}_j$) with the total weight of the coefficients $\alpha=\bigO(1)$. Now, composing individual Hamiltonian simulations for implementing the Markovian $K$-collision map can potentially lead to an exponential scaling $\alpha^K$, detrimentally affecting the circuit depth. However, below, we show that it is possible to implement the Markovian $K$-collision map while bypassing this exponential scaling. 

\begin{figure*}[ht]
\centering{
    \scalebox{0.7}{
    \includegraphics[width=1\textwidth]{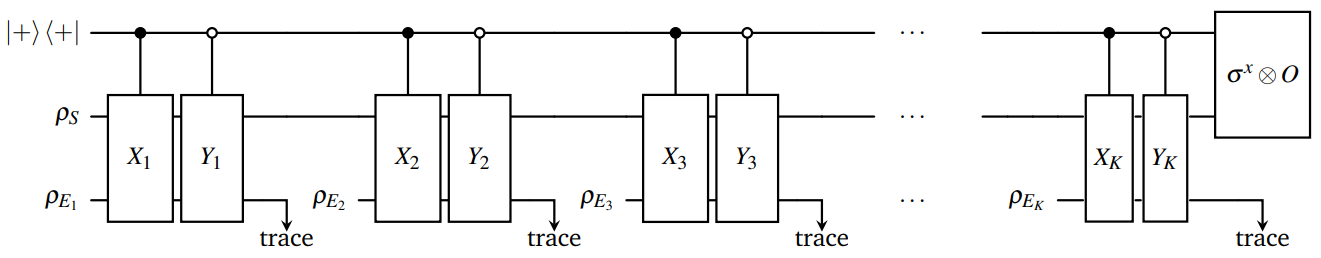}         
    }}
    \caption{
The quantum circuit for simulating a $K$-collision map using Hamiltonian simulation by SA-LCU. The algorithm applies controlled and anti-controlled sampled unitaries ($X_j$ and $Y_j$, respectively) for the interaction between the system and each sub-environment, following which the latter is traced out. This sequence is repeated $K$ times, corresponding to the $K$ collisions. At the end of the process, the ancilla qubit and the system are measured. Notably, only a single environment register suffices, as it can be reused following the tracing out of the previous environment subsystem.
}\label{fig: single ancilla collision model circuit}
\end{figure*}

\subsection{Using Hamiltonian simulation by Single-Ancilla LCU to simulate the Markovian $\mathbf{K}$-collision map}
\label{subsec:ham-sim-single-ancilla-k-collision}
%We begin by discussing briefly the Hamiltonian simulation technique by Single-Ancilla LCU. 
Given a Hamiltonian $H=\sum_{\ell=1}^{L}p_{\ell} P_{\ell}$, where $P_{\ell}$ is a string of Pauli operators and $\sum_{\ell}p_{\ell}$, there is a way to express $U=e^{-i\tau H}$ as a linear combination of Clifford gates and Pauli rotations \cite{wan2022randomized, wang2023qubit, chakraborty2023implementing}. This involves writing down the Taylor series expansion of the time evolution operator $U$ and truncating after some $q$ terms to obtain $\widetilde{U}$, which is now an LCU, expressed as $\widetilde{U}=\sum_{j}\alpha_j W_j$, such that $\norm{U-\widetilde{U}}\leq \varepsilon$ for $q=\bigO(\log(r)/\log\log(r/\varepsilon))$. The parameter $r$ is crucially chosen so that the total weight of the LCU coefficients of $\widetilde{U}$, given by $\alpha=\sum_{j}|\alpha_i|\leq e^{-\tau^2/r}$, converges. Each $W_i$ is a sequence of $q$ Clifford operators and a single Pauli rotation, repeated $r$ times. We formally state this in the following Lemma.
\begin{restatable}[LCU decomposition of time-evolution operator \cite{wan2022randomized, wang2023qubit, chakraborty2023implementing}]{lemma}{}\label{lemma: lcu of time evolution}
    Let $H = \sum_{\ell=1}^L p_{\ell} P_{\ell}$ be a Hermitian operator expressed as a convex combination of strings of Pauli operators $P_{\ell}$. Then, we can construct an Unitary operator $\widetilde{U}$ represented as an LCU that approximates the time evolution operator $U = e^{-i\tau H}$ satisfying: 
    \begin{equation}
        \norm{U - \widetilde{U}}\leq \eps,
    \end{equation}
    where 
    \begin{equation}\label{eq: index set LCU decomposition of time evolution operator}
        \widetilde{U} = \sum_{i} \alpha_i W_i,
    \end{equation}
    and
    \begin{equation}
        \alpha = \sum_i |\alpha_{i}| \leq e^{\tau^2/r}.
    \end{equation}
    Here, each $W_j$ is a sequence of $q$ Clifford operators followed by a single Pauli rotation, repeated $r$ times. The parameter $r$ is chosen such that $\tau/r<1$ (to ensure convergence of the series) and the parameter $q$ controls the precision of the approximation (the order at which the Taylor series is truncated), which we choose to be
    \begin{equation}
        q = \bigO\left(\dfrac{\log(r/\eps)}{\log\log(r/\eps)}\right).
    \end{equation}
\end{restatable}
\noindent
We present a derivation of the LCU decomposition in Appendix~\ref{appendix: LCU decomposition of Unitaries} for completeness. Clearly, for $r=\bigO(\tau^2)$, the total weight $\alpha=\bigO(1)$. Interestingly, the LCU $\widetilde{U}$ can be implemented using only a single ancilla qubit~\cite{chakraborty2023implementing} such that the expectation value of any observable $O$ with respect to the time-evolved state can be estimated up to $\varepsilon$-additive accuracy by running a simple quantum circuit of depth $\bigO(rq)=\widetilde{\bigO}(\tau^2)$, a total of $T=\bigO( \|O\|^2/\varepsilon^2)$ times.

Despite its near-term applicability, there are some issues with using this method as a subroutine to implement a Markovian $K$-collision map. The choice of $r=\bigO(\tau^2)$ only ensures that $\alpha=\bigO(1)$, and composing $K$ such Hamiltonian simulation procedures results in the total weight of $\alpha^K$, which grows exponentially with $K$. This affects the complexity of implementing the $K$-collision map. However, we show that this can be avoided by choosing a different value of the parameter $r$ when considering the composition of SA-LCU Hamiltonian simulation algorithms. This leads to a procedure that can estimate the expectation value of $O$ with respect to a state that has undergone the $K$ memoryless collisions up to an $\varepsilon$-additive accuracy while using the SA-LCU Hamiltonian simulation algorithm as a subroutine.  
%
%Now, we describe the randomized algorithm that uses Hamiltonian simulation by a Single-Ancilla LCU to estimate the desired expectation value up to an $\varepsilon$-additive accuracy. 
Each run of the algorithm involves running the quantum circuit shown in Fig.~\ref{fig: single ancilla collision model circuit}. The outcome of this circuit is a random variable that, in expectation value, estimates the desired quantity. The overall algorithm (outlined in Algorithm \ref{algo: collision model}) involves running this circuit a total of $T$ times and estimating the sample average of the outcomes. 

We require three registers: the system register, the environment register, and a single qubit ancilla register. The system register is initialized in the state $\rho_S$; the environment registers in the state $\rho_{E_1}$, and the ancilla qubit in the state $\ket{+}$. Let us first discuss the implementation of the approximate collision map $\widetilde{\Phi}_j$, corresponding to the $j$-${\mathrm{th}}$ collision. We recall that the (normalized) total Hamiltonian for the $j$-${\mathrm{th}}$ collision $H_j$ ($\|\overline{H}_j \| = 1$) can be expressed as a convex combination of strings of Pauli operators, i.e., $\overline{H}_j = \sum_{k=1}^L p_{jk} P_{jk}$. Then, from Lemma~\ref{lemma: lcu of time evolution}, we pick the parameter $r_j$ (to be determined later) such that for any $j\in [1,K]$,
\begin{equation}
\label{eq: LCU decomposition of Unitary}
    \widetilde{U}_j = \sum_{k} \alpha_{jk} W_{jk},
\end{equation}
and
$$
\norm{U_j-\widetilde{U}_j}\leq \dfrac{\varepsilon}{6 K\norm{O}},
$$
where $U_j=e^{-i\beta_j\Delta t\overline{H}_j}$, here we have scaled the simulation time. Note that from Lemma~\ref{lemma: lcu of time evolution}, each $W_{ji}$ is a string of $q$ Pauli operators and a single Pauli rotation, repeated $r_j$ times, where
\begin{equation}
\label{eq:truncation-taylor-series-collision-map}
    q=\bigO\left(\dfrac{\log\left(r_j K\norm{O}/\varepsilon\right)}{\log\log\left(r_j K\norm{O}/\varepsilon\right)}\right),
\end{equation}
and $\alpha^{(j)}=\sum_{k}\alpha_{jk}\leq e^{-\beta^2_j\Delta t^2/r^2_j}$.

%%%%%%%%%%%%%%%%%%%%%%%%%%%%%%%%%%%%%%%%%%%%%%%%%%%%%%%%%%%%%%%
\begin{algorithm}[t]
\SetAlgoCaptionLayout{bottom}
\caption{Algorithm to estimate the expectation value of an observable $O$ with respect to a quantum state evolved under a $K$-collision map.}\label{algo: collision model}
\KwIn{Initial system state in $\rho_S$, sub-environment states $\rho_{E_1}, \dots, \rho_{E_K}$, observable $O$, unitaries $\widetilde{U}_1$, \dots. $\widetilde{U}_K$, and precision $\eps'$, where the LCU decomposition of each $\widetilde{U}_j=\sum_{k}\alpha_{jk} W_{jk}$, such that $\forall j\in [1, K]$, $\norm{\widetilde{U}_j-e^{-i\Delta t \beta_j\overline{H}_j}}\leq \eps'$.~\\~\\}
    \begin{itemize}
        \item[1.~] Initialize the system and the ancilla in the state $\rho_S$ and $\ket{+}$, respectively.
        \item[2.~] For each collision step, from $j = 1$ to $K$:
            \begin{itemize}
                \item[a.~] Initialize the environment register in state $\rho_{E_j}$.
                \item[b.~] Draw two i.i.d.~samples $X_j$ and $Y_j$ from the ensemble 
                    \begin{equation*}
                        \mathcal{D}_{j}=\left\{ W_{jk}, \frac{\alpha_{jk}}{\alpha^{(j)}}\right\},
                    \end{equation*} 
                    where $\alpha^{(j)} = \sum_k|\alpha_{jk}|$
                \item[c.~] Apply the controlled unitary $X_j^{(c)}$ and the anti-controlled unitary $Y_j^{(a)}$ to the system 
                \item[d.~] Perform a partial trace over the environment register. 
            \end{itemize}
      \item [3.~]  Measure the joint ancilla and system state on the observable $(\sigma^x \otimes O)$ and record the measurement outcome as $\mu_i$.
      \item [4.~] Repeat Steps 1 to 3 a total of $T$ times.  
      \item [5.~] Compute the final estimate $\mu$ as:
    \begin{equation*}
        \mu = \dfrac{\zeta^2}{T}\sum_{j=1}^{T} \mu_j,
    \end{equation*}
      where $\zeta = \prod_{j=1}^K \alpha^{(j)}$.
      \end{itemize}
  \KwOut{Estimated expectation value $\mu$}  
\end{algorithm} 
%%%%%%%%%%%%%%%%%%%%%%%%%%%%%%%%%%%%%%%%%%%%%%%%%%%%%%%%%%%%%%%

Next, we draw two independent and identically distributed (i.i.d.) samples, $X_j$ and $Y_j$, from the ensemble
    \begin{equation}
    \label{eq: probablity emsemble}
    \mathcal{D}_{j}=\left\{ W_{jk}, \frac{\alpha_{jk}}{\alpha^{(j)}}\right\}.
    \end{equation}
Note that $\mathbb{E}[X_j]=\mathbb{E}[Y_j]=\widetilde{U}_j/\alpha^{(j)}$. Then, we coherently apply the controlled version $X_j^{(c)}$ and the anti-controlled version $Y_j^{(a)}$ of these sampled unitaries: %, which we define as
\begin{align}
    \label{eq: control and anti-control notation}
    X_j^{(c)} &= \ket{0}\bra{0} \otimes \mathbb{I} + \ket{1}\bra{1} \otimes X_j,\\
    Y_j^{(a)} &= \ket{0}\bra{0} \otimes Y_j + \ket{1}\bra{1} \otimes \mathbb{I}.  
\end{align}
After applying $X_j^{(c)}$ and $Y_j^{(a)}$, we trace out the environment register, which can be reused for the next collision. 

We repeat these steps for $j=1$ to $K$ (as outlined in step 2 of Algorithm \ref{algo: collision model}) to implement the approximate Markovian $K$-collision map, $\widetilde{\mathcal{M}}_K$. Then, we measure the observable $\sigma^x\otimes O$ in the ancilla and the system register. This corresponds to a single run of the algorithm. The outcome of the $j$-${\mathrm{th}}$ run (for any $j\in [1,T]$) is a random variable $\mu_j$ whose expectation value is given as, 
$$
\mathbb{E}[\mu_j]=\Tr[O~\mathcal{M}_K[\rho_S]]/\zeta^2,
$$
where $\zeta = \prod_{j=1}^K \alpha^{(j)}$. Overall, by taking some $T$ repetitions of this procedure, we collect random variables $\{\mu_{j}\}_{j=1}^{T}$, such that 
$$
\mu=\dfrac{\zeta^2}{T}\sum_{j=1}^{T}\mu_j,
$$
approximates the desired expectation value within $\varepsilon$-additive accuracy with a success probability of at least $1-\delta$. We prove the validity of Algorithm~\ref{algo: collision model} with the following Theorem.
% \LinesNumbered
%\begin{widetext}

%\end{widetext}
%
\begin{restatable}[]{theorem}{}
\label{thm: Algorithm 1 proof}
Let $\varepsilon, \delta \in (0,1)$. Then, for $\eps'=\varepsilon/(6K\norm{O})$, Algorithm \ref{algo: collision model} outputs $\mu$ with at least $1 - \delta$ probability such that
$$
\left|\mu-\Tr[O~\mathcal{M}_K[\rho_S]]\right|\leq \varepsilon,
$$
using $T$ repetitions of the circuit shown in Figure \ref{fig: single ancilla collision model circuit} where
\begin{align}
\label{eq: number of classical repetitions}
       T = \bigO\left( \frac{\|O\|^2 \log(1/\delta)}{\eps^2} \right).
\end{align} 
Each such coherent run has a circuit depth of 
\begin{equation}
\label{eq: circuit depth per run}
    \tau_d = \bigO\left(\beta ^2K^2\Delta t^2 \dfrac{\log(\beta  K \norm{O}\Delta t/\eps)}{\log\log(\beta  K \norm{O}\Delta t/\eps)}+K\tau_{\rho_E}\right).
\end{equation}
Here, $\beta = \max_j \beta_j$ and $\tau_{\rho_E}=\max_{j}\tau_{\rho_{E_j}}$, where $\tau_{\rho_{E_j}}$ is the circuit depth of the unitary preparing the sub-environment in the state $\rho_{E_j}$.
\end{restatable}
\begin{proof}
Following Algorithm \ref{algo: collision model}, we initialize the ancilla and system registers in the state 
    \begin{equation}
    \rho_0 = \ket{+}\bra{+} \otimes \rho_S
    \end{equation}
and the environment register in the state $\rho_{E_1}$. As shown in Fig.~\ref{fig: single ancilla collision model circuit}, we implement $X^{(c)}_1$ and $Y^{(a)}_1$ by sampling $X_1$ and $Y_1$ from $\mathcal{D}_1$, followed by tracing out of the environment $E_1$. The resulting state is
        \begin{align}
            \rho_1 &= \Tr_{E_1}\left[Y_1^{(a)}X_1^{(c)} (\rho_0 \otimes \rho_{E_1}) X_1^{(c)^{\dagger}}Y_1^{(a)^{\dagger}}\right]
        \end{align}
For simplicity, let us define the map $\Phi^{(PQ)}_j$ as
\begin{equation}
    \Phi^{(PQ)}_j\left[.\right] = \Tr_{E_j} \left[ P \left( . \otimes \rho_{E_{j}} \right) Q^{\dagger} \right],
\end{equation}
which represents applying the operator $P$ from the left and $Q^{\dag}$ from the right, followed by the tracing out of the environment register $E_j$. We can write $\rho_1$ in this notation as follows:
\begin{align}
    \rho_1 = &\frac{1}{2} \bigg[
    |0\rangle\langle0| \otimes \Phi^{(YY)}_1[\rho_S ]
    \nonumber + |0\rangle\langle1| \otimes \Phi^{(YX)}_1[\rho_S ] \nonumber \\
    &+ |1\rangle\langle0| \otimes \Phi^{(XY)}_1[\rho_S ]
    \nonumber + |1\rangle\langle1| \otimes \Phi^{(XX)}_1[\rho_S ]
    \bigg].
    \label{eq: expanded form of rho 1}
\end{align}

Next, we prepare the environment register in the state $\rho_{E_2}$, and apply the next set of unitaries (obtained by sampling from $\mathcal{D}_2$), followed by the tracing out the environment $E_2$. It is easy to see that this sequence leave us with a composition of the maps, resulting in the state:
    \begin{align}
    \rho_2 =&\ \frac{1}{2} 
    \bigg[
        |0\rangle\langle0| \otimes \Phi_{2}^{(YY)} \Phi_{1}^{(YY)} \left[\rho_S \right] \nonumber \\ 
        &\ + |0\rangle\langle1| \otimes \Phi_{2}^{(YX)} \Phi_{1}^{(YX)} \left[\rho_S \right] \nonumber\\
        &\ +|1\rangle\langle0| \otimes\Phi_{2}^{(XY)}\Phi_{1}^{(XY)} \left[\rho_S \right]\nonumber\\ 
        &\ + |1\rangle\langle1| \otimes \Phi_{2}^{(XX)}\Phi_{1}^{(XX)} \left[\rho_S \right] \bigg].
    \end{align}
Continuing in this way $K$ times, we obtain by induction, 
    \begin{align}
    \label{eq:rho-K-one-run}
        \rho_{K} 
        =&\ \frac{1}{2} 
        \bigg[
            |0\rangle\langle0| \otimes \bigcirc_{j=1}^{K} \Phi_{j}^{(YY)} \left[\rho_S \right] \nonumber\\ &\ + 
            |0\rangle\langle1| \otimes \bigcirc_{j=1}^{K} \Phi_{j}^{(YX)} \left[\rho_S \right] \nonumber \\&\ +  
            |1\rangle\langle0| \otimes \bigcirc_{j=1}^{K} \Phi_{j}^{(XY)} \left[\rho_S \right] \nonumber \\ &\ +
            |1\rangle\langle1| \otimes \bigcirc_{j=1}^{K} \Phi_{j}^{(XX)} \left[\rho_S \right]
        \bigg].
    \end{align}
Finally, we measure the ancilla and system registers with the observable $\sigma^x \otimes O$. This constitutes one run of Algorithm \ref{algo: collision model}.

Let us now look at the outcome of any such run. On measuring the ancilla in $\sigma^x$, the first and the last terms of Eq.~\eqref{eq:rho-K-one-run} disappear, and so the output of the $k$-${\mathrm{th}}$ run becomes,
    \begin{equation}
        \mu_k = \frac{1}{2} \Tr \Bigg[
            O~\bigg[\bigcirc_{j=1}^{K} \Phi_{j}^{(YX)} \big[\rho_S \big] + 
            \bigcirc_{j=1}^{K} \Phi_{j}^{(XY)} \big[\rho_S \big]
            \bigg]
        \Bigg].
    \end{equation}
Then, by the linearity of expectation, we have 
        \begin{equation}
        \mathbb{E}[\mu_k]=\frac{1}{\zeta^2}\Tr \left[ \bigcirc_{j=1}^{K} \widetilde{\Phi}_{j}  \left[ \rho_S \right] \right]= \frac{1}{\zeta^2} \Tr\left[O 
                \widetilde{\mathcal{M}}_K \left[ \rho_S \right]\right],
        \end{equation}
where $\zeta$ is as defined in Algorithm \ref{algo: collision model}. Thus, the outcome of each run is a random variable whose expectation value gives an estimate of the desired quantity (up to a multiplicative factor of $1/\zeta^2$). 

We observe that the positive operator-valued measurement of the state at the end yields some eigenvalue of $O$ in the range $[-\|O\|,\|O\|]$. So, outcome $\mu_k$ satisfies,
    \begin{align}
        -{\norm{O}\zeta^2 \leq \zeta^2\mu_k\leq \norm{O}\zeta^2}.
    \end{align}
Thus, after $T$ runs, we have a set of random variables $\{\mu_k\}_{k=1}^T$. By using Hoeffding's inequality, we can ensure 
$$
\mu=\dfrac{\zeta^2}{T}\sum_{k=1}^{T}\mu_k,
$$
is close to its expectation value. Indeed, \begin{align*}
        \Pr\Bigg[\bigg|\mu - \Tr\big[O~\widetilde{\mathcal{M}}_K[ \rho_S ] \big]\bigg| \geq \eps/2\Bigg] \leq 2\exp\left[-\dfrac{T\eps^2}{8 \zeta^4\norm{O}^2}\right].
        \end{align*}       
Thus,
        \begin{equation} \label{eq:error in expectation and imprecise map}
        \left|\mu - \Tr\left[O~\widetilde{\mathcal{M}}_K[\rho_S ] \right]\right| 
        \leq \eps/2,
        \end{equation}
with at least $1-\delta$ probability for 
\begin{equation*}
        \label {eq:number-of-repititions}
        T\geq \dfrac{8\norm{O}^2\ln(2/\delta)\zeta^4}{\eps^2}.
\end{equation*}
Now, for any $j\in [1,K]$, we get
$$
\norm{U_j-\widetilde{U}_j}\leq \eps'= \dfrac{\varepsilon}{6K\norm{O}}
$$
from the statement of Lemma~\ref{lemma: lcu of time evolution}. Then, using Lemma~\ref{lemma: approximate K-collision map distance bound} and the triangle inequality, we obtain
\begin{align} 
\label{eq:error in expectation and true map}
        &\hspace{-1pt}\bigg|\mu - \Tr\big[O~\mathcal{M}_K [\rho_S] \big] \bigg|\nonumber\\
        &\leq \bigg|\mu - \Tr\big[O~\widetilde{\mathcal{M}}_K[\rho_S ] \big]\bigg|%\nonumber \\
        %&
        + \bigg|\Tr\big[O~\widetilde{\mathcal{M}}_K[\rho_S ] \big]
        - \Tr\big[O~\mathcal{M}_K[\rho_S ] \big]\bigg|\nonumber\\
        &\leq \eps/2+\eps/2=\eps.
        \end{align}
To estimate the circuit depth of Algorithm \ref{algo: collision model} and the number of classical repetitions $T$, we need to find $\zeta$, which %. This 
crucially depends on the choices of the parameters $r_j$. %for implementing each Hamiltonian simulation procedure. 
For the $j$-${\mathrm{th}}$ collision, Lemma \ref{lemma: lcu of time evolution} shows how we can consider the LCU decomposition $\widetilde{U}_j$ approximating $U_j=e^{-i\Delta t\beta_j\overline{H}}$ with sufficient accuracy. If $\beta=\max_{j}\beta_j$ and $r=\min_j r_j$, then the sum of LCU coefficients $\alpha^{(j)} = \sum_k |\alpha_{jk}|$ satisfies,
        \begin{align}
        \alpha^{(j)} \leq&\ e^{(\beta_{j} \Delta t)^2 / r_j} \leq e^{\left(\beta \Delta t\right)^2 / r},\end{align}
with 
        \begin{align}     
          \zeta \leq&\ \left(e^{\left(\beta  \Delta t\right)^2 / r}\right)^K.  
        \end{align}
Thus we ensure that $\zeta = \bigO(1)$ by choosing $r=\bigO(\beta ^2 \Delta t^2 K)$. 
Consequently, the number of classical repetitions needed is
    \begin{equation}
        T =  \bigO\left( \dfrac{\norm{O}^2\log(1/\delta)}{\eps^2}\right).
    \end{equation}

Now, for the circuit depth of each coherent run, the quantum circuit in Fig.~\ref{fig: single ancilla collision model circuit} consists of $2K$ unitaries of the form $X^{(c)}_j$ and $Y^{(a)}_j$. Each of these unitaries comprises $qr$ Pauli operators and $r$ controlled single-qubit rotations, where $q$ is the truncation parameter of the Taylor series. Thus, the overall circuit depth will be $\tau_d = \bigO(K(qr+r)) = \bigO(Kqr)$. For the above choice of $r$, Eq.~\eqref{eq:truncation-taylor-series-collision-map} gives
    $$
    q = \bigO\left(\dfrac{\log(\beta  K \norm{O}\Delta t/\eps)}{\log\log(\beta  K \norm{O}\Delta t/\eps)}\right),
    $$
which gives the overall circuit depth per coherent run as
    \begin{equation}
        \tau_d = \bigO\left(\beta ^2K^2\Delta t^2 \dfrac{\log(\beta  K \norm{O}\Delta t/\eps)}{\log\log(\beta  K \norm{O}\Delta t/\eps)} + K\tau_{\rho_E}\right).
    \end{equation}
Here, the additive term $K \tau_{\rho_E}$ appears because, in each run of the circuit, the sub-environment needs to be prepared $K$ times, each time requiring a circuit depth of at most $\tau_{\rho_E}$. This completes the proof.
\end{proof}
%We now compare the complexity of implementing the $K$-collision map using other near-term Hamiltonian simulation procedures.

%%%%%%%%%%%%%%%%%%%%%%%%%%%%%%%%%%%%%%%%%%%%%%%%%%%%%%%%%%%%%%%       
\begin{table*}[htbp]
\caption{Comparison of the costs of estimating the expectation value of an observable $O$ with respect to a quantum state that has undergone the $K$-collision map of Definition~\ref{def: K-collision map} using different near-term Hamiltonian simulation algorithms. The goal of the algorithm is to output the desired expectation value within an additive accuracy of $\eps$ with a constant success probability. Here, the ancilla qubits indicate the number of additional qubits (other than the system and environment qubits) required. We assume that for any of the $K$ collisions, the total Hamiltonian is a linear combination of at most $L$ strings of $n$-qubit Pauli operators, with the total weight of the coefficients upper bounded by $\beta$. Each collision corresponds to evolving according to the corresponding (total) Hamiltonian for a time $\Delta t$. Also, $\tau_{\rho_E}=\max_{j\in [1,K]}\tau_{\rho_{E_j}}$, where $\tau_{\rho_{E_j}}$ is the circuit depth of the unitary preparing the sub-environment $E_j$ in the state $\rho_{E_j}$.}
\label{table: collision model complexity}
{\centering
%\renewcommand{\arraystretch}{3}
%\begin{tabularx}{\textwidth}{ | >{\hsize=0.55\hsize}X |  >{\hsize=0.6\hsize}X | >{\hsize=2\hsize}X | >{\hsize=0.4\hsize}X | }
%
\renewcommand\baselinestretch{3}\selectfont
\begin{tabular*}{0.9\textwidth}{l@{\extracolsep{\fill}} ccc}
%\begin{tabularx}{0.9\textwidth}{c c c c}
\hline
Algorithm & No. of ancilla qubits & Circuit depth per coherent run & Classical repetitions \\
\hline\hline
1st-order Trotter & 0 & $\bigO\left( \dfrac{\beta^2K^2 \|O\|L\Delta t^2}{\epsilon}+K\tau_{\rho_E} \right)$ & $\bigO\left(\dfrac{\|O\|^2}{\epsilon^2}\right)$ \\

qDRIFT & 0 & $\bigO\left( \dfrac{ \beta^2 K^2 \|O\|  \Delta t^2 }{\eps}+K\tau_{\rho_E}\right)$ & $\bigO\left(\dfrac{\|O\|^2}{\eps^2}\right)$\\

$2k$-order Trotter  & 0 & $\bigO\left( L(K\beta \Delta t)^{1+\frac{1}{2k}}\left(\dfrac{ \|O\|}{\eps}\right)^{\frac{1}{2k}}+K\tau_{\rho_E}\right)$ & $\bigO\left(\dfrac{\|O\|^2}{\eps^2}\right)$ \\

Single-Ancilla LCU & 1 & $\bigO\left(\beta ^2K^2\Delta t^2 \dfrac{\log(\beta  K \norm{\bigO}\Delta t/\eps)}{\log\log(\beta  K \norm{O}\Delta t/\eps)}+K\tau_{\rho_E}\right)$ & $\bigO\left(\dfrac{\|O\|^2}{\eps^2}\right)$ \\[1ex]
\hline
%\end{tabularx}
\end{tabular*}
}
\end{table*}
%%%%%%%%%%%%%%%%%%%%%%%%%%%%%%%%%%%%%%%%%%%%%%%%%%%%%%%%%%%%%%%

\subsection{Comparing the complexity of implementing a Markovian $K$-collision map using various near-term Hamiltonian procedures}
\label{subsec:complexity-comparison-k-collision}

We now compare the complexity of other near-term Hamiltonian simulation algorithms to output an $\varepsilon$-additive accurate estimate of the expectation value $\Tr[O\,\mathcal{M}_K[\rho_S]]$. Primarily, we will compare the circuit depth, the number of ancilla qubits, and the number of classical repetitions needed. At the onset of the early fault-tolerant era, it is better to have multiple independent runs of a short-depth quantum circuit than a single run of a very deep quantum circuit. Thus, it is standard to separately analyze the cost of each run and the number of classical repetitions separately. The total complexity is, of course, the product of the circuit depth per coherent run and the total number of classical repetitions. 

First, we observe that most Hamiltonian simulation algorithms can be incorporated into Algorithm \ref{algo: collision model}. Steps 2a and 2b, which are essentially implementing the operator $U_j=e^{-i\Delta t \beta_j\overline{H}_j}$ in the Hamiltonian simulation by SA-LCU, can be replaced with other near-term techniques such as qDRIFT or Trotterization. In Step 3, a direct measurement of $O$ on the system register would suffice for these two methods, as they do not require any ancilla registers. So, Lemma~\ref{lemma: approximate K-collision map distance bound} can also be modified to incorporate different procedures:  any Hamiltonian simulation procedure needs to be implemented with precision $\bigO(\varepsilon/(K\|O\|))$, in order to output an $\varepsilon$-accurate estimate of the expectation value of $O$. This circuit depth is essentially the cost of composing the underlying Hamiltonian simulation algorithm $K$ times.

On the other hand, the expectation value of $O$ can either be measured incoherently or coherently. The incoherent approach involves simply measuring $O$ with respect to the prepared state, requiring $\bigO(\|O\|^2/\varepsilon^2)$ classical repetitions. It is also possible to use quantum amplitude estimation~\cite{aaronson2020quantum, grinko2021iterative} to estimate this quantity in cost scaling as $1/\varepsilon$ coherently. However, this requires access to a block encoding~\cite{chakraborty2019power, low2019hamiltonian} of the observable $O$, adding to the number of ancilla qubits required. More precisely, given an $(\alpha_O, a_O, 0)$ block-encoding of $O$, we can coherently estimate the desired expectation value using amplitude estimation for $\bigO\left(\alpha_O(\tau_d + \tau_O)/\eps\right)$ cost. Here, $\tau_d$ is the circuit depth of the composition of $K$ Hamiltonian simulation algorithms, and $\tau_O$ is the circuit depth of implementing the block encoding of $O$. So, along with the number of ancilla qubits, the circuit depth also increases substantially. Thus, amplitude estimation is not a technique that can be deployed in early fault-tolerant quantum computers. Consequently, we restrict ourselves to estimating the cost using the incoherent approach. Table~\ref{table: collision model complexity} summarizes the complexities associated with different near-term Hamiltonian simulation methods.

Let us begin by considering the circuit depth of the First-order Trotter method~\cite{lloyd1996universal, childs2021theory}. We note that no ancilla qubit (other than the system and environment registers) is needed. The worst-case circuit depth of simulating the $j$-${\mathrm{th}}$ collision map scales with the number of terms in the corresponding total Hamiltonian ($L_j$) as $\bigO(KL\beta_j^2\Delta t ^2\|O\|/\varepsilon)$. Then, by composing $K$ such collision maps and using the upper bounds $\beta$ and $L$, we obtain the circuit depth for each coherent run as
$$
\tau_d=\bigO\left(\dfrac{K^2 L\beta^2\Delta t ^2\|O\|}{\varepsilon}+K\tau_{\rho_E}\right).
$$
In order to estimate the expectation value of observable $O$ with a success probability of at least $1-\delta$, the number of independent runs required is $\bigO(\|O\|^2\log(1/\delta)/\varepsilon^2)$.

The randomized Hamiltonian simulation approach, qDRIFT \cite{campbell2019random}, also requires no ancilla qubits. Moreover, the circuit depth does not depend on the number of terms in the Pauli decomposition of the underlying Hamiltonian. The overall circuit depth to simulate the Markovian $K$-collision map is given by
$$
\tau_d=\left(\dfrac{K^2 \beta^2\Delta t ^2\|O\|}{\varepsilon}+K\tau_{\rho_E}\right).
$$
Thus, compared to %Hamiltonian simulation by 
SA-LCU, both First-order Trotter and qDRIFT require an exponentially worse circuit depth in terms of $1/\varepsilon$.

We now move on to the complexity of implementing the Markovian $K$-collision map by using higher-order Trotter methods~\cite{childs2021theory}. For any positive number $k$, the circuit depth of the $2k$-order Trotter method for implementing $e^{-i\Delta t \overline{H}_j}$ to within an accuracy of $\bigO(\eps/K\|O\|)$ is 
 $$
 \bigO\left(5^{k-1}  L_j \left(\beta_j \Delta t\right)^{1+\frac{1}{2k}}\cdot \left(\dfrac{K\|O\|}{\eps}\right)^{\frac{1}{2k}}\right).
 $$
Then, composing the previously mentioned simulation procedure a total of $K$ times requires a circuit depth of 
\begin{align}
    \tau_d&=\bigO\left(5^{k-1}  K L (\beta \Delta t)^{1+\frac{1}{2k}}\cdot \left(\dfrac{K\|O\|}{\eps}\right)^{\frac{1}{2k}}+K\tau_{\rho_E}\right) \nonumber\\
    &=\bigO\left(5^{k-1}  L (K\beta \Delta t)^{1+\frac{1}{2k}}\cdot \left(\dfrac{\|O\|}{\eps}\right)^{\frac{1}{2k}}+K\tau_{\rho_E}\right).
\end{align}
Note that the pre-factor ($(\beta L)^{1+1/(2k)}$) in the complexity of Trotter-based methods scales with the norm of the sum of the nested commutators of the local Pauli terms in the description of the Hamiltonian. In certain cases, the pre-factor scaling is better than the worst-case bounds we consider here. We refer the readers to Ref.~\cite{childs2021theory} for more details. Although higher-order Trotter methods do not require any ancilla qubits, the exponential scaling in the pre-factor makes it difficult to implement these methods for high $k$ values. Typically, in practice, $k=1$ (the second-order method) and $k=2$ (the fourth-order method) are implemented.

Finally, the state-of-the-art Hamiltonian simulation method, qubitization, requires access to a block encoding of the underlying Hamiltonian \cite{low2019hamiltonian}. In the case of simulating each collision $U_j=e^{-i\Delta t \beta_j\overline{H}_j}$, a block encoding to $H_j$ is needed, which requires $\bigO(\log L_j)$ ancilla qubits. However, since we have to implement a composition of these individual collision maps, overall, we need $\bigO(\log L)$ ancilla qubits to implement each of these maps. Additionally, we would require some sophisticated controlled logic in executing this~\cite{pocrnic2023quantum}. Although the circuit depth
$$
\tau_d=\bigO\left(KL\beta \Delta t + K\log(K\|O\|/\eps)+K\tau_{\rho_E}\right),
$$
has a better dependence on $K$ and $\Delta t$, as compared to near-term Hamiltonian simulation methods, implementing a $K$-collision map using qubitization is beyond the reach of early fault-tolerant quantum computers.

Overall, we have developed a general framework to simulate $K$ memoryless collisions on a quantum computer using various near-term Hamiltonian simulation procedures. In the next section, we use such collisions to simulate Lindbladian dynamics. The framework introduced here can also be adapted to incorporate memory effects from interactions between the environment subsystems, leading to non-Markovian dynamics. We later define and simulate a non-Markovian $K$-collision map in Sec.~\ref{sec:Non-markovian dynamics}.   

\section{Simulating Lindbladian Dynamics using quantum collision models}\label{sec:Lindbladian Dynamics Simulation}
The Lindblad master equation describes the time evolution of a quantum system undergoing dissipative dynamics in the presence of an environment. It assumes that the underlying system is weakly coupled to the environment at all times so that the Born-Markov and secular approximations hold~\cite{breuer2002theory}. The Lindblad operator is, in fact, the generator of any quantum Markov semigroup~\cite{Lindblad1976}. For a system with Hamiltonian $H_S$, the Lindblad master equation, describing the reduced state of the system $\rho_S$, is given by  
\begin{equation}
    \mathcal{L}[\rho_S]\equiv \diffp{\rho_S}{t}= -i[H_S, \rho] + \sum_j \left( A_j \rho A_j^\dagger - \frac{1}{2} \{A_j^\dagger A_j, \rho\}\right).
    \label{eq: lindbladian map}
\end{equation}
Here, the evolution comprises two distinct parts: a unitary component governed by the system Hamiltonian $H_S$ and a dissipative component described by the so-called quantum jump operators $A_j$, obtained from the interaction between the system and the environment. Note that for a $d$-dimensional system, $A_j\in \mathbb{C}^{d\times d}$, are not necessarily Hermitian. Simulating the Lindblad dynamics for time $t$ on a quantum computer essentially means implementing the map $e^{\mathcal{L}t}$.

We can now analyze the complexity of simulating Lindblad dynamics using the Markovian $K$-collision map from Sec.~\ref{subsec:K-collision-approx}. As discussed in the previous section, quantum collision models simulate open system dynamics by discretizing the continuous system-environment interaction into a sequence of brief collisions between the system and independent sub-environments. The correspondence between Lindblad dynamics and the collision model has been derived earlier~\cite{ bruneau2014repeated, cattaneo2021collision, pocrnic2023quantum}. We will follow the constructions and error analysis of the recent work by Pocrnic et al.~\cite{pocrnic2023quantum}. The choices of the time of each collision ($\Delta t$) and the total number of collisions ($K$) are crucial for quantum collision models to approximate Lindbladian dynamics. It is only for the right choices that a Markovian $K$-collision map can approximate Lindblad dynamics, and (a slightly modified version of) Algorithm~\ref{algo: collision model} can be used to efficiently estimate $\Tr[O~e^{\mathcal{L}t}[\rho_S]]$ for any observable $O$. 

As in Sec.~\ref{sec:Modified Collision Model}, we consider an $n$-qubit system with Hamiltonian $H_S$, prepared initially in the quantum state $\rho_S$. The environment consists of $m$ discrete, single qubit sub-environments prepared in some state $\rho_{E_j}$, for $j\in [1,m]$. As before, these sub-environments sequentially interact with the system over small but equal time intervals, $\Delta t$, driving its evolution. The system evolves under its free Hamiltonian $H_S$, while the $j$-th sub-environment evolves under its local Hamiltonian $H_{E_j}$. The interaction Hamiltonian $H_{I_j}$ governs the interaction between the system and the $j$-th sub-environment.

While Lindbladian dynamics effectively couple the system to all $m$ environmental subsystems simultaneously, the collision model operates sequentially, with the system interacting with one environment at a time. To reconcile this difference, we renormalize the system Hamiltonian as $H_S \rightarrow \frac{1}{m}H_S$. Furthermore, the Lindblad dynamics is derived from collision maps in a regime where the system-environment coupling parameter $\lambda$ is tuned to satisfy $\lambda^2 \Delta t = 1$. As $\Delta t$ is typically small, intuitively, this results in repeated momentary collisions %of short duration 
between the system and a strongly coupled sub-environment. This ensures the coupling is strong enough to drive dissipative dynamics even within a short interaction time $\Delta t$. In the remainder of this section, we shall assume that the coupling constant $\lambda$ is diverging, i.e., \ $\lambda \rightarrow 1/\sqrt{\Delta t}$, and we refer the readers to Refs.~\cite{ciccarello2022collisionreview, cattaneo2021collision, pocrnic2023quantum} for detailed discussions. Another standard assumption, for the derivation of Lindblad dynamics from collision models, is that the state of each sub-environment or the interaction Hamiltonians are so chosen that $\forall j,~\Tr_{E_j}[[H_{I_j}, \rho_S \otimes \rho_{E_j}]] = 0$. This is equivalent to considering the overall state of the environment $E$ to be thermal. 
%\begin{widetext}

%%%%%%%%%%%%%%%%%%%%%%%%%%%%%%%%%%%%%%%%%%%%%%%%%%%%%%%%%%%%%%%
\begin{algorithm}[t]
\SetAlgoCaptionLayout{bottom}
\caption{Algorithm to estimate the expectation value of an observable $O$ with respect to the Lindblad map applied to a quantum state.}\label{algo: lindblad dynamics}
\KwIn{Initial system state in $\rho_S$, $m$ single qubit sub-environment states, each initialized in the state defined in Eq.~\eqref{eq:sub-environment-thermal}, observable $O$, unitaries $\widetilde{U}_1$, \dots. $\widetilde{U}_m$, the total number of repetitions $\nu$, and precision $\eps'$, where the LCU decomposition of each $\widetilde{U}_j=\sum_{k}\alpha_{jk} W_{jk}$, such that $\forall j\in [1, m]$, $\norm{\widetilde{U}_j-e^{-i\Delta t\beta_j\overline{H}_j}}\leq \eps'$.~\\~\\}
    \begin{itemize}
        \item[1.~] Initialize the system and the ancilla in the state $\rho_S$ and $\ket{+}$, respectively.
        \item[2.~] For iterations from $j = 0$ to $K-1$, where $K=m\nu$:
            \begin{itemize}
                \item[a.~] Set $\ell=j(\mathrm{mod~}m)+1$.
                \item[b.~] Initialize the environment register in state $\rho_{E_{\ell}}$.
                \item[c.~] Draw two i.i.d.~samples $X_j$ and $Y_j$ from the ensemble 
                    \begin{equation*}
                        \mathcal{D}_{\ell}=\left\{ W_{\ell k}, \frac{\alpha_{\ell k}}{\alpha^{(\ell)}}\right\},
                    \end{equation*} 
                    where $\alpha^{(\ell)} = \sum_k|\alpha_{\ell k}|$
                \item[d.~] Apply the controlled unitary $X_j^{(c)}$ and the anti-controlled unitary $Y_j^{(a)}$ to the system. 
                \item[e.~] Perform a partial trace over the environment register. 
            \end{itemize}
      \item [3.~]  Measure the joint ancilla and system state on the observable $(\sigma^x \otimes O)$ and record the measurement outcome as $\mu_i$.
      \item [4.~] Repeat Steps 1 to 3 a total of $T$ times.  
      \item [5.~] Compute the final estimate $\mu$ as:
    \begin{equation*}
        \mu = \dfrac{\zeta^2}{T}\sum_{j=1}^{T} \mu_j,
    \end{equation*}
      where $\zeta = \prod_{j=1}^K \alpha^{(j)}$.
      \end{itemize}
  \KwOut{Estimated expectation value $\mu$}  
\end{algorithm} 
%%%%%%%%%%%%%%%%%%%%%%%%%%%%%%%%%%%%%%%%%%%%%%%%%%%%%%%%%%%%%%%

%\end{widetext}

Let us now describe the structure of the Hamiltonian we consider that satisfies the above constraints. We assume that the system Hamiltonian $H_S$ (rescaled by $m$) is a linear combination of strings of Pauli operators given by
$$
H_S=\sum_{j=1}^{L_S} \lambda_j P_j/m,
$$
such that $\beta_S=\sum_{j}|\lambda_j|$. In our case, the environment is a discrete sum of $m$ sub-environments, such that $H_{E_j}=\beta_{E_j} \sigma^z$ (equivalent to the number operator up to an energy shift). Each sub-environment is prepared in the single qubit thermal state (at some inverse temperature $\omega$), i.e.,
\begin{equation}
\label{eq:sub-environment-thermal}
\rho_{E_j}=\dfrac{\ket{0}\bra{0}+e^{-\omega}\ket{1}\bra{1}}{1+e^{-\omega}}.
\end{equation}
This state can be prepared efficiently by first preparing the entangled pure state
$$
\ket{\psi}=\dfrac{\ket{00}+e^{-\omega/2}\ket{11}}{\sqrt{1+e^{-\omega}}},
$$
and then tracing out the second qubit. Henceforth, we will assume that preparing $\rho_{E_j}$ is a constant-depth unitary procedure.

%%%%%%%%%%%%%%%%%%%%%%%%%%%%%%%%%%%%%%%%%%%%%%%%%%%%%%%%%%%%%%%
\begin{figure*}[ht]
    \centering
    \scalebox{0.8}{
    \includegraphics[width=0.9\textwidth]{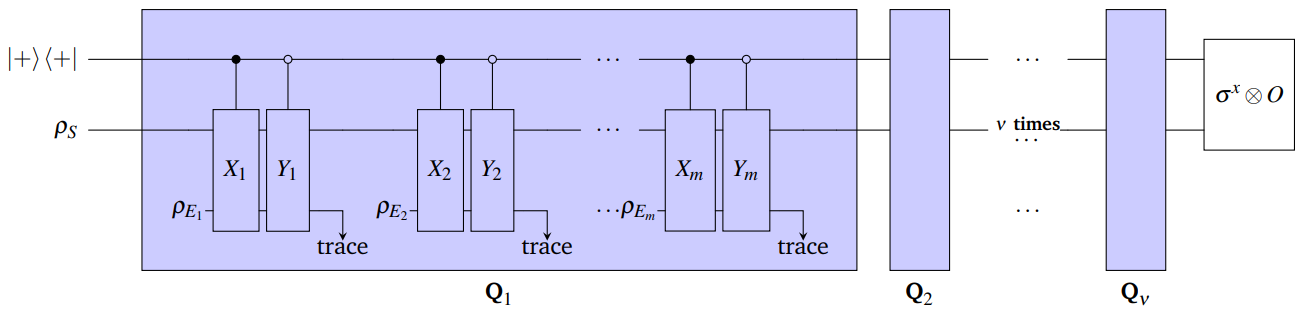}
    }
    \caption{The circuit to estimate the expectation value of an observable $O$ for a system evolved under Lindbladian dynamics. The ancilla qubit and the system is initialized in $|+\rangle \langle +|$ and $\rho_S$ respectively. In each block, the unitaries $X_j $ and $Y_j$ are independently sampled and applied as controlled and anti-controlled operations conditioned on the ancilla qubit. After each interaction, the corresponding environment sub-system $\rho_{E_j}$ is traced out, enforcing the Markovian condition. This process is repeated cyclically over $m$ environments for $\nu$ iterations. Finally, the observable $\sigma^x \otimes O$ is measured to estimate the time-evolved expectation value of $O$.
}
    \label{fig: Lindbladian circuit}
\end{figure*}
%%%%%%%%%%%%%%%%%%%%%%%%%%%%%%%%%%%%%%%%%%%%%%%%%%%%%%%%%%%%%%%

We consider that the interaction Hamiltonian corresponding to the $j$-${\mathrm{th}}$ collision (between the system and $j$-${\mathrm{th}}$ sub-environment)  $H_{I_j}$ can also be expressed as a linear combination of some $L_{I_j}$ Pauli operators. For instance, we can express it in terms of the Lindbladian jump operators as 
\begin{equation*}
H_{I_j} = \left(A_j \otimes \sigma^{+} + A_j^\dag \otimes \sigma^{-}\right),    
\end{equation*}
where $\sigma^{\pm}=(\sigma^x\pm i\sigma^y)/2$ are the raising and lowering operators respectively. We assume $\beta_{I_j}$ denotes the total weight of the coefficients in the description of $H_{I_j}$ and that this Hamiltonian has $L_{I_j}$ terms. Note that in some cases, the individual jump operator $A_j$ may itself be unitary~\cite{pocrnic2023quantum}. Overall, the total Hamiltonian corresponding to the $j$-${\mathrm{th}}$ collision is given by:
\begin{equation}
    \label{eq: collision hamiltonian, Lindbladian}
    \overline{H}_j =  \dfrac{1}{\beta_j}\left(\dfrac{1}{m} H_S +  H_{E_j} + \lambda H_{I_j}\right),
\end{equation}
where $\beta_j\leq \beta_S/m+\lambda\beta_{I_j}+\beta{E_j}$, and $\lambda=1/\sqrt{\Delta t}$. Hence, $\overline{H}_j$ also can be expressed as a linear combination of $L_S+L_{I_j}+1$ Pauli operators with total weight at most $\beta_j$.

Let us now discuss the choice of $\Delta t$ and $K$ for a $K$-collision map to be close (in induced $1$-norm) to $e^{\mathcal{L}t}[.]$. The collisions between the system and the $m$ sub-environments occur one by one in a fixed order ($E_1$, $E_2$, $\ldots$, $E_{m}$). Now, to simulate the Lindblad dynamics for a total evolution time of $t$, this sequence is repeated $\nu$ times, with each collision occurring for a time interval $\Delta t = t/\nu$. Thus, in all, there are $K=m\times \nu$ collisions, making it an $(m,\nu)$-collision map. 

As before, the interaction of $j$-${\mathrm{th}}$ sub-environment with the system is given by
$$
    U_j = e^{-i\beta_j \overline{H}_j \Delta t},
$$
and, from Definition \ref{def: collision map}, the $j$-th collision map is
\begin{equation}
    \Phi_j[.] = \text{Tr}_{E_j} \left[ U_j \left(. \otimes \rho_{E_j}\right) U_j^{\dagger} \right].
\end{equation}
The $(m,\nu)$-collision map is defined as follows:

\begin{restatable}[$(m,\nu)$-collision map]{definition}{}
\label{def: m, nu-collision map}
Let $\Phi_1$ to $\Phi_{m}$ be the collision maps as defined in Definition \ref{def: collision map}. Then a $(m, \nu)$-collision map, $\mathcal{M}_{m, \nu}$, is defined as the application of these maps composed $\nu$ times as follows:
    \begin{align}
        \mathcal{M}_{m, \nu}[.] \equiv  \left( \bigcirc_{j=1}^{m} \Phi_j\left[.\right] \right)^{\circ \nu}.        
    \end{align}    
\end{restatable}
\noindent
The value of $\nu$, i.e., the number of times the sequence of $m$ collisions should be repeated so that a $(m,\nu)$-collision map approximates $e^{\mathcal{L}t}$ up to an additive accuracy $\varepsilon$ was found in Ref.~\cite{pocrnic2023quantum}. For this, let us define
\begin{equation}
\label{eq:Gamma}
\Gamma =  \dfrac{\|\mathcal{L}\|^2_{1\to 1}}{m}+\left[\max_{\ell\in [1,m]}\left(\beta_S, \beta_{I_{\ell}}, \beta_{E_{\ell}}\right)\right]^4, 
\end{equation}
where $\|\mathcal{L}\|_{1\to 1}$ is the induced 1-norm of $\mathcal{L}[.]$. Now, we restate the result of \cite{pocrnic2023quantum} here:
\begin{lemma}[Corollary 2.1 of \cite{pocrnic2023quantum}]
\label{lemma: accuracy of collision model (Wiebe)}
For $\eps\in (0,1)$, a $(m, \nu)$-collision map as defined in Definition \ref{def: m, nu-collision map}, with interaction time $\Delta t = t/\nu$,
\begin{align}
    \nu \geq O \bigg( \frac{t^2 m}{\eps} \Gamma\bigg),
\end{align}
and $\lambda \rightarrow \frac{1}{\sqrt{\Delta t}}$ satisfies  
\begin{equation}    
    \norm{e^{\mathcal{L}t} - \mathcal{M}_{m, \nu}[.]}_{1 \to 1} \leq \eps.
\end{equation}
\end{lemma}
\noindent
This demonstrates that the $(m, \nu)$-collision map $\mathcal{M}_{m, \nu}[.]$ provides an accurate approximation of Lindbladian dynamics, provided $\nu$ is as stated in Lemma \ref{lemma: accuracy of collision model (Wiebe)}.  We use Hamiltonian simulation by SA-LCU and Algorithm \ref{algo: lindblad dynamics} to achieve this.

We intend to implement the circuit shown in Fig.~\ref{fig: Lindbladian circuit}. Each sequence of unitaries $Q_i$ corresponds to $m$ collisions between the system and each single-qubit sub-environment $E_1$ through $E_m$. As this is repeated $\nu$ times, there are overall $K=m\nu$ collisions. Since each block $Q_i$ corresponds to collisions with the same set of sub-environments, we do not have $K$ distinct collisions with distinct sub-environments, but rather $\nu$ blocks of $m$-collisions between the system and the $m$ sub-environments. So, to implement the circuit in Fig.~\ref{fig: Lindbladian circuit}, we slightly modify Algorithm~\ref{algo: collision model}. 

We rewrite the $(m,\nu)$-collision map as
\begin{equation}
\label{eq:m,nu map rewritten as K-collision map}
    \mathcal{M}_{m, \nu}[.] =  \bigcirc_{j=0}^{K-1} \Phi_{j(\mathrm{mod~}m)+1} \left[.\right], 
\end{equation}
where $K=m\nu$. The right-hand side of Eq.~\eqref{eq:m,nu map rewritten as K-collision map} is simply a $\Phi_j[.]$ map, composed $K$ times, where the cyclic order of the $\nu$ repetitions is respected. Thus, the problem of simulating Lindbladian dynamics implies implementing a specific $K$-collision map. Just as in Sec.~\ref{subsec:K-collision-approx}, this map can be implemented using a number of near-term Hamiltonian simulation algorithms. In our case, this change is reflected in Step 2a. of Algorithm~\ref{algo: lindblad dynamics}. For the $j$-${\mathrm{th}}$ iteration in Step 2 of Algorithm \ref{algo: collision model}, the $\Phi_{\ell}[.]$ map is implemented, where $\ell=j(\mathrm{mod~}m)+1$, i.e.\ it implements a collision between the system and the single qubit environment $E_{\ell}$. Overall, the correctness of Algorithm~\ref{algo: lindblad dynamics} is similar to Theorem \ref{thm: Algorithm 1 proof}. Formally, we state the results via the following theorem: 
\begin{restatable}[]{theorem}{}
\label{thm: collision map to lindbladian}
Let us consider an observable $O$, an $n$-qubit system prepared in the initial state $\rho_S$, and $m$ single-qubit sub-environments with each initialized in the single-qubit thermal state defined in Eq.~\eqref{eq:sub-environment-thermal}. Let $\eps, \delta \in (0,1)$ and $K=m\nu$, where
\begin{equation}
    \nu = \bigO \left( \frac{t^2 m\norm{O}\Gamma}{\eps}  \right).
    \label{eq: nu value in theorem}
\end{equation}
Then, for $\eps'=\eps/(12K\|O\|)$, Algorithm~\ref{algo: lindblad dynamics} outputs an estimate $\mu$ with a probability of at least $1-\delta$, such that:
\begin{equation}
    \left| \mu - \Tr\left[ O e^{\mathcal{L}t}[\rho_S] \right] \right| \leq \eps,
\end{equation} 
using $T$ runs of the circuit shown in Figure~\ref{fig: Lindbladian circuit},  
where 
\begin{equation}
\label{eq: number of samples}
       T = \bigO\left( \dfrac{\norm{O}^2\log(1/\delta)}{\eps^2}\right).
\end{equation}
Moreover, the circuit depth of each run is
\begin{equation}
\label{eq:circuit-depth-lindblad-single-ancilla}
 \tau_d=\widetilde{\bigO}\left(\dfrac{m^3 t^3 \norm{O}}{\varepsilon}\Gamma \beta^2_{\max}\right),   
\end{equation}
where 
\begin{equation}
\label{eq:beta-max}
    \beta_{\max}=\max_{\ell\in [1, m]} \left(\beta_{I_{\ell}},~ \sqrt{\frac{t}{m^2\nu}}\beta_S,~ \beta_{E_\ell}\sqrt{\frac{t}{\nu}}\right).
\end{equation}
\end{restatable}
\begin{proof}
First, from the proof of Theorem \ref{thm: Algorithm 1 proof}, we know that for precision $\eps'=\eps/(12\|O\|K)$ and $\nu$ as chosen in the statement of the Theorem, Algorithm~\ref{algo: lindblad dynamics} outputs $\mu$ such that,
$$
\left| \mu - \Tr\left[ O~ \mathcal{M}_{m, \nu}[\rho_S] \right] \right| \leq \frac{\eps}{2},
$$   
with probability at least $(1-\delta)$. This requires 
$$
T=\bigO\left(\dfrac{\|O\|^2}{\eps^2}\log(1/\delta)\right),
$$
repetitions of the circuit in Fig.~\ref{fig: Lindbladian circuit}. Furthermore, from Lemma~\ref{lemma: accuracy of collision model (Wiebe)}, we know that our choice of $\nu$ ensures 
$$
        \left\| e^{t \mathcal{L}} - \mathcal{M}_{m, \nu}[.] \right\|_{1 \to 1} \leq \dfrac{\eps}{2\norm{O}}.
$$
For any valid initial state of the system, $\rho_S$, we have from the definition of induced-1 norm:
    \begin{equation}
        \left\| e^{t \mathcal{L}}[\rho_S] - \mathcal{M}_{m, \nu}[\rho_S] \right\|_1 \leq \dfrac{\eps}{2\norm{O}}.
    \end{equation}   
Using the tracial version of H\"{o}lder's inequality  (Lemma \ref{thm:holder}) we obtain
    \begin{equation}\label{eq: distance lindbladian-collision map}
        \left| \Tr\left[ O e^{t \mathcal{L}}[\rho_S] \right] - \Tr\left[ O \mathcal{M}_{m, \nu}[\rho_S] \right] \right| \leq \dfrac{\eps}{2}.
    \end{equation}
Then, the triangle inequality gives us
\begin{align}
        \bigg| \mu - &\Tr\left[ O e^{t \mathcal{L}}[\rho_S] \right] \bigg| 
        \nonumber \\&\leq \left| \mu - \Tr\left[ O \mathcal{M}_{m, \nu}[\rho_S] \right] \right|\nonumber \\
        &\quad + \left| \Tr\left[ O \mathcal{M}_{m, \nu}[\rho_S] \right] - \Tr\left[ O e^{t \mathcal{L}}[\rho_S] \right] \right| \nonumber \\
        &\leq \frac{\eps}{2} + \frac{\eps}{2} = \eps.
    \end{align}

%%%%%%%%%%%%%%%%%%%%%%%%%%%%%%%%%%%%%%%%%%%%%%%%%%%%%%%%%%%%%%%
\begin{table*}[htbp]
\caption{
   Comparison of the complexities for simulating Lindblad dynamics via the quantum collision model using different near-term Hamiltonian simulation procedures. We consider an $n$-qubit system, prepared in $\rho_S$, with Hamiltonian $H_S$, expressed as a linear combination of $L_S$ strings of Pauli operators, with total weight $\beta_j$. The environment is a discrete sum of $m$ single-qubit sub-environments, each prepared in the (single-qubit) thermal state. The $j$-${\mathrm{th}}$ collision corresponds to the interaction Hamiltonian $H_{I_j}$, which is also a linear combination of strings of Pauli operators of $L_{I_j}$ terms with total weight $\beta_{I_j}$. We implement $m$ collisions between the system and each sub-environment qubit, one by one, such that each block of $m$ collisions is repeated a total of $\nu$ times. For any observable $O$, if $\nu=\bigO(t^2\|O\|m\Gamma/\eps)$, our procedures output an estimate that is an $\eps$-additive accurate estimate of $\Tr[O~e^{\mathcal{L}t}[\rho_S]]$. Here, $L, \beta_{\max}$ and $\Gamma$ are defined in Eq.~\eqref{eq:total-terms-L}, Eq.~\eqref{eq:beta-max}, and Eq.~\eqref{eq:Gamma}, respectively.}
\label{table: lindbladian table}
\centering
\renewcommand\baselinestretch{3}\selectfont
\begin{tabular*}{0.9\textwidth}{l@{\extracolsep{\fill}} ccc}
% \renewcommand{\arraystretch}{3}
% %\begin{tabularx}{\textwidth}{ | >{\hsize=0.55\hsize}X |  >{\hsize=0.6\hsize}X | >{\hsize=2\hsize}X | >{\hsize=0.4\hsize}X | }
% \begin{tabularx}{0.9\textwidth}{c c c c}
\hline
Algorithm & Total no. of. qubits & Circuit depth per coherent run & Classical repetitions \\
\hline\hline
$1$st-order Trotter & $n+1$ & $\bigO\left( \dfrac{ L m^3t^3\|O\|^2}{\eps^2} \Gamma\beta_{\max}^2\right)$ & $\bigO\left(\dfrac{\|O\|^2}{\epsilon^2}\right)$ \\

qDRIFT & $n+1$ & $\bigO\left( \dfrac{ m^3 t^3 \|O\|^2}{\eps^2} \Gamma\beta^2_{\max}  \right)$ & $\bigO\left(\dfrac{\|O\|^2}{\eps^2}\right)$\\

$2$nd-order Trotter  & $n+1$ & $\bigO\left(L(mt)^{9/4}\left(\dfrac{\|O\|}{\eps}\right)^{5/4} \Gamma\beta^{3/2}_{\max}\right)$ & $\bigO\left(\dfrac{\|O\|^2}{\eps^2}\right)$ \\

Single-Ancilla LCU & $n+2$ & $\widetilde{\bigO}\left(\dfrac{m^3 t^3 \norm{O}}{\varepsilon}\Gamma \beta^2_{\max}\right)$ & $\bigO\left(\dfrac{\|O\|^2}{\eps^2}\right)$ \\

$2k$-order Trotter $[k>2]$  & $n+1$ & $\widetilde{\bigO}\left(\dfrac{Lm^2t^2\|O\|}{\eps}\Gamma\beta_{\max}\right)$ & $\bigO\left(\dfrac{\|O\|^2}{\eps^2}\right)$ \\[1ex]
\hline
\end{tabular*}
%\end{tabularx}
\end{table*}
%%%%%%%%%%%%%%%%%%%%%%%%%%%%%%%%%%%%%%%%%%%%%%%%%%%%%%%%%%%%%%%
%Now we move on to analysing the circuit depth of each run of Algorithm \ref{algo: lindblad dynamics}. Note that 
Now, Theorem \ref{thm: Algorithm 1 proof} gives the circuit depth of a $K$-collision map as
$$
\tau_d=\bigO\left(\beta^2K^2\Delta t^2 \frac{\log(\beta K\|O\|\Delta t/\eps)}{\log\log(\beta K\|O\|\Delta t/\eps)}+K\tau_{\rho_E}\right).
$$
In our case, $K=m\nu$, $\Delta t =t/\nu$, and $\tau_{\rho_E}=\bigO(1)$. Moreover, $\beta$ depends on $t$ and $\eps$ as
\begin{align}
\beta =&\ \max_{\ell\in [1,m]} (\beta_j)
= \max_{\ell} \left(\sqrt{\frac{\nu}{t}}\beta_{I_{\ell}} +\frac{1}{m} \beta_S + \beta_{E_{\ell}} \right)\nonumber\\
\leq&\ \sqrt{\frac{\nu}{t}}\times \bigO\left(\beta_{\max}\right),
\label{eq: beta in lindbladian}
\end{align}
where, in the last line, we have used the fact that $\beta_{E_{\ell}}=1$ for any $\ell\in [1,m]$. Here, $\beta_{\max}$ is as defined in the statement of this Theorem. Substituting these parameters, we obtain
\begin{align*}
    \tau_d=&\ \bigO\left(\beta^2 m^2 t^2 \frac{\log(\beta m t \|O\|/\eps)}{\log\log(\beta m t\|O\|/\eps)}+m\nu\right)\\
    =&\ \bigO\left(\nu m^2 t \beta^2_{\max} \frac{\log(\beta_{\max} m \sqrt{\nu t} \|O\|/\eps)}{\log\log(\beta_{\max} m \sqrt{\nu t}\|O\|/\eps)}+m\nu\right).    
\end{align*}
Finally, substituting $\nu=\bigO(t^2 m\|O\|\Gamma/\eps)$, we obtain
\begin{align}
    \tau_d=&\ \bigO\left( \frac{m^3 t^3 \beta^2_{\max}\|O\|\Gamma}{\eps} \frac{\log(\beta_{\max} m t \Gamma \|O\|/\eps)}{\log\log(\beta_{\max} m t \Gamma \|O\|/\eps)}\right)\\
    =&\ \widetilde{\bigO}\left(\frac{m^3 t^3 \|O\|}{\eps}\beta^2_{\max} \Gamma\right).
\end{align}
This completes the proof.
\end{proof}

There are two primary sources of error in simulating Lindblad dynamics using quantum collision models: the first arises from approximating Lindbladian dynamics by collision models, and the second stems from the simulation of individual collision steps, which depends on the precision of the Hamiltonian simulation technique employed. While the latter can be mitigated by choosing Hamiltonian simulation algorithms with optimal precision dependence, the error coming from the inherent gap between the Lindblad map and the $(m,\nu)$-collision map remains unaffected by the choice of Hamiltonian simulation. Indeed, Lindbladian dynamics can be approximated only if the system strongly couples with the sub-environments with a strength $\lambda=1/\sqrt{\Delta t}=\sqrt{\nu/t}\propto t/\sqrt{\eps}$ that grows stronger with the time we intend to simulate the dynamics. The norm of the $j$-${\mathrm{th}}$ collision Hamiltonian is at most $\beta_j$, which also increases monotonically with $t$, affecting the circuit depth of all methods that simulate Lindblad dynamics using quantum collision models~\cite{cattaneo2021collision}.  
 
In contrast, the state-of-the-art methods (i.e., direct approaches) for simulating Lindbladian dynamics require a cost $\bigO(t~\polylog{t/\varepsilon})$~\cite{childs2017efficient, cleve2017efficient, li2023simulating, ding2024simulating}. However, most of these methods require access to block encodings and use complicated, infeasible controlled operations for early fault-tolerant quantum computers. On the other hand, quantum collision models provide an easy-to-implement approach, not just for Lindbladian maps but also for other open systems dynamics. Moreover, as mentioned in Sec.~\ref{subsec:complexity-comparison-k-collision}, Algorithm \ref{algo: lindblad dynamics} also provides a unified framework to compare the cost of implementing the $(m,\nu)$-collision map using different near-term Hamiltonian simulation techniques. Consequently, in the next section, we compare the complexity of Algorithm \ref{algo: lindblad dynamics} when other near-term Hamiltonian techniques are used to output an $\eps$-additive estimate of $\Tr[O~e^{\mathcal{L}t}[\rho_S]]$.

\subsection{Comparison with other near-term Hamiltonian simulation algorithms}

We will borrow the circuit depths obtained in Sec.~\ref{subsec:complexity-comparison-k-collision} (Table \ref{table: collision model complexity}) for estimating $\Tr[O~e^{\mathcal{L}t}[\rho_S]]$ to $\varepsilon$-additive accuracy. In this case, the parameters are $K=m\nu$, $\Delta t=t/\nu$, and $\beta=\bigO(\beta_{\max}\sqrt{\nu/t})$,
where $\nu=\bigO(t^2\|O\|m\Gamma/\eps)$. From Theorem \ref{thm: collision map to lindbladian}, we know that any Hamiltonian simulation procedure needs to be implemented with precision $\eps'=\bigO(\varepsilon/K\|O\|)$. The circuit depth per coherent run, the total number of qubits needed, and the number of classical repetitions required are outlined in Table \ref{table: lindbladian table}. 

Let us first analyze the circuit depth for the first-order Trotter approach. In the worst case, it would depend on the maximum number of terms in $\overline{H}_{\ell}=H_S+H_{I_{\ell}}+H_{E_{\ell}}$ corresponding to the collisions. Let 
\begin{equation}
\label{eq:total-terms-L}
    L=L_S+\max_{\ell\in [1,m]} L_{I_{\ell}}+1.
\end{equation}  
Then, for the first-order Trotter method, the appropriate substitution of the parameters yields
\begin{equation}
    \tau_d=\bigO\left(\dfrac{L m^3 t^3\|O\|^2}{\eps^2}\Gamma\beta^2_{\max}\right),
\end{equation}
which indicates that the circuit depth is worse than the circuit depth of Hamiltonian simulation by SA-LCU [Eq.~\eqref{eq:circuit-depth-lindblad-single-ancilla}]. This method is, however, qubit-efficient, requiring $n+1$ qubits overall.

The circuit depth of any procedure using qDRIFT to estimate the desired expectation value is given by
\begin{equation}
    \tau_d=\bigO\left(\dfrac{ m^3 t^3\|O\|^2}{\eps^2}\Gamma\beta^2_{\max}\right),
\end{equation}
wherein the advantage over first-order Trotter is in the absence of any dependence on $L$. However, this circuit depth is also worse than Eq.~\eqref{eq:circuit-depth-lindblad-single-ancilla}. The qDRIFT approach also requires $n+1$ qubits overall, which is one less than Hamiltonian simulation by SA-LCU.

%Now let us move on to higher-order Trotter methods. 
For any $2k$-order Trotter method, we also incorporate the additive cost coming from the repeated preparation of the sub-environment register in the single qubit thermal state a total of $K=m\nu$ times (each such state can be prepared in $\bigO(1)$ circuit depth). Overall, we have
\begin{align}
\label{eq:2k-order-trotter-lindblad}
    \hspace{-0.15cm}\tau_d= \bigO\Big(L(mt)^{\frac{3}{2}+\frac{3}{4k}}\left(\dfrac{\|O\|}{\eps}\right)^{\frac{1}{2}+\frac{3}{4k}}\left(\Gamma\beta^2_{\max}\right)^{\frac{1}{2}+\frac{1}{4k}}+m\nu\Big).\hspace{-0.15cm}
\end{align}
However, as mentioned previously, only low-order Trotter methods are preferred for near-term implementation. In particular, for the second-order Trotter method ($k=1$), this becomes
\begin{equation}  \tau_d=\bigO\Big(L(mt)^{9/4}\left(\dfrac{\|O\|}{\eps}\right)^{5/4} \Gamma\beta^{3/2}_{\max}\Big).
\end{equation}
Compared to the circuit depth obtained by Hamiltonian simulation by SA-LCU, the second-order Trotter method has a better dependence on $m,~t$ and $\beta_{\max}$, and a worse dependence on $\norm{O}$ and $1/\eps$, in addition to scaling with $L$. Thus, the circuit depth in Eq.~\eqref{eq:circuit-depth-lindblad-single-ancilla} is shorter in settings where $L\ll \beta_{\max}$, and a high precision of the desired expectation value is demanded. 

In summary, SA-LCU achieves significantly shorter circuit depths than second-order Trotterization for high-precision simulations over short time scales. More precisely, the ratio between the circuit depths per coherent run of second-order Trotter and SA-LCU scales as $O(\eps^{1/4}/t^{3/4})$ (ignoring the dependence on all other parameters). Therefore, for simulating Lindblad dynamics over very long time durations where $\eps^{1/4}/t^{3/4}\ll 1$, second-order Trotterization can offer shorter circuit depths. 

For higher orders of this method ($k>2$), the additive term $m\nu$ starts to dominate, and in such cases, the asymptotic circuit depth is
$$
\tau_d=\widetilde{\bigO}\left(\dfrac{Lm^2t^2\|O\|}{\eps}\Gamma\beta_{\max}\right).
$$
Thus, even at very high Trotter orders, the dependence on $1/\eps$, $\|O\|$, and $\Gamma$ can be no better than Eq.~\eqref{eq:circuit-depth-lindblad-single-ancilla}. However, the dependence on $m, ~t$ and $\beta_{\max}$ is quadratically better. So, the circuit depth in Eq.~\eqref{eq:circuit-depth-lindblad-single-ancilla} is shorter when $L\gg m\beta_{\max} t$. This happens when we wish to simulate Lindblad dynamics for short $t$, and moreover, the maximum number of terms in the underlying Hamiltonians $\overline{H}_j$ ($L$) is substantially large \cite{campbell2019random}. As mentioned before, we have listed the worst-case complexity for Trotterization. It is possible that for specific Hamiltonians, the scaling of the prefactor is better than the worst-case \cite{childs2021theory}.

Finally, qubitization requires $\bigO(\log L)$ ancilla qubits, coherent access to a block encoding of the underlying Hamiltonians $\bar{H_j}$, and sophisticated controlled operations. The circuit depth is given by 
\begin{align}
     \tau_d&=\bigO\left(L\beta mt +m\nu\log\left(m\nu\|O\|/\eps\right)\right)\\
     &=\widetilde{\bigO}\left(\dfrac{Lm^2 t^2\|O\|}{\eps}\Gamma\beta_{\max}\right).
\end{align}
Thus, scaling of the circuit depth is similar to a very high order Trotter (up to logarithmic factors).

From the above discussion, it is clear that any procedure would at least require a circuit depth of $m\nu=\bigO(mt^2\|O\|\Gamma/\eps)$, simply because the sub-environments are prepared a total of $m\nu$ times. This can be seen as a lower bound for the circuit depth of estimating $\Tr[O~e^{\mathcal{L}t}[\rho_S]]$ using incoherent measurements of $O$, and matches with the lower bound of Ref.~\cite{cleve2017efficient}. In the next section, we apply these methods to a concrete problem.

Overall, our methods provide qubit-efficient, end-to-end quantum algorithms for simulating Lindbladian dynamics via the quantum collision model. It is, however, important to distinguish them from direct approaches such as Refs.~\cite{cleve2017efficient, li2023simulating, ding2024simulating, borras2025simulatelindblad}. These methods assume access to specialized oracles such as block encodings \cite{chakraborty2019power, low2019hamiltonian}, i.e.\ unitaries that embed the system Hamiltonian (say $U_{H_S}$) and each of the Lindblad jump operators (say $U_{A_j}$), in their top-left block. The complexity is expressed in terms of the number of queries made to the oracles $U_{H}$ and $U_{A_j}$, with the query complexity scaling as $O(t\cdot \mathrm{polylog}(t/\varepsilon))$ (ignoring dependence on other parameters). The actual circuit depth and gate counts depend on the detailed structure of $H_S$ and $A_j$, making a direct comparison with our end-to-end methods infeasible. Moreover, constructing such block-encodings often requires substantial overhead (in terms of ancilla qubits, multi-qubit controlled operations), rendering these methods impractical for near-term quantum devices.

%%%%%%%%%%%%%%%%%%%%%%%%%%%%%%%%%%%%%%%%%%%%%%%%%%%%%%%%%%%%%%%
\begin{figure*}[htp!]
\captionsetup[subfigure]{labelformat=empty}
\subfloat[\quad\quad\quad\quad(a)]{\includegraphics[width=0.9\columnwidth]{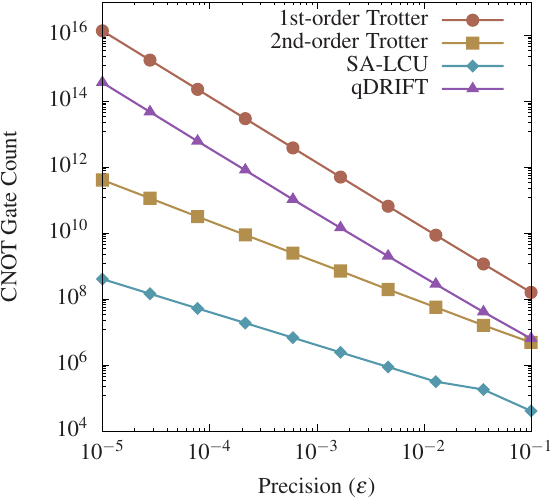}\label{fig:lindbladian_a}}\hspace{1cm}
\subfloat[\quad\quad\quad\quad(b)]{\includegraphics[width=0.9\columnwidth]{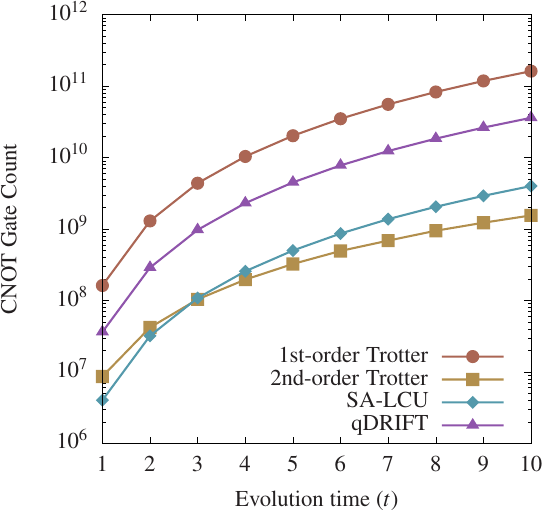}\label{fig:lindbladian_b}}
\caption{We consider the problem of estimating the average transverse-field magnetization of a $10$-qubit Heisenberg XXX model under amplitude damping. The corresponding Lindbladian dynamics can be approximated by quantum collision models. The randomized quantum algorithms we develop for simulating quantum collision models can be used to estimate the desired expectation value. In these plots, we show the CNOT gate count per coherent run of our algorithm for different near-term Hamiltonian simulation procedures: the First-order Trotter method (brown circles), the Second-order Trotter method (olive squares), Hamiltonian simulation by Single-Ancilla LCU (SA-LCU, blue diamonds), and qDRIFT (purple triangles). In (a), we vary the precision ($\eps$) for a fixed evolution time ($t=1$) of the underlying Lindbladian. The Hamiltonian simulation by SA-LCU outperforms the first- and second-order Trotter methods and qDRIFT. In (b), we vary the evolution time ($t$) for a fixed precision, $\eps=0.01$, where second-order Trotter outperforms the other methods. \label{fig:magnetization}}    
\label{fig:br}
\end{figure*}
%%%%%%%%%%%%%%%%%%%%%%%%%%%%%%%%%%%%%%%%%%%%%%%%%%%%%%%%%%%%%%%

\subsection{Numerical benchmarking: Ising model under amplitude damping}
We numerically benchmark the performance of the Markovian quantum collision model for simulating Lindblad dynamics by applying it to a concrete problem. We consider the one-dimensional transverse-field Ising model with nearest-neighbor interactions (also known as the Heisenberg XXX model), a widely used testbed for benchmarking Hamiltonian simulation techniques due to its physical significance in condensed matter physics~\cite{sachdev2011quantum}. We look at the dynamics of this system when the environment is a discrete sum of sub-environments, each corresponding to a single-qubit amplitude-damping channel. The $j$-th collision corresponds to the $j$-th sub-environment qubit interacting non-trivially with site $j$ of the system Hamiltonian. Thus, the total number of sub-environments is the same as the number of sites in the Ising chain. Consequently, let us define the system Hamiltonian as follows:
\begin{equation}
\label{eq: ising model}
H_S = -J \sum_{j=1}^{m-1} \sigma_i^z \sigma_{i+1}^z 
      - h \sum_{j=1}^{m} \sigma_i^x,
\end{equation}
where $J$ denotes the coupling strength between nearest neighbors, $h$ is the transverse magnetic field strength, and $\sigma_j^{z}$ and $\sigma_j^{x}$ are Pauli operators acting on site $j$. The total number of sites $m$, is the same as the number of sub-environments. Finally, we accommodate for the fact that the system interacts with each sub-environment one at a time by considering the rescaled system Hamiltonian $H_S/m$, with $\beta_S = (J + h)$. The environment is a discrete sum of $m$ single-qubit number operators, with each sub-environment being in the state $\ket{0}$, corresponding to a thermal state at zero temperature $(\omega\to\infty)$, i.e., $\rho_{E_j}=\ket{0}\bra{0}$, for all $j\in [1,m]$.

We define the interaction Hamiltonian corresponding to the $j$-th collision as
\begin{align}
    H_{I_j} =&\ \sqrt{\gamma} (\mathbb{I}^{j-1} \otimes \hat{\sigma}^+_j \otimes \mathbb{I}^{m-j-1}\otimes \hat{\sigma}^-_a \nonumber \\
    &\qquad + \mathbb{I}^{j-1} \otimes \hat{\sigma}^-_j \otimes \mathbb{I}^{m-j-1}\otimes \hat{\sigma}^+_a),
    \label{eq:lindbladian-amplitude-interaction-ham}
\end{align}
where $a$ denotes the environment qubit, $\gamma$ denotes the damping strength and $\sigma_-$ denotes the lowering operator. We estimate the average transverse-field magnetization,
\begin{equation}
\label{eq: Magnetization operator}
M_z=\dfrac{1}{m}\sum_{j=1}^{m}\sigma^{z}_{j},
\end{equation}
with respect to the reduced state $e^{\mathcal{L}t}[\rho_S]$, i.e., we obtain $\mu$ such that
$$
\left|\mu-\Tr\left[M_z~e^{\mathcal{L}t}[\rho_S]\right]\right|\leq \eps.
$$
Note that the Lindblad master equation dynamics, which we numerically simulate, is given by Eq.~\eqref{eq: lindbladian map}. That is, we have:
\begin{equation}
    \mathcal{L}[\rho_S]\equiv \diffp{\rho_S}{t}= -i[H_S, \rho] + \sum_j \left( A_j \rho A_j^\dagger - \frac{1}{2} \{A_j^\dagger A_j, \rho\}\right)
    \label{eq: lindbladian-map-amp-damp},
\end{equation}
where the jump operator for the amplitude damping on $j\text{th}$ qubit,
\begin{equation}
    A_j = \sqrt{\gamma} \: \mathbb{I}^{j-1} \otimes \sigma_{-} \otimes \mathbb{I}^{n-j}.
    \label{eq:jump-operator-amp-damp}
\end{equation}
% The corresponding Lindbladian super-operator  of dimension $2^{n^2}\times 2^{n^2}$ can be written as
% \begin{equation}
%     \begin{split}
%         \mathcal{L} = i (H_S^T \otimes I^{n} &-\mathbb{I}^{n} \otimes H_S) \\
%         &+ \sum_k A_k^{*} \otimes A_k - \frac{1}{2} \mathbb{I}^{n} \otimes (A_k^{\dagger} A_k) - \frac{1}{2} A_k^{T}A_k^{*} \otimes \mathbb{I}^{n}
%     \end{split}
% \end{equation}

We perform numerical benchmarking for estimating the desired expectation value on a $10$-qubit ($m=10$) transverse field Ising model under amplitude damping via the quantum collision model using different Hamiltonian simulation procedures (first and second order Trotter methods, SA-LCU, and qDRIFT). In Fig.~\ref{fig:lindbladian_a}, we compare the CNOT gate counts per coherent run to fix the Lindblad evolution time to $t=1$ and estimate $\mu$ for different values of $\eps$. In Fig.~\ref{fig:lindbladian_b}, we fix $\eps=0.01$ and vary $t$ instead. To obtain these plots, we set the coupling strength $J=1$ and the transverse magnetic field strength $h=0.1$ in the system Hamiltonian $H_S$. The strength of the amplitude damping channel for each interaction Hamiltonian is fixed to $\gamma=1$. The time of each collision, $\Delta t$, depends on the precision $\eps$ and $t$. Hence, $\Delta t$ differs for different values of $\eps$ and $t$ in these figures. 

To obtain the CNOT gate count for each Hamiltonian simulation procedure, we construct the entire circuit on Qiskit using a fully-connected circuit architecture. This corresponds to a composition of Hamiltonian simulations. For this purpose, we use the Solovay-Kitaev theorem (available in Qiskit) to decompose the circuit into a basis comprising single-qubit rotations and CNOT gates. The usual gate optimizations available on Qiskit are applied to all the circuits to obtain a non-trivial total CNOT count. Finally, for a fair comparison, we choose the number of Trotter steps based on the tighter commutator bounds~\cite{childs2021theory} for the first and second-order Trotter methods. 

Fig.~\ref{fig:lindbladian_a} shows that, for a fixed $t$, the Hamiltonian simulation by the SA-LCU method performs better than the first and second-order Trotter methods and qDRIFT. This is because it has a better dependence on the precision. On the other hand, in Fig.~\ref{fig:lindbladian_b}, the second-order Trotter method outperforms the rest when $\eps$ is fixed and $t$ is increased.

%%%%%%%%%%%%%%%%%%%%%%%%%%%%%%%%%%%%%%%%%%%%%%%%%%%%%%%%%%%%%%%
\begin{figure*}[!t]
    \centering
    \includegraphics[width=0.8\textwidth]{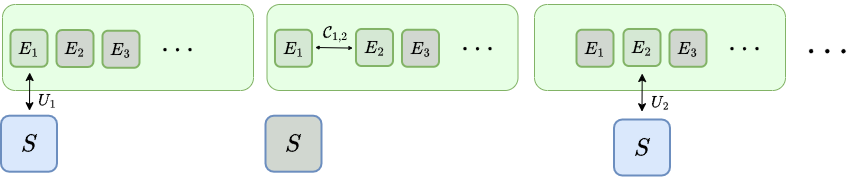}
    \caption{Non-Markovian evolution via the collision model illustrating the interleaved dynamics between system-environment and the additional intra-environment interactions. The system $S$ (blue) interacts with the environmental subsystems $E_i$ (green) through unitary operations $U_i$, while only the adjacent environmental subsystems interact via channel $C_{i,i+1}$. This sequential structure may create a propagating chain of correlations, where information flows not only between the system and environment but also via nearest-neighbor interactions. The three panels represent consecutive time steps of the evolution, demonstrating how correlations may build up and propagate through the environmental subsystems, capturing the memory effects and non-Markovian behavior of the quantum dynamics.}
    \label{fig:collision-second}
\end{figure*}
%%%%%%%%%%%%%%%%%%%%%%%%%%%%%%%%%%%%%%%%%%%%%%%%%%%%%%%%%%%%%%%
\section{Simulation of non-Markovian quantum collision models}\label{sec:Non-markovian dynamics}

In the quantum algorithm we develop in Sec.~\ref{sec:Modified Collision Model}, the interactions correspond to memoryless collisions: each sub-environment qubit interacts for $\Delta t$ before being traced out. Thus, a $K$-collision map can only generate Markovian dynamics. Within the collision model framework, this translates to each sub-environment interacting with the system in isolation without influencing other sub-environments. For instance, Lindbladian dynamics adhere to a strict Markovian assumption, where the environmental dynamics occur on much faster timescales than the system, effectively ensuring that the environment remains unaffected by the interaction. On the other hand, non-Markovian dynamics naturally arise when the collisions are memory-retaining: interactions between the sub-environments ensure some information about the prior collisions is retained~\cite{ciccarello2022collisionreview}. In this section, we extend the $K$-collision map to the non-Markovian framework, incorporating interactions between the different sub-environments. In particular, we consider interactions that preserve the CPTP nature of the maps, such that they can be seamlessly composed to obtain the reduced dynamics of the system. To this end, we discuss next the framework introduced by Ciccarello et al.~\cite{Ciccarello2013PRA}, which provides a simple way to incorporate sub-environment interactions. 

As in Sec.~\ref{sec:Modified Collision Model}, we consider an $n$-qubit system described by the Hamiltonian $H_S$ and prepared in the state $\rho_S$ from the Hilbert space $\mathcal{H}_S$. It interacts with an environment made up of $m$ discrete sub-environments described by the Hamiltonian $H_{E_j}$ and prepared in states $\rho_{E_j}$, for $j\in [1,m]$.  Each sub-environment lives in the Hilbert space $\mathcal{H}_{E_j}$ and the overall environment belongs to $\otimes_{j=1}^m \mathcal{H}_{E_j}$. For simplicity, we assume that the dimension of $\mathcal{H}_{E_j}$ is the same for all $j\in [1, m]$. Also, the interaction Hamiltonian corresponding to the $j$-${\mathrm{th}}$ collision between the system and the $j$-${\mathrm{th}}$ sub-environment is $H_{I_j}$. So, the total Hamiltonian for the $j$-${\mathrm{th}}$ collision remains the same as in Sec.~\ref{sec:Modified Collision Model}, i.e., 
$$
H_{j}=H_S+H_{I_j}+H_{E_j}.
$$
However, now, we assume that a collision between the system and the environment is followed by a collision between consecutive sub-environments via some CPTP quantum channel $\mathcal{C}_{i,j}$ as follows: Initially, the system in the state $\rho_S$ interacts with $E_1$ prepared in the state $\rho_{E_1}$ for $\Delta t$. Then, $E_1$ interacts with the next sub-environment $E_2$, prepared in $\rho_{E_2}$ via the quantum channel $\mathcal{C}_{1,2}$. It is only after this intra-environment interaction that $E_1$ is traced out. Thus, a single iteration now consists of two collisions: a system-sub-environment collision followed by a collision between two consecutive sub-environments. 

This sequence of interactions continues for some $K$ iterations, as depicted in Fig.~\ref{fig:collision-second}, and leads to a non-Markovian $K$-collision map. At the $j$-${\mathrm{th}}$ iteration, the system and $j$-th sub-environment interact for time $\Delta t$. Subsequently, a new environment state $\rho_{E_{j+1}}$ is initialized, and $E_j$ interacts with $E_{j+1}$ via the CPTP channel $\mathcal{C}_{j, j+1}$, following which $E_j$ is then traced out. Ciccarello et al. \cite{Ciccarello2013PRA} demonstrated that when $\mathcal{C}_{j,j+1}$ is the partial swap operation between consecutive sub-environment qubits, the corresponding collision model leads to a non-Markovian master equation in the limit where the number of collisions $K\to\infty$. Or, more precisely, for $p\in [0,1]$ and two states $\rho_j$ and $\sigma_{j+1}$ of the same dimension,
\begin{align}
    \label{eq:non-markovian-partial-swap}
    &\hspace{-0.5cm}\mathcal{C}_{j, j+1}[\rho_j\otimes \sigma_{j+1}]=\nonumber\\ 
    &(1 - p) (\rho_{j}\otimes\sigma_{j+1}) + S_{j,j+1}(\rho_j\otimes\sigma_{j+1})\text{S}_{j,j+1}^{\dagger},
\end{align}
where $S_{j,j+1}$ swaps the two states $\rho_j$ and $\sigma_{j+1}$. 
The parameter $p$ provides a handle over the degree of non-Markovianity, ranging from no information transfer (memoryless collisions) for $p=0$ to perfect swapping of information between consecutive sub-environments for $p=1$. 

We now define the map corresponding to the $j$-${\mathrm{th}}$ iteration (similar to Definition \ref{def: collision map}) with two sub-environment registers (instead of one as in the Markovian case). There is one subtlety in defining the maps, namely, which sub-environment register interacts with the system and which one gets traced out. At the $j$-${\mathrm{th}}$ iteration, if $j$ is odd (even), the second (first) sub-environment register stores $\rho_{E_{j+1}}$, while the first (second) sub-environment register interacts with the system, and is ultimately traced out. For $j\in [1,m]$, let us then define, 
\begin{equation}
\label{eq:nm-system-env-collision}
U_{S_j}=\bar{U}_j\otimes I_E,    
\end{equation}
where $\bar{U}_{j}=e^{-i\Delta t \beta_j\overline{H}_j}$ denotes the collision between the system $S$ and the $j$-${\mathrm{th}}$ sub-environment. Here $\overline{H}_j$ is the normalized hamiltonian to make it's decomposition a convex combination of pauli operators and $\beta_j$ is the normalizing factor. %More precisely, this acts on the system and the first (second) environment register, when $j$ is odd (even). 
Furthermore, we define the CPTP map,
\begin{equation}
\label{eq:nm-env-env-collision}
    \mathcal{V}_{j,j+1}=I_S\otimes C_{j,j+1},
\end{equation}
where $C_{j,j+1}$ is the partial swap operation defined in Eq.~\eqref{eq:non-markovian-partial-swap}. The collision map is applied to a composite state of the system and the two sub-environments $E_j$ and $E_{j+1}$, with the state of the latter prepared in $\rho_{E_{j+1}}$. As explained above, the sub-environment register storing $\rho_{E_{j+1}}$ changes depending on whether $j$ is even or odd. Thus, we can now formally define the $j$-${\mathrm{th}}$ non-Markovian collision map $\Phi^{\mathcal{N}}_j$ as:
\begin{equation}
    \label{eq:nm-collision-map}
        \Phi^{\mathcal{N}}_j[\cdot] \equiv 
        \begin{cases}
\Tr_{E_{j}}\big[\mathcal{V}_{j,j+1}\big[U_{S_j}\left(\cdot\otimes\cdot\otimes \rho_{E_{j+1}}\right)U_{S_j}^{\dagger}\big]\big],~ j~\text{is odd,}\\\vspace{-0.5cm}\\
\Tr_{E_{j}}\big[\mathcal{V}_{j,j+1}\big[U_{S_j}\left(\cdot\otimes \rho_{E_{j+1}} \otimes\cdot\right)U_{S_j}^{\dagger}\big]\big],~j~\text{is even.}
        \end{cases}     
\end{equation}

Note that irrespective of whether $j$ is odd or even, the $j$-${\mathrm{th}}$ sub-environment $E_j$ is traced out [it is the first (second) sub-environment register when $j$ is odd (even)]. Let us now look at the overall map for $K$ iterations of system-environment and environment-environment collisions. It is essentially equivalent to composing the $\Phi^{\mathcal{N}}_j$[.] map a total of $K-1$ times, following which the reduced state is the composite state of the system and the $K$-th sub-environment. At this stage, the sub-environment $K$ does not interact with the next sub-environment. So, the final map is an interaction between the system and this final sub-environment, according to $U_{S_K}$ for time $\Delta t$, followed by the tracing out of $E_K$, leaving only the transformed state of the system. Thus, the non-Markovian $K$-collision map can be defined as
\begin{equation}
    \label{eq:nm-k-collision}
\mathcal{N}_K[\cdot]\equiv\Tr_{E_K}\left[U_{S_K}\left(\bigcirc_{j=1}^{K-1}\Phi^{\mathcal{N}}_j[\cdot]\right)U^{\dag}_{S_K}\right].
\end{equation}
\vspace{-0.4cm}
%%%%%%%%%%%%%%%%%%%%%%%%%%%%%%%%%%%%%%%%%%%%%%%%%%%%%%%%%%%%%%%
\begin{algorithm*}[t]
\SetAlgoCaptionLayout{bottom}
    \caption{Algorithm to estimate the expectation value of an observable $O$ with respect to the non-Markovian $K$-collision map applied to a quantum state.}\label{algo: nm-dynamics}
  \KwIn{Initial system state in $\rho_S$, sub-environment states $\rho_{E_1}$ to $\rho_{E_K}$, observable $O$, unitaries $\widetilde{U}_1$, \dots. $\widetilde{U}_K$, and precision $\eps'$, where the LCU decomposition of each $\widetilde{U}_j=\sum_{k}\alpha_{jk} W_{jk}$, such that $\forall j\in [1, K]$, $\norm{\widetilde{U}_j-e^{-i\Delta t\overline{H}_j}}\leq \eps'$.~\\~\\}
    \begin{itemize}
        \item[1.~] Initialize the system, the first environment register, and the ancilla in states $\rho_S$, $\rho_{E_1}$, and $\ket{+}$, respectively.
        \item[2.~] For iterations from $j = 1$ to $K-1$:
            \begin{itemize}
                \item[a.~] Draw two i.i.d.~samples $X_j$ and $Y_j$ from the ensemble 
                    \begin{equation*}
                        \mathcal{D}_{j}=\left\{ W_{j \ell}, \frac{\alpha_{j \ell}}{\alpha^{(j)}}\right\},
                    \end{equation*} 
                    where $\alpha^{(j)} = \sum_{\ell}|\alpha_{j \ell}|$
                \item[b.~] Apply the controlled unitary $X_j^{(c)}$ and the anti-controlled unitary $Y_j^{(a)}$, controlled on the ancilla qubit with the target being the system register and the first (second) environment register if $j$ is odd (even).
            \item[c.~] If $j$ is odd, initialize the second environment register in the state $\rho_{E_{j+1}}$;  otherwise, initialize the first environment register in $\rho_{E_{j+1}}$.
                \item[d.~] With probability $p$ apply the swap gate $S_{j,j+1}$ between the two environment registers.
                \item[e.~] If $j$ is odd, perform a partial trace over the first environment register; otherwise, perform a partial trace over the second environment register. 
            \end{itemize}
      \item[4.~] Draw two i.i.d.~samples $X_K$ and $Y_K$ from the ensemble 
                    \begin{equation*}
                        \mathcal{D}_{K}=\left\{ W_{K \ell}, \frac{\alpha_{K \ell}}{\alpha^{(K)}}\right\},
                    \end{equation*} 
                    where $\alpha^{(K)} = \sum_{\ell}|\alpha_{K \ell}|$.
        \item[5.~] Apply the unitary $X_K^{(c)}$ and the anti-controlled unitary $Y_K^{(a)}$, controlled on the ancilla with the target being the system and the first (second) environment register if $K$ is odd (even). 
        
      \item[6.~] If $K$ is odd, perform a partial trace over the first environment register; otherwise, perform a partial trace over the second environment register.
      
      \item [7.~]  Measure the joint ancilla and system state on the observable $(\sigma^x \otimes O)$ and record the measurement outcome as $\mu_i$.
      
      \item [8.~] Repeat Steps 1 to 7 a total of $T$ times.
    
      \item [9.~] Compute the final estimate $\mu$ as:
    \begin{equation*}
        \mu = \dfrac{\zeta^2}{T}\sum_{j=1}^{T} \mu_j,
    \end{equation*}
      where $\zeta = \prod_{j=1}^K \alpha^{(j)}$.
      \end{itemize}
  \KwOut{Estimated expectation value $\mu$}  
\end{algorithm*} 
%%%%%%%%%%%%%%%%%%%%%%%%%%%%%%%%%%%%%%%%%%%%%%%%%%%%%%%%%%%%%%%

To develop quantum algorithms for implementing the non-Markovian $K$-collision map $\mathcal{N}_K[.]$, there are a few things to consider. First, it is possible to implement the partial swap $C_{j,j+1}$ efficiently in a probabilistic manner: we apply $S_{j,j+1}$ with probability $p$ and with probability $1-p$, we do not perform any operation. Second, as in the Markovian case, the system environment collisions $U_{S_j}$ can be implemented with any Hamiltonian simulation procedure. We assume that the swap operation can be implemented perfectly so that the only source of error is the underlying Hamiltonian simulation procedure. That is, for each iteration, we implement an approximate map
\begin{equation}
\label{eq: approximate non-markovian collision map}
\widetilde{\Phi}^{\mathcal{N}}_j[\cdot]\equiv \begin{cases}
\Tr_{E_{j}}\big[\mathcal{V}_{j,j+1}\big[\widetilde{U}_{S_j}\left(\cdot\otimes \cdot \otimes \rho_{E_{j+1}}\right)\widetilde{U}_{S_j}^{\dagger}\big]\big],~j~\text{is odd,}\\\vspace{-0.5cm} \\
\Tr_{E_{j}}\big[\mathcal{V}_{j,j+1}\big[\widetilde{U}_{S_j}\left(\cdot\otimes \rho_{E_{j+1}} \otimes \cdot\right)\widetilde{U}_{S_j}^{\dagger}\big]\big],~j~\text{is even.}
\end{cases}     
\end{equation}
Analogously, for the overall map, we define
\begin{equation}
\label{eq: approximate non-markovian K-collision map}
\widetilde{\mathcal{N}}_K[.]\equiv\Tr_{E_K}\left[\widetilde{U}_{S_K}\left(\bigcirc_{j=1}^{K-1}\widetilde{\Phi}^{\mathcal{N}}_j[.]\right)\widetilde{U}^{\dag}_{S_K}\right].    
\end{equation}
Next, we show that the required precision for a Hamiltonian simulation procedure to estimate $\Tr[O~\mathcal{N}_K[\rho_S]]$ with $\eps$-additive accuracy remains the same as in the Markovian case. For this, we prove the following Lemma:
\begin{restatable}[Bounds on the non-Markovian approximate collision Map]{lemma}{}\label{lemma: nm-approximate K-collision map distance bound}
    Let $O$ be an observable and  $\widetilde{\mathcal{N}}_K$ be the approximate non-Markovian $K$-collision map in Eq.~\eqref{eq:approx-K-collision-map}, 
    where
    \begin{equation}
        \max_{1\leq j \leq K}\norm{U_j - \widetilde{U}_j} \leq \frac{\eps}{3K\norm{O}}.
    \end{equation}
    Then, the expectation value of $O$ with respect to the state transformed under the approximate map is $\eps$-close to the expectation value under the exact map. That is, for any $\rho$,
    \begin{equation}
        \left|\Tr\left[O~\mathcal{N}_K[\rho]\right] - \Tr\left[O~ \widetilde{\mathcal{N}}_K[\rho]\right]\right| \leq \eps.
    \end{equation}
\end{restatable}
\begin{proof}
We consider the error between the operations performed on the state $\rho$ under $U_{S_j}$ and $\widetilde{U}_{S_j}$. More precisely, let $\norm{U_j-\widetilde{U}_j}\leq \xi_j$. Then, immediately, we have
$\|U_{S_j}-\widetilde{U}_{S_j}\|\leq \xi_j$. As before, we denote the maximum error in any of the Hamiltonian simulation procedures in the definition of the approximate $K$-collision map by $\xi_{\max}$, i.e., $\xi_{\max}=\max_{1\leq j \leq K} \xi_{j}$. Then, using Theorem~\ref{thm:distance-expectation} for any quantum state $\rho$, we have
    \begin{align}
        &\big\| U_{S_j} \rho  U_{S_j}^{\dagger} - \widetilde{U}_{S_j} \rho \widetilde{U}_{S_j}^{\dagger}\big\|_1 
        \leq 3 \xi_j \leq 3\xi_{\max}.
    \end{align}
Now, since partial trace and $\mathcal{V}_{j,j+1}$ are CPTP maps, we can make use of the distance between the composition of CPTP maps (Theorem \ref{thm: Bounds on the distance between the composition of maps}) to get
\begin{widetext}
    \begin{align}
        \left\| \bigcirc_{j=1}^{K-1} \Phi^{\mathcal{N}}_j[\rho] - \bigcirc_{j=1}^{K-1} \widetilde{\Phi}^{\mathcal{N}}_j[\rho] \right\|_1 \leq 3 (K-1) \xi_{\max}.
    \end{align}
For the final step, we need to bound
    \begin{align}
        \norm{\left[\mathcal{N}_K[\rho]\right] - \widetilde{\mathcal{N}}_K[\rho]}_1 &= \norm{ \Tr_{E_K}\left[ U_{S_K}\left(\bigcirc_{j=1}^{K-1}\Phi^{\mathcal{N}}_j[\rho]\right)U^{\dag}_{S_K}\right]-\Tr_{E_K}\left[\widetilde{U}_{S_K}\left(\bigcirc_{j=1}^{K-1}\widetilde{\Phi}^{\mathcal{N}}_j[\rho]\right)\widetilde{U}^{\dag}_{S_K}\right]}_1\nonumber\\
        &\leq \norm{U_{S_K}\left(\bigcirc_{j=1}^{K-1}\Phi^{\mathcal{N}}_j[\rho]\right)U^{\dag}_{S_K}-\widetilde{U}_{S_K}\left(\bigcirc_{j=1}^{K-1}\widetilde{\Phi}^{\mathcal{N}}_j[\rho]\right)\widetilde{U}^{\dag}_{S_K}}_1\nonumber\\
        &\leq \norm{\bigcirc_{j=1}^{K-1}\Phi^{\mathcal{N}}_j[\rho]-\bigcirc_{j=1}^{K-1}\widetilde{\Phi}^{\mathcal{N}}_j[\rho]}_1+ 2\norm{U_{S_K}-\widetilde{U}_{S_K}}
        \leq 3(K-1)\xi_{\max}+2\xi_{\max} < 3 K\xi_{\max},
    \end{align}
\end{widetext}
where we have used Theorem~\ref{thm:dist-unitary-cptp} to arrive at the third line from the second. Finally, as in the proof of Lemma~\ref{lemma: approximate K-collision map distance bound}, we can use Theorem~\ref{thm:distance-expectation} and set $\xi_{\max} = \eps/(3 K \|O\|)$ and get:
    \begin{equation}
        \left| \Tr\left[O~ \mathcal{N}_K[\rho] \right] - \Tr\left[O~ \widetilde{\mathcal{N}}_K[\rho] \right] \right| \leq \eps.
    \end{equation}
\end{proof}

%%%%%%%%%%%%%%%%%%%%%%%%%%%%%%%%%%%%%%%%%%%%%%%%%%%%%%%%%%%%%%%

This implies that for estimating the desired expectation value, we need to implement a Hamiltonian simulation procedure with precision $\eps'=\bigO(\eps/K\|O\|)$, which is the same as in the Markovian case. Thus, we can develop a randomized quantum algorithm that can incorporate any near-term Hamiltonian simulation procedure. In Algorithm~\ref{algo: nm-dynamics}, we use Hamiltonian simulation by the SA-LCU method to estimate the desired expectation value. We show the circuit corresponding to each run of the algorithm in Fig.~\ref{fig: single ancilla non markovian collision model circuit}. We now need two sub-environment registers, along with the system register and a single qubit ancilla register. For the first $K-1$ iterations, a composition of the map $\widetilde{\Phi}^{\mathcal{N}}[.]$ is implemented. Note that at odd (even) iterations, the second (first) register stores the state of the subsequent sub-environment. Following the interaction between the two sub-environment registers, the first (second) environment register is traced out. 

%%%%%%%%%%%%%%%%%%%%%%%%%%%%%%%%%%%%%%%%%%%%%%%%%%%%%%%%%%%%%%%
\begin{figure*}[!t]
    \scalebox{0.8}{
    \includegraphics[width=1\textwidth]{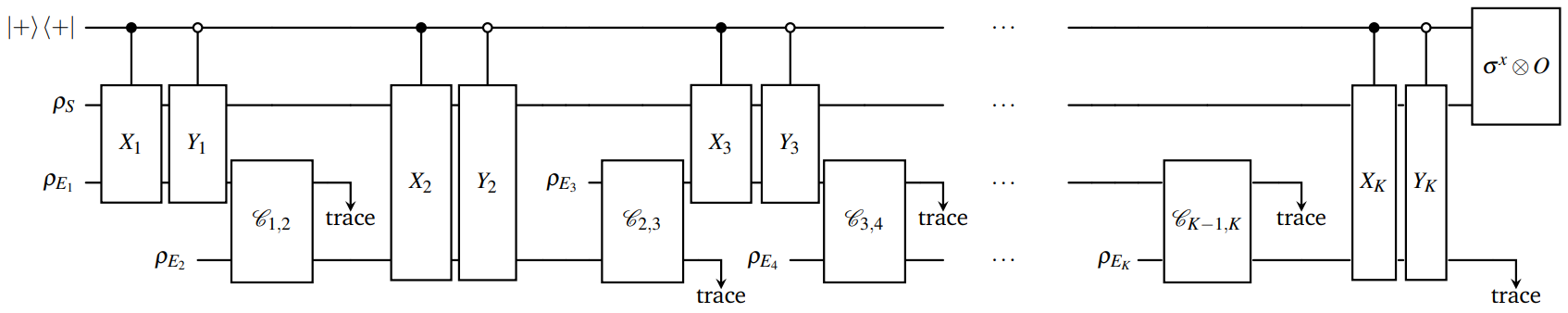}
    }
    \centering
    \caption{
Quantum circuit corresponding to each run of Algorithm \ref{algo: nm-dynamics}, simulating a non-Markovian $K$-collision map, using Hamiltonian simulation by SA-LCU. The algorithm applies controlled and anti-controlled sampled unitaries ($X_j, Y_j$) for the interaction between the system and each sub-environment, followed by an interaction between consecutive sub-environments using the channel ($\mathcal{C}_{j, j+1}$). This sequence is repeated for $K$ collisions. At the end of the process, the ancilla qubit and the system are measured.
}\label{fig: single ancilla non markovian collision model circuit}
\end{figure*}
%%%%%%%%%%%%%%%%%%%%%%%%%%%%%%%%%%%%%%%%%%%%%%%%%%%%%%%%%%%%%%%

Following an approach similar to Theorem~\ref{thm: Algorithm 1 proof}, we prove the correctness of Algorithm~\ref{algo: nm-dynamics} in the Appendix (Theorem~\ref{thm: Algorithm 3 proof}, Appendix~\ref{sec:app-correctness-algo-nm}). In step 2d.~of Algorithm~\ref{algo: nm-dynamics}, we implement a map that in expectation value is simply $\mathcal{C}_{j,j+1}$. Interestingly, if we assume that the cost of implementing a swap gate between two consecutive sub-environments is constant, the circuit depth of the procedure using the different simulation techniques is the same as those listed in Table~\ref{table: collision model complexity}. Thus, this is a unified framework to simulate non-Markovian collisions using near-term Hamiltonian simulation procedures. Again, we do not require many ancilla qubits or need access to block encodings.

It would be interesting to explore whether continuous-time non-Markovian master equations can be approximated by this collision model for finite $K$. This would require obtaining the number of collisions $K$ for which it is $\eps$-close (in, say, induced $1$-norm) to the generator of the underlying non-Markovian master equation. While this problem has been investigated in the Markovian setting (closeness of the $K$-collision map and exponential of the Lindbladian~\cite{cattaneo2021collision, pocrnic2023quantum}), very little is known in the non-Markovian case. Ciccarello et al.~\cite{Ciccarello2013PRA}, showed that the collision model we consider gives rise to a non-Markovian master equation for $K\to\infty$ and $p=e^{-\lambda t}$, where $\lambda$ is a continuous parameter, determining the \textit{memory rate}. However, we leave the question of the precise scaling of the error in this approximation open with (finite) $K$.

\section{Discussion and Outlook}\label{sec:Discussion and Outlook}
In this paper, we developed randomized quantum algorithms for simulating open quantum system dynamics using quantum collision models designed for early fault-tolerant quantum computers. Our approach enables estimating expectation values of observables for a system that undergoes an arbitrary number of collisions, encompassing both memoryless (Markovian) and memory-retaining (non-Markovian) interactions. Thus, it provides a unified framework for simulating a broad range of open-system dynamics. 

A key advantage of quantum collision models is that they naturally decompose the environment into discrete sub-environments, making system-environment interactions efficiently implementable on near-term quantum devices. Indeed, our methods for simulating both Markovian and non-Markovian dynamics avoid the need for block encodings, significantly reducing ancilla requirements.

We rigorously analyzed the cost of implementing $K$ Markovian collisions, where the system sequentially interacts with different sub-environments before they are traced out, by leveraging various near-term Hamiltonian simulation techniques such as Trotterization, qDRIFT, and Hamiltonian simulation by SA-LCU. This provided a unified framework to simulate Lindblad dynamics using quantum collision models. We undertook a detailed comparison of the end-to-end complexities of these techniques for simulating Lindblad dynamics, which revealed that Hamiltonian simulation by the SA-LCU method outperforms the first and second-order Trotter methods and qDRIFT. However, overall, $2k$-order Trotter methods offer the most competitive circuit depth, scaling as $\widetilde{\bigO}(t^2 / \varepsilon)$, for large $k$. Notably, this matches known lower bounds for implementing $e^{\mathcal{L}t}$ using quantum collision models~\cite{cleve2017efficient, cattaneo2021collision}.

While there have been several direct approaches to simulating Linbladian dynamics~\cite{cleve2017efficient, li2023simulating, ding2024simulating}, they are resource-demanding and hence are not implementable on early fault-tolerant machines. Quantum collision models provide a simple, easy-to-implement route. It would be interesting to explore whether it is possible to exploit recently developed extrapolation techniques (used to improve the circuit depth of Trotterization~\cite{watson2024exponentiallyreducedcircuitdepths}, qDRIFT~\cite{watson2024randomly}, and even more general quantum algorithms~\cite{chakraborty2025quantum}), in conjunction with our randomized quantum algorithm to further improve the overall circuit depth. While the $t^2/\eps$-dependence on the circuit depth is unavoidable for quantum collision models, extrapolation techniques might lead to optimal complexity even with the first-order Trotter method and qDRIFT. Thus, better direct approaches must be investigated to simulate Lindblad dynamics on early fault-tolerant quantum computers with even shorter circuit depths. In this regard, during the preparation of this manuscript, we became aware of Ref.~\cite{kato2024exponentially}. There, the authors assume that the Lindblad jump operators can be expressed as a linear combination of strings of Pauli operators and then write down $e^{\mathcal{L}t}$ as an LCU (much like the LCU decomposition of $e^{-iHt}$ in SA-LCU~\cite{wan2022randomized, wang2023qubit, chakraborty2023implementing}). The authors estimate the expectation value of $O$ with respect to $e^{\mathcal{L}t}$ with $\eps$-additive accuracy in $\bigO(\|O\|^2/\eps^2)$ classical repetitions and circuit depth $\bigO(t^2\log(t/\eps))$. The procedure requires $n+4$ qubits overall and has an exponentially better circuit depth than any procedure making use of quantum collision models.

Another advantage of quantum collision models is that their applicability goes beyond the simulation of Markovian dynamics. Indeed, we extended our framework to non-Markovian interactions. In particular, we considered the setting of~\cite{Ciccarello2013PRA}, wherein a collision between the system and a sub-environment is followed by an interaction between consecutive sub-environments, thereby retaining some memory of prior interactions. We developed a randomized method to simulate arbitrary such collisions using near-term Hamiltonian simulation procedures. The circuit depth for implementing $K$ such collisions scales similarly to the Markovian case - once again, no block encoding is required, and the number of ancilla qubits needed is minimal. Thus, this provides a general method to implement such dynamics resource-efficiently on quantum devices. Much like the Markovian case, our approach stands in contrast to the direct approaches (such as Ref.~\cite{li2023succinct}), which are resource-demanding and hence unsuitable for near-term implementation.

Finally, an interesting advance would be to simulate Lindblad master equations with time-dependent decay parameters via quantum collision models. This would require establishing convergence between such time-dependent Lindblad maps and quantum collision models generated by time-dependent Hamiltonians. One could then use quantum algorithms for simulating time-dependent Hamiltonians \cite{berry2020timedependent, chen2021quantumalgorithm,watkins2024time-dependent,fang2025time} to construct end-to-end simulation protocols for a broader class of open-system dynamics, including both Markovian and non-Markovian cases.

Overall, our results demonstrate that quantum collision models can be used to simulate a wide range of open systems dynamics, both under Markovian and non-Markovian environments, using early fault-tolerant quantum computers. 

\begin{acknowledgments}
We acknowledge funding from the Ministry of Electronics and Information Technology (MeitY), Government of India, under Grant No. 4(3)/2024-ITEA. SC also acknowledges support from Fujitsu Ltd, Japan and IIIT Hyderabad.
\end{acknowledgments}

\bibliography{apssamp}{} % Produces the bibliography via BibTeX.

\begin{thebibliography}{72}%
\makeatletter
\providecommand \@ifxundefined [1]{%
 \@ifx{#1\undefined}
}%
\providecommand \@ifnum [1]{%
 \ifnum #1\expandafter \@firstoftwo
 \else \expandafter \@secondoftwo
 \fi
}%
\providecommand \@ifx [1]{%
 \ifx #1\expandafter \@firstoftwo
 \else \expandafter \@secondoftwo
 \fi
}%
\providecommand \natexlab [1]{#1}%
\providecommand \enquote  [1]{``#1''}%
\providecommand \bibnamefont  [1]{#1}%
\providecommand \bibfnamefont [1]{#1}%
\providecommand \citenamefont [1]{#1}%
\providecommand \href@noop [0]{\@secondoftwo}%
\providecommand \href [0]{\begingroup \@sanitize@url \@href}%
\providecommand \@href[1]{\@@startlink{#1}\@@href}%
\providecommand \@@href[1]{\endgroup#1\@@endlink}%
\providecommand \@sanitize@url [0]{\catcode `\\12\catcode `\$12\catcode
  `\&12\catcode `\#12\catcode `\^12\catcode `\_12\catcode `\%12\relax}%
\providecommand \@@startlink[1]{}%
\providecommand \@@endlink[0]{}%
\providecommand \url  [0]{\begingroup\@sanitize@url \@url }%
\providecommand \@url [1]{\endgroup\@href {#1}{\urlprefix }}%
\providecommand \urlprefix  [0]{URL }%
\providecommand \Eprint [0]{\href }%
\providecommand \doibase [0]{https://doi.org/}%
\providecommand \selectlanguage [0]{\@gobble}%
\providecommand \bibinfo  [0]{\@secondoftwo}%
\providecommand \bibfield  [0]{\@secondoftwo}%
\providecommand \translation [1]{[#1]}%
\providecommand \BibitemOpen [0]{}%
\providecommand \bibitemStop [0]{}%
\providecommand \bibitemNoStop [0]{.\EOS\space}%
\providecommand \EOS [0]{\spacefactor3000\relax}%
\providecommand \BibitemShut  [1]{\csname bibitem#1\endcsname}%
\let\auto@bib@innerbib\@empty
%</preamble>
\bibitem [{\citenamefont {Lloyd}(1996)}]{lloyd1996universal}%
  \BibitemOpen
  \bibfield  {author} {\bibinfo {author} {\bibfnamefont {S.}~\bibnamefont
  {Lloyd}},\ }\bibfield  {title} {\bibinfo {title} {Universal quantum
  simulators},\ }\href {https://doi.org/10.1126/science.273.5278.1073}
  {\bibfield  {journal} {\bibinfo  {journal} {Science}\ }\textbf {\bibinfo
  {volume} {273}},\ \bibinfo {pages} {1073} (\bibinfo {year}
  {1996})}\BibitemShut {NoStop}%
\bibitem [{\citenamefont {Berry}\ \emph {et~al.}(2014)\citenamefont {Berry},
  \citenamefont {Childs}, \citenamefont {Cleve}, \citenamefont {Kothari},\ and\
  \citenamefont {Somma}}]{berry2014exponential}%
  \BibitemOpen
  \bibfield  {author} {\bibinfo {author} {\bibfnamefont {D.~W.}\ \bibnamefont
  {Berry}}, \bibinfo {author} {\bibfnamefont {A.~M.}\ \bibnamefont {Childs}},
  \bibinfo {author} {\bibfnamefont {R.}~\bibnamefont {Cleve}}, \bibinfo
  {author} {\bibfnamefont {R.}~\bibnamefont {Kothari}},\ and\ \bibinfo {author}
  {\bibfnamefont {R.~D.}\ \bibnamefont {Somma}},\ }\bibfield  {title} {\bibinfo
  {title} {Exponential improvement in precision for simulating sparse
  {H}amiltonians},\ }in\ \href {https://doi.org/10.1145/2591796.2591854} {\emph
  {\bibinfo {booktitle} {Proceedings of the Forty-Sixth Annual ACM Symposium on
  Theory of Computing}}},\ \bibinfo {series and number} {STOC '14}\ (\bibinfo
  {publisher} {Association for Computing Machinery},\ \bibinfo {address} {New
  York, NY, USA},\ \bibinfo {year} {2014})\ p.\ \bibinfo {pages}
  {283–292}\BibitemShut {NoStop}%
\bibitem [{\citenamefont {Berry}\ \emph
  {et~al.}(2015{\natexlab{a}})\citenamefont {Berry}, \citenamefont {Childs},
  \citenamefont {Cleve}, \citenamefont {Kothari},\ and\ \citenamefont
  {Somma}}]{berry2015simulating}%
  \BibitemOpen
  \bibfield  {author} {\bibinfo {author} {\bibfnamefont {D.~W.}\ \bibnamefont
  {Berry}}, \bibinfo {author} {\bibfnamefont {A.~M.}\ \bibnamefont {Childs}},
  \bibinfo {author} {\bibfnamefont {R.}~\bibnamefont {Cleve}}, \bibinfo
  {author} {\bibfnamefont {R.}~\bibnamefont {Kothari}},\ and\ \bibinfo {author}
  {\bibfnamefont {R.~D.}\ \bibnamefont {Somma}},\ }\bibfield  {title} {\bibinfo
  {title} {Simulating {H}amiltonian dynamics with a truncated {T}aylor
  series},\ }\href {https://doi.org/10.1103/PhysRevLett.114.090502} {\bibfield
  {journal} {\bibinfo  {journal} {Phys. Rev. Lett.}\ }\textbf {\bibinfo
  {volume} {114}},\ \bibinfo {pages} {090502} (\bibinfo {year}
  {2015}{\natexlab{a}})}\BibitemShut {NoStop}%
\bibitem [{\citenamefont {Campbell}(2019)}]{campbell2019random}%
  \BibitemOpen
  \bibfield  {author} {\bibinfo {author} {\bibfnamefont {E.}~\bibnamefont
  {Campbell}},\ }\bibfield  {title} {\bibinfo {title} {Random compiler for fast
  hamiltonian simulation},\ }\href
  {https://doi.org/10.1103/PhysRevLett.123.070503} {\bibfield  {journal}
  {\bibinfo  {journal} {Phys. Rev. Lett.}\ }\textbf {\bibinfo {volume} {123}},\
  \bibinfo {pages} {070503} (\bibinfo {year} {2019})}\BibitemShut {NoStop}%
\bibitem [{\citenamefont {Low}\ and\ \citenamefont
  {Chuang}(2019)}]{low2019hamiltonian}%
  \BibitemOpen
  \bibfield  {author} {\bibinfo {author} {\bibfnamefont {G.~H.}\ \bibnamefont
  {Low}}\ and\ \bibinfo {author} {\bibfnamefont {I.~L.}\ \bibnamefont
  {Chuang}},\ }\bibfield  {title} {\bibinfo {title} {Hamiltonian simulation by
  qubitization},\ }\href {https://doi.org/10.22331/q-2019-07-12-163} {\bibfield
   {journal} {\bibinfo  {journal} {Quantum}\ }\textbf {\bibinfo {volume} {3}},\
  \bibinfo {pages} {163} (\bibinfo {year} {2019})}\BibitemShut {NoStop}%
\bibitem [{\citenamefont {Breuer}\ and\ \citenamefont
  {Petruccione}(2002)}]{breuer2002theory}%
  \BibitemOpen
  \bibfield  {author} {\bibinfo {author} {\bibfnamefont {H.-P.}\ \bibnamefont
  {Breuer}}\ and\ \bibinfo {author} {\bibfnamefont {F.}~\bibnamefont
  {Petruccione}},\ }\href@noop {} {\emph {\bibinfo {title} {The theory of open
  quantum systems}}}\ (\bibinfo  {publisher} {Oxford University Press, USA},\
  \bibinfo {year} {2002})\BibitemShut {NoStop}%
\bibitem [{\citenamefont {Rivas}\ and\ \citenamefont
  {Huelga}(2012)}]{rivas2012open}%
  \BibitemOpen
  \bibfield  {author} {\bibinfo {author} {\bibfnamefont {A.}~\bibnamefont
  {Rivas}}\ and\ \bibinfo {author} {\bibfnamefont {S.~F.}\ \bibnamefont
  {Huelga}},\ }\href {https://doi.org/10.1007/978-3-642-23354-8} {\emph
  {\bibinfo {title} {Open quantum systems}}},\ Vol.~\bibinfo {volume} {10}\
  (\bibinfo  {publisher} {Springer},\ \bibinfo {year} {2012})\BibitemShut
  {NoStop}%
\bibitem [{\citenamefont {Gorini}\ \emph {et~al.}(1976)\citenamefont {Gorini},
  \citenamefont {Kossakowski},\ and\ \citenamefont {Sudarshan}}]{Gorini_1976}%
  \BibitemOpen
  \bibfield  {author} {\bibinfo {author} {\bibfnamefont {V.}~\bibnamefont
  {Gorini}}, \bibinfo {author} {\bibfnamefont {A.}~\bibnamefont
  {Kossakowski}},\ and\ \bibinfo {author} {\bibfnamefont {E.~C.~G.}\
  \bibnamefont {Sudarshan}},\ }\bibfield  {title} {\bibinfo {title}
  {{Completely Positive Dynamical Semigroups of N Level Systems}},\ }\href
  {https://doi.org/10.1063/1.522979} {\bibfield  {journal} {\bibinfo  {journal}
  {J. Math. Phys.}\ }\textbf {\bibinfo {volume} {17}},\ \bibinfo {pages} {821}
  (\bibinfo {year} {1976})}\BibitemShut {NoStop}%
\bibitem [{\citenamefont {Lindblad}(1976)}]{Lindblad1976}%
  \BibitemOpen
  \bibfield  {author} {\bibinfo {author} {\bibfnamefont {G.}~\bibnamefont
  {Lindblad}},\ }\bibfield  {title} {\bibinfo {title} {On the generators of
  quantum dynamical semigroups},\ }\href {https://doi.org/10.1007/BF01608499}
  {\bibfield  {journal} {\bibinfo  {journal} {Communications in Mathematical
  Physics}\ }\textbf {\bibinfo {volume} {48}},\ \bibinfo {pages} {119–130}
  (\bibinfo {year} {1976})}\BibitemShut {NoStop}%
\bibitem [{\citenamefont {Endo}\ \emph {et~al.}(2018)\citenamefont {Endo},
  \citenamefont {Benjamin},\ and\ \citenamefont {Li}}]{endo2018practical}%
  \BibitemOpen
  \bibfield  {author} {\bibinfo {author} {\bibfnamefont {S.}~\bibnamefont
  {Endo}}, \bibinfo {author} {\bibfnamefont {S.~C.}\ \bibnamefont {Benjamin}},\
  and\ \bibinfo {author} {\bibfnamefont {Y.}~\bibnamefont {Li}},\ }\bibfield
  {title} {\bibinfo {title} {Practical quantum error mitigation for near-future
  applications},\ }\href {https://doi.org/10.1103/PhysRevX.8.031027} {\bibfield
   {journal} {\bibinfo  {journal} {Phys. Rev. X}\ }\textbf {\bibinfo {volume}
  {8}},\ \bibinfo {pages} {031027} (\bibinfo {year} {2018})}\BibitemShut
  {NoStop}%
\bibitem [{\citenamefont {Kandala}\ \emph {et~al.}(2019)\citenamefont
  {Kandala}, \citenamefont {Temme}, \citenamefont {C{\'o}rcoles}, \citenamefont
  {Mezzacapo}, \citenamefont {Chow},\ and\ \citenamefont
  {Gambetta}}]{kandala2019error}%
  \BibitemOpen
  \bibfield  {author} {\bibinfo {author} {\bibfnamefont {A.}~\bibnamefont
  {Kandala}}, \bibinfo {author} {\bibfnamefont {K.}~\bibnamefont {Temme}},
  \bibinfo {author} {\bibfnamefont {A.~D.}\ \bibnamefont {C{\'o}rcoles}},
  \bibinfo {author} {\bibfnamefont {A.}~\bibnamefont {Mezzacapo}}, \bibinfo
  {author} {\bibfnamefont {J.~M.}\ \bibnamefont {Chow}},\ and\ \bibinfo
  {author} {\bibfnamefont {J.~M.}\ \bibnamefont {Gambetta}},\ }\bibfield
  {title} {\bibinfo {title} {Error mitigation extends the computational reach
  of a noisy quantum processor},\ }\href
  {https://doi.org/10.1038/s41586-019-1040-7} {\bibfield  {journal} {\bibinfo
  {journal} {Nature}\ }\textbf {\bibinfo {volume} {567}},\ \bibinfo {pages}
  {491} (\bibinfo {year} {2019})}\BibitemShut {NoStop}%
\bibitem [{\citenamefont {Ding}\ \emph
  {et~al.}(2024{\natexlab{a}})\citenamefont {Ding}, \citenamefont {Chen},\ and\
  \citenamefont {Lin}}]{ding2024single}%
  \BibitemOpen
  \bibfield  {author} {\bibinfo {author} {\bibfnamefont {Z.}~\bibnamefont
  {Ding}}, \bibinfo {author} {\bibfnamefont {C.-F.}\ \bibnamefont {Chen}},\
  and\ \bibinfo {author} {\bibfnamefont {L.}~\bibnamefont {Lin}},\ }\bibfield
  {title} {\bibinfo {title} {Single-ancilla ground state preparation via
  {L}indbladians},\ }\href {https://doi.org/10.1103/PhysRevResearch.6.033147}
  {\bibfield  {journal} {\bibinfo  {journal} {Phys. Rev. Res.}\ }\textbf
  {\bibinfo {volume} {6}},\ \bibinfo {pages} {033147} (\bibinfo {year}
  {2024}{\natexlab{a}})}\BibitemShut {NoStop}%
\bibitem [{\citenamefont {Chen}\ \emph {et~al.}(2023)\citenamefont {Chen},
  \citenamefont {Kastoryano}, \citenamefont {Brand{\~a}o},\ and\ \citenamefont
  {Gily{\'e}n}}]{chen2023quantum}%
  \BibitemOpen
  \bibfield  {author} {\bibinfo {author} {\bibfnamefont {C.-F.}\ \bibnamefont
  {Chen}}, \bibinfo {author} {\bibfnamefont {M.~J.}\ \bibnamefont
  {Kastoryano}}, \bibinfo {author} {\bibfnamefont {F.~G.}\ \bibnamefont
  {Brand{\~a}o}},\ and\ \bibinfo {author} {\bibfnamefont {A.}~\bibnamefont
  {Gily{\'e}n}},\ }\bibfield  {title} {\bibinfo {title} {Quantum thermal state
  preparation},\ }\bibfield  {journal} {\bibinfo  {journal} {arXiv:2303.18224}\
  }\href {https://doi.org/10.48550/arXiv.2303.18224}
  {10.48550/arXiv.2303.18224} (\bibinfo {year} {2023})\BibitemShut {NoStop}%
\bibitem [{\citenamefont {Chakraborty}\ \emph {et~al.}(2019)\citenamefont
  {Chakraborty}, \citenamefont {Gily{\'e}n},\ and\ \citenamefont
  {Jeffery}}]{chakraborty2019power}%
  \BibitemOpen
  \bibfield  {author} {\bibinfo {author} {\bibfnamefont {S.}~\bibnamefont
  {Chakraborty}}, \bibinfo {author} {\bibfnamefont {A.}~\bibnamefont
  {Gily{\'e}n}},\ and\ \bibinfo {author} {\bibfnamefont {S.}~\bibnamefont
  {Jeffery}},\ }\bibfield  {title} {\bibinfo {title} {{The Power of
  Block-Encoded Matrix Powers: Improved Regression Techniques via Faster
  Hamiltonian Simulation}},\ }in\ \href
  {https://doi.org/10.4230/LIPIcs.ICALP.2019.33} {\emph {\bibinfo {booktitle}
  {46th International Colloquium on Automata, Languages, and Programming (ICALP
  2019)}}},\ \bibinfo {series} {Leibniz International Proceedings in
  Informatics (LIPIcs)}, Vol.\ \bibinfo {volume} {132}\ (\bibinfo  {publisher}
  {Schloss Dagstuhl--Leibniz-Zentrum fuer Informatik},\ \bibinfo {address}
  {Dagstuhl, Germany},\ \bibinfo {year} {2019})\ pp.\ \bibinfo {pages}
  {33:1--33:14}\BibitemShut {NoStop}%
\bibitem [{\citenamefont {Cleve}\ and\ \citenamefont
  {Wang}(2017)}]{cleve2017efficient}%
  \BibitemOpen
  \bibfield  {author} {\bibinfo {author} {\bibfnamefont {R.}~\bibnamefont
  {Cleve}}\ and\ \bibinfo {author} {\bibfnamefont {C.}~\bibnamefont {Wang}},\
  }\bibfield  {title} {\bibinfo {title} {Efficient quantum algorithms for
  simulating {L}indblad evolution},\ }in\ \href
  {https://doi.org/10.4230/LIPIcs.ICALP.2017.17} {\emph {\bibinfo {booktitle}
  {44th International Colloquium on Automata, Languages, and Programming (ICALP
  2017)}}},\ Vol.~\bibinfo {volume} {80}\ (\bibinfo {organization} {Schloss
  Dagstuhl--Leibniz-Zentrum f{\"u}r Informatik},\ \bibinfo {year} {2017})\ pp.\
  \bibinfo {pages} {17:1--17:14}\BibitemShut {NoStop}%
\bibitem [{\citenamefont {Childs}\ and\ \citenamefont
  {Li}(2017)}]{childs2017efficient}%
  \BibitemOpen
  \bibfield  {author} {\bibinfo {author} {\bibfnamefont {A.~M.}\ \bibnamefont
  {Childs}}\ and\ \bibinfo {author} {\bibfnamefont {T.}~\bibnamefont {Li}},\
  }\bibfield  {title} {\bibinfo {title} {Efficient simulation of sparse
  {M}arkovian quantum dynamics},\ }\href
  {https://doi.org/10.26421/QIC17.11-12-1} {\bibfield  {journal} {\bibinfo
  {journal} {Quantum Information \& Computation}\ }\textbf {\bibinfo {volume}
  {17}},\ \bibinfo {pages} {901} (\bibinfo {year} {2017})}\BibitemShut
  {NoStop}%
\bibitem [{\citenamefont {Li}\ and\ \citenamefont
  {Wang}(2023{\natexlab{a}})}]{li2023simulating}%
  \BibitemOpen
  \bibfield  {author} {\bibinfo {author} {\bibfnamefont {X.}~\bibnamefont
  {Li}}\ and\ \bibinfo {author} {\bibfnamefont {C.}~\bibnamefont {Wang}},\
  }\bibfield  {title} {\bibinfo {title} {Simulating {M}arkovian open quantum
  systems using higher-order series expansion},\ }in\ \href
  {https://doi.org/10.4230/LIPIcs.ICALP.2023.87} {\emph {\bibinfo {booktitle}
  {50th International Colloquium on Automata, Languages, and Programming (ICALP
  2023)}}},\ Vol.\ \bibinfo {volume} {261}\ (\bibinfo {organization} {Schloss
  Dagstuhl-Leibniz-Zentrum fur Informatik GmbH, Dagstuhl Publishing},\ \bibinfo
  {year} {2023})\ pp.\ \bibinfo {pages} {87:1--87:20}\BibitemShut {NoStop}%
\bibitem [{\citenamefont {Patel}\ and\ \citenamefont
  {Wilde}(2023{\natexlab{a}})}]{patel2023wave1}%
  \BibitemOpen
  \bibfield  {author} {\bibinfo {author} {\bibfnamefont {D.}~\bibnamefont
  {Patel}}\ and\ \bibinfo {author} {\bibfnamefont {M.~M.}\ \bibnamefont
  {Wilde}},\ }\bibfield  {title} {\bibinfo {title} {Wave matrix
  {L}indbladization {I}: Quantum programs for simulating {M}arkovian
  dynamics},\ }\href {https://doi.org/10.1142/S1230161223500105} {\bibfield
  {journal} {\bibinfo  {journal} {Open Systems \& Information Dynamics}\
  }\textbf {\bibinfo {volume} {30}},\ \bibinfo {pages} {2350010} (\bibinfo
  {year} {2023}{\natexlab{a}})}\BibitemShut {NoStop}%
\bibitem [{\citenamefont {Patel}\ and\ \citenamefont
  {Wilde}(2023{\natexlab{b}})}]{patel2023wave2}%
  \BibitemOpen
  \bibfield  {author} {\bibinfo {author} {\bibfnamefont {D.}~\bibnamefont
  {Patel}}\ and\ \bibinfo {author} {\bibfnamefont {M.~M.}\ \bibnamefont
  {Wilde}},\ }\bibfield  {title} {\bibinfo {title} {Wave matrix
  {L}indbladization {II}: General {L}indbladians, linear combinations, and
  polynomials},\ }\href {https://doi.org/10.1142/S1230161223500142} {\bibfield
  {journal} {\bibinfo  {journal} {Open Systems \& Information Dynamics}\
  }\textbf {\bibinfo {volume} {30}},\ \bibinfo {pages} {2350014} (\bibinfo
  {year} {2023}{\natexlab{b}})}\BibitemShut {NoStop}%
\bibitem [{\citenamefont {Katabarwa}\ \emph {et~al.}(2024)\citenamefont
  {Katabarwa}, \citenamefont {Gratsea}, \citenamefont {Caesura},\ and\
  \citenamefont {Johnson}}]{katabarwa2024early}%
  \BibitemOpen
  \bibfield  {author} {\bibinfo {author} {\bibfnamefont {A.}~\bibnamefont
  {Katabarwa}}, \bibinfo {author} {\bibfnamefont {K.}~\bibnamefont {Gratsea}},
  \bibinfo {author} {\bibfnamefont {A.}~\bibnamefont {Caesura}},\ and\ \bibinfo
  {author} {\bibfnamefont {P.~D.}\ \bibnamefont {Johnson}},\ }\bibfield
  {title} {\bibinfo {title} {Early fault-tolerant quantum computing},\ }\href
  {https://doi.org/10.1103/PRXQuantum.5.020101} {\bibfield  {journal} {\bibinfo
   {journal} {PRX Quantum}\ }\textbf {\bibinfo {volume} {5}},\ \bibinfo {pages}
  {020101} (\bibinfo {year} {2024})}\BibitemShut {NoStop}%
\bibitem [{\citenamefont {Bruneau}\ \emph {et~al.}(2014)\citenamefont
  {Bruneau}, \citenamefont {Joye},\ and\ \citenamefont
  {Merkli}}]{bruneau2014repeated}%
  \BibitemOpen
  \bibfield  {author} {\bibinfo {author} {\bibfnamefont {L.}~\bibnamefont
  {Bruneau}}, \bibinfo {author} {\bibfnamefont {A.}~\bibnamefont {Joye}},\ and\
  \bibinfo {author} {\bibfnamefont {M.}~\bibnamefont {Merkli}},\ }\bibfield
  {title} {\bibinfo {title} {Repeated interactions in open quantum systems},\
  }\href {https://doi.org/10.1063/1.4879240} {\bibfield  {journal} {\bibinfo
  {journal} {Journal of Mathematical Physics}\ }\textbf {\bibinfo {volume}
  {55}},\ \bibinfo {pages} {075204} (\bibinfo {year} {2014})}\BibitemShut
  {NoStop}%
\bibitem [{\citenamefont {Ciccarello}\ \emph {et~al.}(2022)\citenamefont
  {Ciccarello}, \citenamefont {Lorenzo}, \citenamefont {Giovannetti},\ and\
  \citenamefont {Palma}}]{ciccarello2022collisionreview}%
  \BibitemOpen
  \bibfield  {author} {\bibinfo {author} {\bibfnamefont {F.}~\bibnamefont
  {Ciccarello}}, \bibinfo {author} {\bibfnamefont {S.}~\bibnamefont {Lorenzo}},
  \bibinfo {author} {\bibfnamefont {V.}~\bibnamefont {Giovannetti}},\ and\
  \bibinfo {author} {\bibfnamefont {G.~M.}\ \bibnamefont {Palma}},\ }\bibfield
  {title} {\bibinfo {title} {Quantum collision models: Open system dynamics
  from repeated interactions},\ }\href
  {https://doi.org/10.1016/j.physrep.2022.01.001} {\bibfield  {journal}
  {\bibinfo  {journal} {Physics Reports}\ }\textbf {\bibinfo {volume} {954}},\
  \bibinfo {pages} {1} (\bibinfo {year} {2022})}\BibitemShut {NoStop}%
\bibitem [{\citenamefont {Ciccarello}\ \emph {et~al.}(2013)\citenamefont
  {Ciccarello}, \citenamefont {Palma},\ and\ \citenamefont
  {Giovannetti}}]{Ciccarello2013PRA}%
  \BibitemOpen
  \bibfield  {author} {\bibinfo {author} {\bibfnamefont {F.}~\bibnamefont
  {Ciccarello}}, \bibinfo {author} {\bibfnamefont {G.~M.}\ \bibnamefont
  {Palma}},\ and\ \bibinfo {author} {\bibfnamefont {V.}~\bibnamefont
  {Giovannetti}},\ }\bibfield  {title} {\bibinfo {title} {Collision-model-based
  approach to non-{M}arkovian quantum dynamics},\ }\href
  {https://doi.org/10.1103/PhysRevA.87.040103} {\bibfield  {journal} {\bibinfo
  {journal} {Phys. Rev. A}\ }\textbf {\bibinfo {volume} {87}},\ \bibinfo
  {pages} {040103} (\bibinfo {year} {2013})}\BibitemShut {NoStop}%
\bibitem [{\citenamefont {Cattaneo}\ \emph {et~al.}(2021)\citenamefont
  {Cattaneo}, \citenamefont {De~Chiara}, \citenamefont {Maniscalco},
  \citenamefont {Zambrini},\ and\ \citenamefont
  {Giorgi}}]{cattaneo2021collision}%
  \BibitemOpen
  \bibfield  {author} {\bibinfo {author} {\bibfnamefont {M.}~\bibnamefont
  {Cattaneo}}, \bibinfo {author} {\bibfnamefont {G.}~\bibnamefont {De~Chiara}},
  \bibinfo {author} {\bibfnamefont {S.}~\bibnamefont {Maniscalco}}, \bibinfo
  {author} {\bibfnamefont {R.}~\bibnamefont {Zambrini}},\ and\ \bibinfo
  {author} {\bibfnamefont {G.~L.}\ \bibnamefont {Giorgi}},\ }\bibfield  {title}
  {\bibinfo {title} {Collision models can efficiently simulate any multipartite
  {M}arkovian quantum dynamics},\ }\href
  {https://doi.org/10.1103/PhysRevLett.126.130403} {\bibfield  {journal}
  {\bibinfo  {journal} {Phys. Rev. Lett.}\ }\textbf {\bibinfo {volume} {126}},\
  \bibinfo {pages} {130403} (\bibinfo {year} {2021})}\BibitemShut {NoStop}%
\bibitem [{\citenamefont {Strasberg}\ \emph {et~al.}(2017)\citenamefont
  {Strasberg}, \citenamefont {Schaller}, \citenamefont {Brandes},\ and\
  \citenamefont {Esposito}}]{strasberg2017quantum}%
  \BibitemOpen
  \bibfield  {author} {\bibinfo {author} {\bibfnamefont {P.}~\bibnamefont
  {Strasberg}}, \bibinfo {author} {\bibfnamefont {G.}~\bibnamefont {Schaller}},
  \bibinfo {author} {\bibfnamefont {T.}~\bibnamefont {Brandes}},\ and\ \bibinfo
  {author} {\bibfnamefont {M.}~\bibnamefont {Esposito}},\ }\bibfield  {title}
  {\bibinfo {title} {Quantum and information thermodynamics: A unifying
  framework based on repeated interactions},\ }\href
  {https://doi.org/10.1103/PhysRevX.7.021003} {\bibfield  {journal} {\bibinfo
  {journal} {Physical Review X}\ }\textbf {\bibinfo {volume} {7}},\ \bibinfo
  {pages} {021003} (\bibinfo {year} {2017})}\BibitemShut {NoStop}%
\bibitem [{\citenamefont {Filipowicz}\ \emph {et~al.}(1986)\citenamefont
  {Filipowicz}, \citenamefont {Javanainen},\ and\ \citenamefont
  {Meystre}}]{filipowicz1986maser}%
  \BibitemOpen
  \bibfield  {author} {\bibinfo {author} {\bibfnamefont {P.}~\bibnamefont
  {Filipowicz}}, \bibinfo {author} {\bibfnamefont {J.}~\bibnamefont
  {Javanainen}},\ and\ \bibinfo {author} {\bibfnamefont {P.}~\bibnamefont
  {Meystre}},\ }\bibfield  {title} {\bibinfo {title} {Theory of a microscopic
  maser},\ }\href {https://doi.org/10.1103/PhysRevA.34.3077} {\bibfield
  {journal} {\bibinfo  {journal} {Phys. Rev. A}\ }\textbf {\bibinfo {volume}
  {34}},\ \bibinfo {pages} {3077} (\bibinfo {year} {1986})}\BibitemShut
  {NoStop}%
\bibitem [{\citenamefont {Barra}(2019)}]{barra2019dissipative}%
  \BibitemOpen
  \bibfield  {author} {\bibinfo {author} {\bibfnamefont {F.}~\bibnamefont
  {Barra}},\ }\bibfield  {title} {\bibinfo {title} {Dissipative charging of a
  quantum battery},\ }\href {https://doi.org/10.1103/PhysRevLett.122.210601}
  {\bibfield  {journal} {\bibinfo  {journal} {Phys. Rev. Lett.}\ }\textbf
  {\bibinfo {volume} {122}},\ \bibinfo {pages} {210601} (\bibinfo {year}
  {2019})}\BibitemShut {NoStop}%
\bibitem [{\citenamefont {Seah}\ \emph {et~al.}(2021)\citenamefont {Seah},
  \citenamefont {Perarnau-Llobet}, \citenamefont {Haack}, \citenamefont
  {Brunner},\ and\ \citenamefont {Nimmrichter}}]{seah2021battery}%
  \BibitemOpen
  \bibfield  {author} {\bibinfo {author} {\bibfnamefont {S.}~\bibnamefont
  {Seah}}, \bibinfo {author} {\bibfnamefont {M.}~\bibnamefont
  {Perarnau-Llobet}}, \bibinfo {author} {\bibfnamefont {G.}~\bibnamefont
  {Haack}}, \bibinfo {author} {\bibfnamefont {N.}~\bibnamefont {Brunner}},\
  and\ \bibinfo {author} {\bibfnamefont {S.}~\bibnamefont {Nimmrichter}},\
  }\bibfield  {title} {\bibinfo {title} {Quantum speed-up in collisional
  battery charging},\ }\href {https://doi.org/10.1103/PhysRevLett.127.100601}
  {\bibfield  {journal} {\bibinfo  {journal} {Phys. Rev. Lett.}\ }\textbf
  {\bibinfo {volume} {127}},\ \bibinfo {pages} {100601} (\bibinfo {year}
  {2021})}\BibitemShut {NoStop}%
\bibitem [{\citenamefont {Lorenzo}\ \emph {et~al.}(2015)\citenamefont
  {Lorenzo}, \citenamefont {McCloskey}, \citenamefont {Ciccarello},
  \citenamefont {Paternostro},\ and\ \citenamefont
  {Palma}}]{lorenzo2015landauer}%
  \BibitemOpen
  \bibfield  {author} {\bibinfo {author} {\bibfnamefont {S.}~\bibnamefont
  {Lorenzo}}, \bibinfo {author} {\bibfnamefont {R.}~\bibnamefont {McCloskey}},
  \bibinfo {author} {\bibfnamefont {F.}~\bibnamefont {Ciccarello}}, \bibinfo
  {author} {\bibfnamefont {M.}~\bibnamefont {Paternostro}},\ and\ \bibinfo
  {author} {\bibfnamefont {G.~M.}\ \bibnamefont {Palma}},\ }\bibfield  {title}
  {\bibinfo {title} {Landauer's principle in multipartite open quantum system
  dynamics},\ }\href {https://doi.org/10.1103/PhysRevLett.115.120403}
  {\bibfield  {journal} {\bibinfo  {journal} {Phys. Rev. Lett.}\ }\textbf
  {\bibinfo {volume} {115}},\ \bibinfo {pages} {120403} (\bibinfo {year}
  {2015})}\BibitemShut {NoStop}%
\bibitem [{\citenamefont {Scarani}\ \emph {et~al.}(2002)\citenamefont
  {Scarani}, \citenamefont {Ziman}, \citenamefont {\ifmmode \check{S}\else
  \v{S}\fi{}telmachovi\ifmmode~\check{c}\else \v{c}\fi{}}, \citenamefont
  {Gisin},\ and\ \citenamefont {Bu\ifmmode~\check{z}\else
  \v{z}\fi{}ek}}]{Scarani2002}%
  \BibitemOpen
  \bibfield  {author} {\bibinfo {author} {\bibfnamefont {V.}~\bibnamefont
  {Scarani}}, \bibinfo {author} {\bibfnamefont {M.}~\bibnamefont {Ziman}},
  \bibinfo {author} {\bibfnamefont {P.}~\bibnamefont {\ifmmode \check{S}\else
  \v{S}\fi{}telmachovi\ifmmode~\check{c}\else \v{c}\fi{}}}, \bibinfo {author}
  {\bibfnamefont {N.}~\bibnamefont {Gisin}},\ and\ \bibinfo {author}
  {\bibfnamefont {V.}~\bibnamefont {Bu\ifmmode~\check{z}\else \v{z}\fi{}ek}},\
  }\bibfield  {title} {\bibinfo {title} {Thermalizing quantum machines:
  Dissipation and entanglement},\ }\href
  {https://doi.org/10.1103/PhysRevLett.88.097905} {\bibfield  {journal}
  {\bibinfo  {journal} {Phys. Rev. Lett.}\ }\textbf {\bibinfo {volume} {88}},\
  \bibinfo {pages} {097905} (\bibinfo {year} {2002})}\BibitemShut {NoStop}%
\bibitem [{\citenamefont {Manatuly}\ \emph {et~al.}(2019)\citenamefont
  {Manatuly}, \citenamefont {Niedenzu}, \citenamefont {Rom\'an-Ancheyta},
  \citenamefont {\ifmmode~\mbox{\c{C}}\else \c{C}\fi{}akmak}, \citenamefont
  {M\"ustecapl\ifmmode \imath \else \i \fi{}o\ifmmode~\breve{g}\else
  \u{g}\fi{}lu},\ and\ \citenamefont {Kurizki}}]{Manatuly2019}%
  \BibitemOpen
  \bibfield  {author} {\bibinfo {author} {\bibfnamefont {A.}~\bibnamefont
  {Manatuly}}, \bibinfo {author} {\bibfnamefont {W.}~\bibnamefont {Niedenzu}},
  \bibinfo {author} {\bibfnamefont {R.}~\bibnamefont {Rom\'an-Ancheyta}},
  \bibinfo {author} {\bibfnamefont {B.}~\bibnamefont
  {\ifmmode~\mbox{\c{C}}\else \c{C}\fi{}akmak}}, \bibinfo {author}
  {\bibfnamefont {O.~E.}\ \bibnamefont {M\"ustecapl\ifmmode \imath \else \i
  \fi{}o\ifmmode~\breve{g}\else \u{g}\fi{}lu}},\ and\ \bibinfo {author}
  {\bibfnamefont {G.}~\bibnamefont {Kurizki}},\ }\bibfield  {title} {\bibinfo
  {title} {Collectively enhanced thermalization via multiqubit collisions},\
  }\href {https://doi.org/10.1103/PhysRevE.99.042145} {\bibfield  {journal}
  {\bibinfo  {journal} {Phys. Rev. E}\ }\textbf {\bibinfo {volume} {99}},\
  \bibinfo {pages} {042145} (\bibinfo {year} {2019})}\BibitemShut {NoStop}%
\bibitem [{\citenamefont {Rodrigues}\ \emph {et~al.}(2019)\citenamefont
  {Rodrigues}, \citenamefont {De~Chiara}, \citenamefont {Paternostro},\ and\
  \citenamefont {Landi}}]{Rodrigues2019}%
  \BibitemOpen
  \bibfield  {author} {\bibinfo {author} {\bibfnamefont {F.~L.~S.}\
  \bibnamefont {Rodrigues}}, \bibinfo {author} {\bibfnamefont {G.}~\bibnamefont
  {De~Chiara}}, \bibinfo {author} {\bibfnamefont {M.}~\bibnamefont
  {Paternostro}},\ and\ \bibinfo {author} {\bibfnamefont {G.~T.}\ \bibnamefont
  {Landi}},\ }\bibfield  {title} {\bibinfo {title} {Thermodynamics of weakly
  coherent collisional models},\ }\href
  {https://doi.org/10.1103/PhysRevLett.123.140601} {\bibfield  {journal}
  {\bibinfo  {journal} {Phys. Rev. Lett.}\ }\textbf {\bibinfo {volume} {123}},\
  \bibinfo {pages} {140601} (\bibinfo {year} {2019})}\BibitemShut {NoStop}%
\bibitem [{\citenamefont {Leitch}\ \emph {et~al.}(2022)\citenamefont {Leitch},
  \citenamefont {Piccione}, \citenamefont {Bellomo},\ and\ \citenamefont
  {De~Chiara}}]{Leitch2022}%
  \BibitemOpen
  \bibfield  {author} {\bibinfo {author} {\bibfnamefont {H.}~\bibnamefont
  {Leitch}}, \bibinfo {author} {\bibfnamefont {N.}~\bibnamefont {Piccione}},
  \bibinfo {author} {\bibfnamefont {B.}~\bibnamefont {Bellomo}},\ and\ \bibinfo
  {author} {\bibfnamefont {G.}~\bibnamefont {De~Chiara}},\ }\bibfield  {title}
  {\bibinfo {title} {Driven quantum harmonic oscillators: A working medium for
  thermal machines},\ }\href {https://doi.org/10.1116/5.0072067} {\bibfield
  {journal} {\bibinfo  {journal} {AVS Quantum Science}\ }\textbf {\bibinfo
  {volume} {4}},\ \bibinfo {pages} {012001} (\bibinfo {year}
  {2022})}\BibitemShut {NoStop}%
\bibitem [{\citenamefont {Grimmer}\ \emph {et~al.}(2018)\citenamefont
  {Grimmer}, \citenamefont {Brown}, \citenamefont {Kempf}, \citenamefont
  {Mann},\ and\ \citenamefont {Mart\'{\i}n-Mart\'{\i}nez}}]{Grimmer2018}%
  \BibitemOpen
  \bibfield  {author} {\bibinfo {author} {\bibfnamefont {D.}~\bibnamefont
  {Grimmer}}, \bibinfo {author} {\bibfnamefont {E.}~\bibnamefont {Brown}},
  \bibinfo {author} {\bibfnamefont {A.}~\bibnamefont {Kempf}}, \bibinfo
  {author} {\bibfnamefont {R.~B.}\ \bibnamefont {Mann}},\ and\ \bibinfo
  {author} {\bibfnamefont {E.}~\bibnamefont {Mart\'{\i}n-Mart\'{\i}nez}},\
  }\bibfield  {title} {\bibinfo {title} {Gaussian ancillary bombardment},\
  }\href {https://doi.org/10.1103/PhysRevA.97.052120} {\bibfield  {journal}
  {\bibinfo  {journal} {Phys. Rev. A}\ }\textbf {\bibinfo {volume} {97}},\
  \bibinfo {pages} {052120} (\bibinfo {year} {2018})}\BibitemShut {NoStop}%
\bibitem [{\citenamefont {Hammam}\ \emph {et~al.}(2022)\citenamefont {Hammam},
  \citenamefont {Leitch}, \citenamefont {Hassouni},\ and\ \citenamefont
  {De~Chiara}}]{Hammam_2022}%
  \BibitemOpen
  \bibfield  {author} {\bibinfo {author} {\bibfnamefont {K.}~\bibnamefont
  {Hammam}}, \bibinfo {author} {\bibfnamefont {H.}~\bibnamefont {Leitch}},
  \bibinfo {author} {\bibfnamefont {Y.}~\bibnamefont {Hassouni}},\ and\
  \bibinfo {author} {\bibfnamefont {G.}~\bibnamefont {De~Chiara}},\ }\bibfield
  {title} {\bibinfo {title} {Exploiting coherence for quantum thermodynamic
  advantage},\ }\href {https://doi.org/10.1088/1367-2630/aca49b} {\bibfield
  {journal} {\bibinfo  {journal} {New Journal of Physics}\ }\textbf {\bibinfo
  {volume} {24}},\ \bibinfo {pages} {113053} (\bibinfo {year}
  {2022})}\BibitemShut {NoStop}%
\bibitem [{\citenamefont {Ciccarello}(2017)}]{ciccarello2017optics}%
  \BibitemOpen
  \bibfield  {author} {\bibinfo {author} {\bibfnamefont {F.}~\bibnamefont
  {Ciccarello}},\ }\bibfield  {title} {\bibinfo {title} {Collision models in
  quantum optics},\ }\href {https://doi.org/doi:10.1515/qmetro-2017-0007}
  {\bibfield  {journal} {\bibinfo  {journal} {Quantum Measurements and Quantum
  Metrology}\ }\textbf {\bibinfo {volume} {4}},\ \bibinfo {pages} {53}
  (\bibinfo {year} {2017})}\BibitemShut {NoStop}%
\bibitem [{\citenamefont {Grimsmo}(2015)}]{Grimsmo2015}%
  \BibitemOpen
  \bibfield  {author} {\bibinfo {author} {\bibfnamefont {A.~L.}\ \bibnamefont
  {Grimsmo}},\ }\bibfield  {title} {\bibinfo {title} {Time-delayed quantum
  feedback control},\ }\href {https://doi.org/10.1103/PhysRevLett.115.060402}
  {\bibfield  {journal} {\bibinfo  {journal} {Phys. Rev. Lett.}\ }\textbf
  {\bibinfo {volume} {115}},\ \bibinfo {pages} {060402} (\bibinfo {year}
  {2015})}\BibitemShut {NoStop}%
\bibitem [{\citenamefont {Whalen}\ \emph {et~al.}(2017)\citenamefont {Whalen},
  \citenamefont {Grimsmo},\ and\ \citenamefont {Carmichael}}]{Whalen2017}%
  \BibitemOpen
  \bibfield  {author} {\bibinfo {author} {\bibfnamefont {S.~J.}\ \bibnamefont
  {Whalen}}, \bibinfo {author} {\bibfnamefont {A.~L.}\ \bibnamefont
  {Grimsmo}},\ and\ \bibinfo {author} {\bibfnamefont {H.~J.}\ \bibnamefont
  {Carmichael}},\ }\bibfield  {title} {\bibinfo {title} {Open quantum systems
  with delayed coherent feedback},\ }\href
  {https://doi.org/10.1088/2058-9565/aa8331} {\bibfield  {journal} {\bibinfo
  {journal} {Quantum Science and Technology}\ }\textbf {\bibinfo {volume}
  {2}},\ \bibinfo {pages} {044008} (\bibinfo {year} {2017})}\BibitemShut
  {NoStop}%
\bibitem [{\citenamefont {Pichler}\ and\ \citenamefont
  {Zoller}(2016)}]{Pichler2016}%
  \BibitemOpen
  \bibfield  {author} {\bibinfo {author} {\bibfnamefont {H.}~\bibnamefont
  {Pichler}}\ and\ \bibinfo {author} {\bibfnamefont {P.}~\bibnamefont
  {Zoller}},\ }\bibfield  {title} {\bibinfo {title} {Photonic circuits with
  time delays and quantum feedback},\ }\href
  {https://doi.org/10.1103/PhysRevLett.116.093601} {\bibfield  {journal}
  {\bibinfo  {journal} {Phys. Rev. Lett.}\ }\textbf {\bibinfo {volume} {116}},\
  \bibinfo {pages} {093601} (\bibinfo {year} {2016})}\BibitemShut {NoStop}%
\bibitem [{\citenamefont {Fischer}\ \emph {et~al.}(2018)\citenamefont
  {Fischer}, \citenamefont {Trivedi}, \citenamefont {Ramasesh}, \citenamefont
  {Siddiqi},\ and\ \citenamefont {Vučković}}]{Fischer2018}%
  \BibitemOpen
  \bibfield  {author} {\bibinfo {author} {\bibfnamefont {K.~A.}\ \bibnamefont
  {Fischer}}, \bibinfo {author} {\bibfnamefont {R.}~\bibnamefont {Trivedi}},
  \bibinfo {author} {\bibfnamefont {V.}~\bibnamefont {Ramasesh}}, \bibinfo
  {author} {\bibfnamefont {I.}~\bibnamefont {Siddiqi}},\ and\ \bibinfo {author}
  {\bibfnamefont {J.}~\bibnamefont {Vučković}},\ }\bibfield  {title}
  {\bibinfo {title} {Scattering into one-dimensional waveguides from a
  coherently-driven quantum-optical system},\ }\href
  {https://doi.org/10.22331/q-2018-05-28-69} {\bibfield  {journal} {\bibinfo
  {journal} {Quantum}\ }\textbf {\bibinfo {volume} {2}},\ \bibinfo {pages} {69}
  (\bibinfo {year} {2018})}\BibitemShut {NoStop}%
\bibitem [{\citenamefont {Fischer}(2018)}]{Fischer_2018_JOPC}%
  \BibitemOpen
  \bibfield  {author} {\bibinfo {author} {\bibfnamefont {K.}~\bibnamefont
  {Fischer}},\ }\bibfield  {title} {\bibinfo {title} {Derivation of the
  quantum-optical master equation based on coarse-graining of time},\ }\href
  {https://doi.org/10.1088/2399-6528/aadaf8} {\bibfield  {journal} {\bibinfo
  {journal} {Journal of Physics Communications}\ }\textbf {\bibinfo {volume}
  {2}},\ \bibinfo {pages} {091001} (\bibinfo {year} {2018})}\BibitemShut
  {NoStop}%
\bibitem [{\citenamefont {Cilluffo}\ \emph {et~al.}(2020)\citenamefont
  {Cilluffo}, \citenamefont {Carollo}, \citenamefont {Lorenzo}, \citenamefont
  {Gross}, \citenamefont {Palma},\ and\ \citenamefont
  {Ciccarello}}]{ciluffo2020collisional}%
  \BibitemOpen
  \bibfield  {author} {\bibinfo {author} {\bibfnamefont {D.}~\bibnamefont
  {Cilluffo}}, \bibinfo {author} {\bibfnamefont {A.}~\bibnamefont {Carollo}},
  \bibinfo {author} {\bibfnamefont {S.}~\bibnamefont {Lorenzo}}, \bibinfo
  {author} {\bibfnamefont {J.~A.}\ \bibnamefont {Gross}}, \bibinfo {author}
  {\bibfnamefont {G.~M.}\ \bibnamefont {Palma}},\ and\ \bibinfo {author}
  {\bibfnamefont {F.}~\bibnamefont {Ciccarello}},\ }\bibfield  {title}
  {\bibinfo {title} {Collisional picture of quantum optics with giant
  emitters},\ }\href {https://doi.org/10.1103/PhysRevResearch.2.043070}
  {\bibfield  {journal} {\bibinfo  {journal} {Phys. Rev. Res.}\ }\textbf
  {\bibinfo {volume} {2}},\ \bibinfo {pages} {043070} (\bibinfo {year}
  {2020})}\BibitemShut {NoStop}%
\bibitem [{\citenamefont {Gross}\ \emph {et~al.}(2018)\citenamefont {Gross},
  \citenamefont {Caves}, \citenamefont {Milburn},\ and\ \citenamefont
  {Combes}}]{Gross_2018}%
  \BibitemOpen
  \bibfield  {author} {\bibinfo {author} {\bibfnamefont {J.~A.}\ \bibnamefont
  {Gross}}, \bibinfo {author} {\bibfnamefont {C.~M.}\ \bibnamefont {Caves}},
  \bibinfo {author} {\bibfnamefont {G.~J.}\ \bibnamefont {Milburn}},\ and\
  \bibinfo {author} {\bibfnamefont {J.}~\bibnamefont {Combes}},\ }\bibfield
  {title} {\bibinfo {title} {Qubit models of weak continuous measurements:
  {M}arkovian conditional and open-system dynamics},\ }\href
  {https://doi.org/10.1088/2058-9565/aaa39f} {\bibfield  {journal} {\bibinfo
  {journal} {Quantum Science and Technology}\ }\textbf {\bibinfo {volume}
  {3}},\ \bibinfo {pages} {024005} (\bibinfo {year} {2018})}\BibitemShut
  {NoStop}%
\bibitem [{\citenamefont {Seah}\ \emph {et~al.}(2019)\citenamefont {Seah},
  \citenamefont {Nimmrichter}, \citenamefont {Grimmer}, \citenamefont {Santos},
  \citenamefont {Scarani},\ and\ \citenamefont {Landi}}]{seah2019collisional}%
  \BibitemOpen
  \bibfield  {author} {\bibinfo {author} {\bibfnamefont {S.}~\bibnamefont
  {Seah}}, \bibinfo {author} {\bibfnamefont {S.}~\bibnamefont {Nimmrichter}},
  \bibinfo {author} {\bibfnamefont {D.}~\bibnamefont {Grimmer}}, \bibinfo
  {author} {\bibfnamefont {J.~P.}\ \bibnamefont {Santos}}, \bibinfo {author}
  {\bibfnamefont {V.}~\bibnamefont {Scarani}},\ and\ \bibinfo {author}
  {\bibfnamefont {G.~T.}\ \bibnamefont {Landi}},\ }\bibfield  {title} {\bibinfo
  {title} {Collisional quantum thermometry},\ }\href
  {https://doi.org/10.1103/PhysRevLett.123.180602} {\bibfield  {journal}
  {\bibinfo  {journal} {Phys. Rev. Lett.}\ }\textbf {\bibinfo {volume} {123}},\
  \bibinfo {pages} {180602} (\bibinfo {year} {2019})}\BibitemShut {NoStop}%
\bibitem [{\citenamefont {Shu}\ \emph {et~al.}(2020)\citenamefont {Shu},
  \citenamefont {Seah},\ and\ \citenamefont {Scarani}}]{shu2020surpassing}%
  \BibitemOpen
  \bibfield  {author} {\bibinfo {author} {\bibfnamefont {A.}~\bibnamefont
  {Shu}}, \bibinfo {author} {\bibfnamefont {S.}~\bibnamefont {Seah}},\ and\
  \bibinfo {author} {\bibfnamefont {V.}~\bibnamefont {Scarani}},\ }\bibfield
  {title} {\bibinfo {title} {Surpassing the thermal cram\'er-rao bound with
  collisional thermometry},\ }\href
  {https://doi.org/10.1103/PhysRevA.102.042417} {\bibfield  {journal} {\bibinfo
   {journal} {Phys. Rev. A}\ }\textbf {\bibinfo {volume} {102}},\ \bibinfo
  {pages} {042417} (\bibinfo {year} {2020})}\BibitemShut {NoStop}%
\bibitem [{\citenamefont {Pellegrini}\ and\ \citenamefont
  {Petruccione}(2009)}]{Pellegrini2009-eh}%
  \BibitemOpen
  \bibfield  {author} {\bibinfo {author} {\bibfnamefont {C.}~\bibnamefont
  {Pellegrini}}\ and\ \bibinfo {author} {\bibfnamefont {F.}~\bibnamefont
  {Petruccione}},\ }\bibfield  {title} {\bibinfo {title} {{Non-Markovian}
  quantum repeated interactions and measurements},\ }\href
  {https://doi.org/10.1088/1751-8113/42/42/425304} {\bibfield  {journal}
  {\bibinfo  {journal} {J. Phys. A Math. Theor.}\ }\textbf {\bibinfo {volume}
  {42}},\ \bibinfo {pages} {425304} (\bibinfo {year} {2009})}\BibitemShut
  {NoStop}%
\bibitem [{\citenamefont {Man}\ \emph {et~al.}(2019)\citenamefont {Man},
  \citenamefont {Xia},\ and\ \citenamefont {Lo~Franco}}]{Man2019}%
  \BibitemOpen
  \bibfield  {author} {\bibinfo {author} {\bibfnamefont {Z.-X.}\ \bibnamefont
  {Man}}, \bibinfo {author} {\bibfnamefont {Y.-J.}\ \bibnamefont {Xia}},\ and\
  \bibinfo {author} {\bibfnamefont {R.}~\bibnamefont {Lo~Franco}},\ }\bibfield
  {title} {\bibinfo {title} {Validity of the {L}andauer principle and quantum
  memory effects via collisional models},\ }\href
  {https://doi.org/10.1103/PhysRevA.99.042106} {\bibfield  {journal} {\bibinfo
  {journal} {Phys. Rev. A}\ }\textbf {\bibinfo {volume} {99}},\ \bibinfo
  {pages} {042106} (\bibinfo {year} {2019})}\BibitemShut {NoStop}%
\bibitem [{\citenamefont {Lorenzo}\ \emph {et~al.}(2017)\citenamefont
  {Lorenzo}, \citenamefont {Ciccarello}, \citenamefont {Palma},\ and\
  \citenamefont {Vacchini}}]{Lorenzo2017-je}%
  \BibitemOpen
  \bibfield  {author} {\bibinfo {author} {\bibfnamefont {S.}~\bibnamefont
  {Lorenzo}}, \bibinfo {author} {\bibfnamefont {F.}~\bibnamefont {Ciccarello}},
  \bibinfo {author} {\bibfnamefont {G.~M.}\ \bibnamefont {Palma}},\ and\
  \bibinfo {author} {\bibfnamefont {B.}~\bibnamefont {Vacchini}},\ }\bibfield
  {title} {\bibinfo {title} {Quantum non-{M}arkovian piecewise dynamics from
  collision models},\ }\href {https://doi.org/10.1142/S123016121740011X}
  {\bibfield  {journal} {\bibinfo  {journal} {Open Syst. Inf. Dyn.}\ }\textbf
  {\bibinfo {volume} {24}},\ \bibinfo {pages} {1740011} (\bibinfo {year}
  {2017})}\BibitemShut {NoStop}%
\bibitem [{\citenamefont {Lorenzo}\ \emph {et~al.}(2016)\citenamefont
  {Lorenzo}, \citenamefont {Ciccarello},\ and\ \citenamefont
  {Palma}}]{Lorenzo2016}%
  \BibitemOpen
  \bibfield  {author} {\bibinfo {author} {\bibfnamefont {S.}~\bibnamefont
  {Lorenzo}}, \bibinfo {author} {\bibfnamefont {F.}~\bibnamefont
  {Ciccarello}},\ and\ \bibinfo {author} {\bibfnamefont {G.~M.}\ \bibnamefont
  {Palma}},\ }\bibfield  {title} {\bibinfo {title} {Class of exact
  memory-kernel master equations},\ }\href
  {https://doi.org/10.1103/PhysRevA.93.052111} {\bibfield  {journal} {\bibinfo
  {journal} {Phys. Rev. A}\ }\textbf {\bibinfo {volume} {93}},\ \bibinfo
  {pages} {052111} (\bibinfo {year} {2016})}\BibitemShut {NoStop}%
\bibitem [{\citenamefont {Berry}\ \emph
  {et~al.}(2015{\natexlab{b}})\citenamefont {Berry}, \citenamefont {Childs},\
  and\ \citenamefont {Kothari}}]{berry2015hamiltonian}%
  \BibitemOpen
  \bibfield  {author} {\bibinfo {author} {\bibfnamefont {D.~W.}\ \bibnamefont
  {Berry}}, \bibinfo {author} {\bibfnamefont {A.~M.}\ \bibnamefont {Childs}},\
  and\ \bibinfo {author} {\bibfnamefont {R.}~\bibnamefont {Kothari}},\
  }\bibfield  {title} {\bibinfo {title} {Hamiltonian simulation with nearly
  optimal dependence on all parameters},\ }in\ \href
  {https://doi.org/10.1109/FOCS.2015.54} {\emph {\bibinfo {booktitle} {2015
  IEEE 56th Annual Symposium on Foundations of Computer Science}}}\ (\bibinfo
  {year} {2015})\ pp.\ \bibinfo {pages} {792--809}\BibitemShut {NoStop}%
\bibitem [{\citenamefont {Low}\ and\ \citenamefont
  {Chuang}(2017)}]{low2017optimal}%
  \BibitemOpen
  \bibfield  {author} {\bibinfo {author} {\bibfnamefont {G.~H.}\ \bibnamefont
  {Low}}\ and\ \bibinfo {author} {\bibfnamefont {I.~L.}\ \bibnamefont
  {Chuang}},\ }\bibfield  {title} {\bibinfo {title} {Optimal {H}amiltonian
  simulation by quantum signal processing},\ }\href
  {https://doi.org/10.1103/PhysRevLett.118.010501} {\bibfield  {journal}
  {\bibinfo  {journal} {Phys. Rev. Lett.}\ }\textbf {\bibinfo {volume} {118}},\
  \bibinfo {pages} {010501} (\bibinfo {year} {2017})}\BibitemShut {NoStop}%
\bibitem [{\citenamefont {Childs}\ \emph {et~al.}(2018)\citenamefont {Childs},
  \citenamefont {Maslov}, \citenamefont {Nam}, \citenamefont {Ross},\ and\
  \citenamefont {Su}}]{childs2018toward}%
  \BibitemOpen
  \bibfield  {author} {\bibinfo {author} {\bibfnamefont {A.~M.}\ \bibnamefont
  {Childs}}, \bibinfo {author} {\bibfnamefont {D.}~\bibnamefont {Maslov}},
  \bibinfo {author} {\bibfnamefont {Y.}~\bibnamefont {Nam}}, \bibinfo {author}
  {\bibfnamefont {N.~J.}\ \bibnamefont {Ross}},\ and\ \bibinfo {author}
  {\bibfnamefont {Y.}~\bibnamefont {Su}},\ }\bibfield  {title} {\bibinfo
  {title} {Toward the first quantum simulation with quantum speedup},\ }\href
  {https://doi.org/10.1073/pnas.1801723115} {\bibfield  {journal} {\bibinfo
  {journal} {Proceedings of the National Academy of Sciences}\ }\textbf
  {\bibinfo {volume} {115}},\ \bibinfo {pages} {9456} (\bibinfo {year}
  {2018})}\BibitemShut {NoStop}%
\bibitem [{\citenamefont {Childs}\ \emph {et~al.}(2021)\citenamefont {Childs},
  \citenamefont {Su}, \citenamefont {Tran}, \citenamefont {Wiebe},\ and\
  \citenamefont {Zhu}}]{childs2021theory}%
  \BibitemOpen
  \bibfield  {author} {\bibinfo {author} {\bibfnamefont {A.~M.}\ \bibnamefont
  {Childs}}, \bibinfo {author} {\bibfnamefont {Y.}~\bibnamefont {Su}}, \bibinfo
  {author} {\bibfnamefont {M.~C.}\ \bibnamefont {Tran}}, \bibinfo {author}
  {\bibfnamefont {N.}~\bibnamefont {Wiebe}},\ and\ \bibinfo {author}
  {\bibfnamefont {S.}~\bibnamefont {Zhu}},\ }\bibfield  {title} {\bibinfo
  {title} {Theory of {T}rotter error with commutator scaling},\ }\href
  {https://doi.org/10.1103/PhysRevX.11.011020} {\bibfield  {journal} {\bibinfo
  {journal} {Physical Review X}\ }\textbf {\bibinfo {volume} {11}},\ \bibinfo
  {pages} {011020} (\bibinfo {year} {2021})}\BibitemShut {NoStop}%
\bibitem [{\citenamefont {Ding}\ \emph
  {et~al.}(2024{\natexlab{b}})\citenamefont {Ding}, \citenamefont {Li},\ and\
  \citenamefont {Lin}}]{ding2024simulating}%
  \BibitemOpen
  \bibfield  {author} {\bibinfo {author} {\bibfnamefont {Z.}~\bibnamefont
  {Ding}}, \bibinfo {author} {\bibfnamefont {X.}~\bibnamefont {Li}},\ and\
  \bibinfo {author} {\bibfnamefont {L.}~\bibnamefont {Lin}},\ }\bibfield
  {title} {\bibinfo {title} {Simulating open quantum systems using
  {H}amiltonian simulations},\ }\href
  {https://doi.org/10.1103/PRXQuantum.5.020332} {\bibfield  {journal} {\bibinfo
   {journal} {PRX Quantum}\ }\textbf {\bibinfo {volume} {5}},\ \bibinfo {pages}
  {020332} (\bibinfo {year} {2024}{\natexlab{b}})}\BibitemShut {NoStop}%
\bibitem [{\citenamefont {Pocrnic}\ \emph {et~al.}(2023)\citenamefont
  {Pocrnic}, \citenamefont {Segal},\ and\ \citenamefont
  {Wiebe}}]{pocrnic2023quantum}%
  \BibitemOpen
  \bibfield  {author} {\bibinfo {author} {\bibfnamefont {M.}~\bibnamefont
  {Pocrnic}}, \bibinfo {author} {\bibfnamefont {D.}~\bibnamefont {Segal}},\
  and\ \bibinfo {author} {\bibfnamefont {N.}~\bibnamefont {Wiebe}},\ }\bibfield
   {title} {\bibinfo {title} {Quantum simulation of {L}indbladian dynamics via
  repeated interactions},\ }\bibfield  {journal} {\bibinfo  {journal}
  {arXiv:2312.05371}\ }\href {https://doi.org/10.48550/arXiv.2312.05371}
  {10.48550/arXiv.2312.05371} (\bibinfo {year} {2023})\BibitemShut {NoStop}%
\bibitem [{\citenamefont {Wan}\ \emph {et~al.}(2022)\citenamefont {Wan},
  \citenamefont {Berta},\ and\ \citenamefont {Campbell}}]{wan2022randomized}%
  \BibitemOpen
  \bibfield  {author} {\bibinfo {author} {\bibfnamefont {K.}~\bibnamefont
  {Wan}}, \bibinfo {author} {\bibfnamefont {M.}~\bibnamefont {Berta}},\ and\
  \bibinfo {author} {\bibfnamefont {E.~T.}\ \bibnamefont {Campbell}},\
  }\bibfield  {title} {\bibinfo {title} {Randomized quantum algorithm for
  statistical phase estimation},\ }\href
  {https://doi.org/10.1103/PhysRevLett.129.030503} {\bibfield  {journal}
  {\bibinfo  {journal} {Phys. Rev. Lett.}\ }\textbf {\bibinfo {volume} {129}},\
  \bibinfo {pages} {030503} (\bibinfo {year} {2022})}\BibitemShut {NoStop}%
\bibitem [{\citenamefont {Wang}\ \emph {et~al.}(2024)\citenamefont {Wang},
  \citenamefont {McArdle},\ and\ \citenamefont {Berta}}]{wang2023qubit}%
  \BibitemOpen
  \bibfield  {author} {\bibinfo {author} {\bibfnamefont {S.}~\bibnamefont
  {Wang}}, \bibinfo {author} {\bibfnamefont {S.}~\bibnamefont {McArdle}},\ and\
  \bibinfo {author} {\bibfnamefont {M.}~\bibnamefont {Berta}},\ }\bibfield
  {title} {\bibinfo {title} {Qubit-efficient randomized quantum algorithms for
  linear algebra},\ }\href {https://doi.org/10.1103/PRXQuantum.5.020324}
  {\bibfield  {journal} {\bibinfo  {journal} {PRX Quantum}\ }\textbf {\bibinfo
  {volume} {5}},\ \bibinfo {pages} {020324} (\bibinfo {year}
  {2024})}\BibitemShut {NoStop}%
\bibitem [{\citenamefont {Chakraborty}(2024)}]{chakraborty2023implementing}%
  \BibitemOpen
  \bibfield  {author} {\bibinfo {author} {\bibfnamefont {S.}~\bibnamefont
  {Chakraborty}},\ }\bibfield  {title} {\bibinfo {title} {Implementing any
  linear combination of unitaries on intermediate-term quantum computers},\
  }\href {https://doi.org/10.22331/q-2024-10-10-1496} {\bibfield  {journal}
  {\bibinfo  {journal} {Quantum}\ }\textbf {\bibinfo {volume} {8}},\ \bibinfo
  {pages} {1496} (\bibinfo {year} {2024})}\BibitemShut {NoStop}%
\bibitem [{\citenamefont {Aaronson}\ and\ \citenamefont
  {Rall}(2020)}]{aaronson2020quantum}%
  \BibitemOpen
  \bibfield  {author} {\bibinfo {author} {\bibfnamefont {S.}~\bibnamefont
  {Aaronson}}\ and\ \bibinfo {author} {\bibfnamefont {P.}~\bibnamefont
  {Rall}},\ }\bibfield  {title} {\bibinfo {title} {Quantum approximate
  counting, simplified},\ }in\ \href
  {https://doi.org/10.1137/1.9781611976014.5} {\emph {\bibinfo {booktitle}
  {Symposium on simplicity in algorithms}}}\ (\bibinfo {organization} {SIAM},\
  \bibinfo {year} {2020})\ pp.\ \bibinfo {pages} {24--32}\BibitemShut {NoStop}%
\bibitem [{\citenamefont {Grinko}\ \emph {et~al.}(2021)\citenamefont {Grinko},
  \citenamefont {Gacon}, \citenamefont {Zoufal},\ and\ \citenamefont
  {Woerner}}]{grinko2021iterative}%
  \BibitemOpen
  \bibfield  {author} {\bibinfo {author} {\bibfnamefont {D.}~\bibnamefont
  {Grinko}}, \bibinfo {author} {\bibfnamefont {J.}~\bibnamefont {Gacon}},
  \bibinfo {author} {\bibfnamefont {C.}~\bibnamefont {Zoufal}},\ and\ \bibinfo
  {author} {\bibfnamefont {S.}~\bibnamefont {Woerner}},\ }\bibfield  {title}
  {\bibinfo {title} {Iterative quantum amplitude estimation},\ }\href
  {https://doi.org/10.1038/s41534-021-00379-1} {\bibfield  {journal} {\bibinfo
  {journal} {npj Quantum Information}\ }\textbf {\bibinfo {volume} {7}},\
  \bibinfo {pages} {52} (\bibinfo {year} {2021})}\BibitemShut {NoStop}%
\bibitem [{\citenamefont {Borras}\ and\ \citenamefont
  {Marvian}(2025)}]{borras2025simulatelindblad}%
  \BibitemOpen
  \bibfield  {author} {\bibinfo {author} {\bibfnamefont {E.}~\bibnamefont
  {Borras}}\ and\ \bibinfo {author} {\bibfnamefont {M.}~\bibnamefont
  {Marvian}},\ }\bibfield  {title} {\bibinfo {title} {Quantum algorithm to
  simulate lindblad master equations},\ }\href
  {https://doi.org/10.1103/PhysRevResearch.7.023076} {\bibfield  {journal}
  {\bibinfo  {journal} {Phys. Rev. Res.}\ }\textbf {\bibinfo {volume} {7}},\
  \bibinfo {pages} {023076} (\bibinfo {year} {2025})}\BibitemShut {NoStop}%
\bibitem [{\citenamefont {Sachdev}(2011)}]{sachdev2011quantum}%
  \BibitemOpen
  \bibfield  {author} {\bibinfo {author} {\bibfnamefont {S.}~\bibnamefont
  {Sachdev}},\ }\href {https://doi.org/10.1017/CBO9780511973765} {\emph
  {\bibinfo {title} {Quantum Phase Transitions}}},\ \bibinfo {edition} {2nd}\
  ed.\ (\bibinfo  {publisher} {Cambridge University Press},\ \bibinfo {year}
  {2011})\BibitemShut {NoStop}%
\bibitem [{\citenamefont {Watson}\ and\ \citenamefont
  {Watkins}(2024)}]{watson2024exponentiallyreducedcircuitdepths}%
  \BibitemOpen
  \bibfield  {author} {\bibinfo {author} {\bibfnamefont {J.~D.}\ \bibnamefont
  {Watson}}\ and\ \bibinfo {author} {\bibfnamefont {J.}~\bibnamefont
  {Watkins}},\ }\bibfield  {title} {\bibinfo {title} {Exponentially reduced
  circuit depths using {T}rotter error mitigation},\ }\bibfield  {journal}
  {\bibinfo  {journal} {arXiv:2408.14385}\ }\href
  {https://doi.org/10.48550/arXiv.2408.14385} {10.48550/arXiv.2408.14385}
  (\bibinfo {year} {2024})\BibitemShut {NoStop}%
\bibitem [{\citenamefont {Watson}(2024)}]{watson2024randomly}%
  \BibitemOpen
  \bibfield  {author} {\bibinfo {author} {\bibfnamefont {J.~D.}\ \bibnamefont
  {Watson}},\ }\bibfield  {title} {\bibinfo {title} {Randomly compiled quantum
  simulation with exponentially reduced circuit depths},\ }\bibfield  {journal}
  {\bibinfo  {journal} {arXiv preprint arXiv:2411.04240}\ }\href
  {https://doi.org/10.48550/arXiv.2411.04240} {10.48550/arXiv.2411.04240}
  (\bibinfo {year} {2024})\BibitemShut {NoStop}%
\bibitem [{\citenamefont {Chakraborty}\ \emph {et~al.}(2025)\citenamefont
  {Chakraborty}, \citenamefont {Hazra}, \citenamefont {Li}, \citenamefont
  {Shao}, \citenamefont {Wang},\ and\ \citenamefont
  {Zhang}}]{chakraborty2025quantum}%
  \BibitemOpen
  \bibfield  {author} {\bibinfo {author} {\bibfnamefont {S.}~\bibnamefont
  {Chakraborty}}, \bibinfo {author} {\bibfnamefont {S.}~\bibnamefont {Hazra}},
  \bibinfo {author} {\bibfnamefont {T.}~\bibnamefont {Li}}, \bibinfo {author}
  {\bibfnamefont {C.}~\bibnamefont {Shao}}, \bibinfo {author} {\bibfnamefont
  {X.}~\bibnamefont {Wang}},\ and\ \bibinfo {author} {\bibfnamefont
  {Y.}~\bibnamefont {Zhang}},\ }\bibfield  {title} {\bibinfo {title} {Quantum
  singular value transformation without block encodings: Near-optimal
  complexity with minimal ancilla},\ }\bibfield  {journal} {\bibinfo  {journal}
  {arXiv preprint arXiv:2504.02385}\ }\href
  {https://doi.org/10.48550/arXiv.2504.02385} {10.48550/arXiv.2504.02385}
  (\bibinfo {year} {2025})\BibitemShut {NoStop}%
\bibitem [{\citenamefont {Kato}\ \emph {et~al.}(2024)\citenamefont {Kato},
  \citenamefont {Wada}, \citenamefont {Ito},\ and\ \citenamefont
  {Yamamoto}}]{kato2024exponentially}%
  \BibitemOpen
  \bibfield  {author} {\bibinfo {author} {\bibfnamefont {J.}~\bibnamefont
  {Kato}}, \bibinfo {author} {\bibfnamefont {K.}~\bibnamefont {Wada}}, \bibinfo
  {author} {\bibfnamefont {K.}~\bibnamefont {Ito}},\ and\ \bibinfo {author}
  {\bibfnamefont {N.}~\bibnamefont {Yamamoto}},\ }\bibfield  {title} {\bibinfo
  {title} {Exponentially accurate open quantum simulation via randomized
  dissipation with minimal ancilla},\ }\bibfield  {journal} {\bibinfo
  {journal} {arXiv preprint arXiv:2412.19453}\ }\href
  {https://doi.org/10.48550/arXiv.2412.19453} {10.48550/arXiv.2412.19453}
  (\bibinfo {year} {2024})\BibitemShut {NoStop}%
\bibitem [{\citenamefont {Li}\ and\ \citenamefont
  {Wang}(2023{\natexlab{b}})}]{li2023succinct}%
  \BibitemOpen
  \bibfield  {author} {\bibinfo {author} {\bibfnamefont {X.}~\bibnamefont
  {Li}}\ and\ \bibinfo {author} {\bibfnamefont {C.}~\bibnamefont {Wang}},\
  }\bibfield  {title} {\bibinfo {title} {Succinct description and efficient
  simulation of non-{M}arkovian open quantum systems},\ }\href
  {https://doi.org/10.1007/s00220-023-04638-4} {\bibfield  {journal} {\bibinfo
  {journal} {Communications in Mathematical Physics}\ }\textbf {\bibinfo
  {volume} {401}},\ \bibinfo {pages} {147} (\bibinfo {year}
  {2023}{\natexlab{b}})}\BibitemShut {NoStop}%
\bibitem [{\citenamefont {Berry}\ \emph {et~al.}(2020)\citenamefont {Berry},
  \citenamefont {Childs}, \citenamefont {Su}, \citenamefont {Wang},\ and\
  \citenamefont {Wiebe}}]{berry2020timedependent}%
  \BibitemOpen
  \bibfield  {author} {\bibinfo {author} {\bibfnamefont {D.~W.}\ \bibnamefont
  {Berry}}, \bibinfo {author} {\bibfnamefont {A.~M.}\ \bibnamefont {Childs}},
  \bibinfo {author} {\bibfnamefont {Y.}~\bibnamefont {Su}}, \bibinfo {author}
  {\bibfnamefont {X.}~\bibnamefont {Wang}},\ and\ \bibinfo {author}
  {\bibfnamefont {N.}~\bibnamefont {Wiebe}},\ }\bibfield  {title} {\bibinfo
  {title} {Time-dependent {H}amiltonian simulation with {$L^1$}-norm scaling},\
  }\href {https://doi.org/10.22331/q-2020-04-20-254} {\bibfield  {journal}
  {\bibinfo  {journal} {{Quantum}}\ }\textbf {\bibinfo {volume} {4}},\ \bibinfo
  {pages} {254} (\bibinfo {year} {2020})}\BibitemShut {NoStop}%
\bibitem [{\citenamefont {Chen}\ \emph {et~al.}(2021)\citenamefont {Chen},
  \citenamefont {Kalev},\ and\ \citenamefont {Hen}}]{chen2021quantumalgorithm}%
  \BibitemOpen
  \bibfield  {author} {\bibinfo {author} {\bibfnamefont {Y.-H.}\ \bibnamefont
  {Chen}}, \bibinfo {author} {\bibfnamefont {A.}~\bibnamefont {Kalev}},\ and\
  \bibinfo {author} {\bibfnamefont {I.}~\bibnamefont {Hen}},\ }\bibfield
  {title} {\bibinfo {title} {Quantum algorithm for time-dependent hamiltonian
  simulation by permutation expansion},\ }\href
  {https://doi.org/10.1103/PRXQuantum.2.030342} {\bibfield  {journal} {\bibinfo
   {journal} {PRX Quantum}\ }\textbf {\bibinfo {volume} {2}},\ \bibinfo {pages}
  {030342} (\bibinfo {year} {2021})}\BibitemShut {NoStop}%
\bibitem [{\citenamefont {Watkins}\ \emph {et~al.}(2024)\citenamefont
  {Watkins}, \citenamefont {Wiebe}, \citenamefont {Roggero},\ and\
  \citenamefont {Lee}}]{watkins2024time-dependent}%
  \BibitemOpen
  \bibfield  {author} {\bibinfo {author} {\bibfnamefont {J.}~\bibnamefont
  {Watkins}}, \bibinfo {author} {\bibfnamefont {N.}~\bibnamefont {Wiebe}},
  \bibinfo {author} {\bibfnamefont {A.}~\bibnamefont {Roggero}},\ and\ \bibinfo
  {author} {\bibfnamefont {D.}~\bibnamefont {Lee}},\ }\bibfield  {title}
  {\bibinfo {title} {Time-dependent hamiltonian simulation using discrete-clock
  constructions},\ }\href {https://doi.org/10.1103/PRXQuantum.5.040316}
  {\bibfield  {journal} {\bibinfo  {journal} {PRX Quantum}\ }\textbf {\bibinfo
  {volume} {5}},\ \bibinfo {pages} {040316} (\bibinfo {year}
  {2024})}\BibitemShut {NoStop}%
\bibitem [{\citenamefont {Fang}\ \emph {et~al.}(2025)\citenamefont {Fang},
  \citenamefont {Liu},\ and\ \citenamefont {Sarkar}}]{fang2025time}%
  \BibitemOpen
  \bibfield  {author} {\bibinfo {author} {\bibfnamefont {D.}~\bibnamefont
  {Fang}}, \bibinfo {author} {\bibfnamefont {D.}~\bibnamefont {Liu}},\ and\
  \bibinfo {author} {\bibfnamefont {R.}~\bibnamefont {Sarkar}},\ }\bibfield
  {title} {\bibinfo {title} {Time-dependent hamiltonian simulation via magnus
  expansion: Algorithm and superconvergence},\ }\href
  {https://doi.org/10.1007/s00220-025-05314-5} {\bibfield  {journal} {\bibinfo
  {journal} {Communications in Mathematical Physics}\ }\textbf {\bibinfo
  {volume} {406}},\ \bibinfo {pages} {1} (\bibinfo {year} {2025})}\BibitemShut
  {NoStop}%
\bibitem [{\citenamefont {Ruskai}(1972)}]{ruskai1972inequalities}%
  \BibitemOpen
  \bibfield  {author} {\bibinfo {author} {\bibfnamefont {M.~B.}\ \bibnamefont
  {Ruskai}},\ }\bibfield  {title} {\bibinfo {title} {Inequalities for traces on
  von {N}eumann algebras},\ }\href {https://doi.org/10.1007/BF01645523}
  {\bibfield  {journal} {\bibinfo  {journal} {Communications in Mathematical
  Physics}\ }\textbf {\bibinfo {volume} {26}},\ \bibinfo {pages} {280}
  (\bibinfo {year} {1972})}\BibitemShut {NoStop}%
\end{thebibliography}%

\newpage
\widetext
\setcounter{equation}{0}
\setcounter{figure}{0}
\setcounter{table}{0}
\setcounter{algocf}{0}
\setcounter{section}{0}
\setcounter{theorem}{0}
%\makeatletter
\renewcommand{\theequation}{A\arabic{equation}}
\renewcommand{\thetable}{A\arabic{table}}
\renewcommand{\thefigure}{A\arabic{figure}}
\renewcommand{\thetheorem}{A\arabic{theorem}}
\renewcommand{\thelemma}{A\arabic{lemma}}
\renewcommand{\thealgocf}{A\arabic{algocf}}
\renewcommand{\thecorollary}{A\arabic{corollary}}
\renewcommand{\thesection}{A - \Roman{section}}
\newpage
\appendix
\begin{center}
\textbf{\huge Appendix}
\end{center}

In the Appendix, we provide a comprehensive list of mathematical symbols, an LCU decomposition of unitaries, prove some results concerning distances between quantum states that we make use of in the main article, and formally demonstrate the correctness of Algorithm~\ref{algo: nm-dynamics}.

\section{List of mathematical symbols and variables}
\label{sec:symbols}
%Table~\ref{table:notation-table} provides a comprehensive list of mathematical symbols used throughout the paper, along with their definitions and hyperlinks to their first appearance in the main article.

\setlength{\extrarowheight}{2pt}
\begin{table}[htbp]
\scriptsize
\caption{Consolidated list of mathematical symbols}
\label{table:notation-table}
%\begin{tabular}{|c|l|l||c|l|l|}
{\centering
\renewcommand\baselinestretch{1.20}\selectfont
\begin{tabular*}{\textwidth}{c@{\extracolsep{\fill}} llcll}
\hline
\textbf{Symbol} & \textbf{Definition} & \textbf{Reference} & \textbf{Symbol} & \textbf{Definition} & \textbf{Reference} \\
\hline\hline
$n$ & Number of qubits in the system & ~~~--- & $t$ & Evolution time & ~~~--- \\
%\hline
$\varepsilon$ & Precision/additive accuracy & ~~~--- & $O$ & some arbitrary Observable & ~~~--- \\
%\hline
$g(n) = \bigO(f(n))$ & Big O notation & ~~~--- & $\widetilde{\bigO}(f(n))$ & Big O notation (hiding polylog factors) & ~~~--- \\
%\hline
$\Tr[A]$ & Trace of operator $A$ & ~~~--- & $\mathbb{E}[A]$ & Expectation value of operator $A$ & ~~~--- \\
%\hline
$\Pr[X]$ & Probability of event $X$ & ~~~--- & $\|X\|_p$ & Schatten $p$-norm of operator $X$ & ~~~--- \\
%\hline
$\sigma_j(X)$ & $j$-th singular value of $X$ & ~~~--- & $\|X\|$ & Spectral norm of operator $X$ & ~~~--- \\
%\hline
$\|\mathcal{M}\|_{1 \to 1}$ & Induced 1-norm of superoperator & ~~~--- & $\sigma^x, \sigma^y, \sigma^z$ & Pauli matrices & ~~~--- \\
%\hline
$\sigma^{\pm}$ & Raising/lowering operators & ~~~--- & $I$ & Identity operator & ~~~--- \\
%\hline
$|0\rangle, |1\rangle$ & Computational basis states & ~~~--- & $|+\rangle$ & Plus state: $(|0\rangle + |1\rangle)/\sqrt{2}$ & ~~~--- \\
%\hline
$\mathcal{H}_S$ & System Hilbert space & Sec.~\ref{subsec:K-collision-approx} & $\mathcal{H}_E$ & Environment Hilbert space & Sec.~\ref{subsec:K-collision-approx} \\
%\hline
$\mathcal{H}_{E_j}$ & $j$-th sub-environment Hilbert space & Sec.~\ref{subsec:K-collision-approx} & $H_S$ & System Hamiltonian & Eq.~\eqref{eq:total-Hamiltonian} \\
%\hline
$H_{E_j}$ & Hamiltonian of $j$-th sub-environment & Eq.~\eqref{eq:total-Hamiltonian} & $H_{I_j}$ & Interaction Hamiltonian for $j$-th collision & Eq.~\eqref{eq:total-Hamiltonian} \\
%\hline
$H$ & Total Hamiltonian & Eq.~\eqref{eq:total-Hamiltonian} & $H_j$ & Total Hamiltonian for $j$-th collision & Eq.~\eqref{eq:j-th collision Hamiltonian} \\
%\hline
$\overline{H}_j$ & Normalized Hamiltonian for $j$-th collision & Eq.~\eqref{eq: Collision unitary} & $m$ & Number of sub-environments & Sec.~\ref{subsec:K-collision-approx} \\
%\hline
$\rho_S$ & State of the system & Sec.~\ref{subsec:K-collision-approx} & $\rho_{E_j}$ & State of the $j$-th sub-environment & Sec.~\ref{subsec:K-collision-approx} \\
%\hline
$P_{i,j}$ & Pauli operators in Hamiltonian decomp. & Eq.~\eqref{eq:j-th collision Hamiltonian} & $h_{i,j}$ & Coefficients in Pauli decomposition & Eq.~\eqref{eq:j-th collision Hamiltonian} \\
%\hline
$L_j$ & Number of terms in the total Hamiltonian $H_j$ & Eq.~\eqref{eq:j-th collision Hamiltonian} & $L$ & Maximum number of terms in $H_j$, across collisions & Eq.~\eqref{eq:j-th collision Hamiltonian} \\
%\hline
$\beta_j$ & Total weight of the Pauli coefficients: $\sum_{i=1}^{L_j} h_{i,j}$ & Eq.~\eqref{eq:j-th collision Hamiltonian} & $\beta$ & Maximum weight: $\max_j \beta_j$ & Theorem~\ref{thm: Algorithm 1 proof} \\
%\hline
$\Delta t$ & Time duration of each collision & Sec.~\ref{subsec:K-collision-approx} & $K$ & Total number of collisions & Sec.~\ref{subsec:K-collision-approx} \\
%\hline
$U_j$ & Time evolution unitary for $j$-th collision & Eq.~\eqref{eq: Collision unitary} & $\widetilde{U}_j$ & Approximate time evolution operator & Eq.~\eqref{eq: approximate collision map} \\
%\hline
$\Phi_j$ & $j$-th collision map & Definition~\ref{def: collision map} & $\widetilde{\Phi}_j$ & Approximate $j$-th collision map & Eq.~\eqref{eq: approximate collision map} \\
%\hline
$\mathcal{M}_K$ & Markovian $K$-collision map & Definition~\ref{def: K-collision map} & $\widetilde{\mathcal{M}}_K$ & Approx. Markovian $K$-collision map & Eq.~\eqref{eq:approx-K-collision-map} \\
%\hline
$\mu$ & Algorithm output estimate & Eq.~\eqref{eq: error between experiment and the collision map} & $\widetilde{U}$ & LCU approximation of time evolution & Eq.~\eqref{eq: index set LCU decomposition of time evolution operator} \\
%\hline
$\alpha_i$ & LCU coefficients & Eq.~\eqref{eq: index set LCU decomposition of time evolution operator} & $W_i$ & LCU unitaries & Eq.~\eqref{eq: index set LCU decomposition of time evolution operator} \\
%\hline
$\alpha$ & Total LCU weight: $\sum_i |\alpha_i|$ & Lemma~\ref{lemma: lcu of time evolution} & $r$ & LCU parameter to control the error & Lemma~\ref{lemma: lcu of time evolution} \\
%\hline
$q$ & Taylor series truncation parameter & Lemma~\ref{lemma: lcu of time evolution} & $X_j, Y_j$ & i.i.d. sampled unitaries in SA-LCU algorithm & Eq.~\eqref{eq: probablity emsemble} \\
%\hline
$X_j^{(c)}$ & Controlled version of $X_j$ & Eq.~\eqref{eq: control and anti-control notation} & $Y_j^{(a)}$ & Anti-controlled version of $Y_j$ & Eq.~\eqref{eq: control and anti-control notation} \\
%\hline
$\mathcal{D}_j$ & Sampling ensemble & Eq.~\eqref{eq: probablity emsemble} & $\alpha^{(j)}$ & LCU weight for $j$-th collision & Algorithm~\ref{algo: collision model} \\
%\hline
$\zeta$ & Product of LCU weights: $\prod_{j=1}^K \alpha^{(j)}$ & Algorithm~\ref{algo: collision model} & $T$ & Number of classical repetitions & Eq.~\eqref{eq: number of classical repetitions} \\
%\hline
$\tau_d$ & Circuit depth per coherent run & Eq.~\eqref{eq: circuit depth per run} & $\tau_{\rho_E}$ & Max. circuit depth to prepare the environment in $\rho_E$ & Eq.~\eqref{eq: circuit depth per run} \\
%\hline
$\tau_{\rho_{E_j}}$ & Circuit depth for preparing $\rho_{E_j}$ & Eq.~\eqref{eq: circuit depth per run} & $\delta$ & Failure probability & Theorem~\ref{thm: Algorithm 1 proof} \\
%\hline
$\mathcal{L}$ & Lindbladian superoperator & Eq.~\eqref{eq: lindbladian map} & $A_j$ & Lindblad jump operators & Eq.~\eqref{eq: lindbladian map} \\
%\hline
$\omega$ & Inverse temperature & Eq.~\eqref{eq:sub-environment-thermal} & $\lambda$ & System-environment coupling & Sec.~\ref{sec:Lindbladian Dynamics Simulation} \\
%\hline
$\nu$ & Repetitions of collision sequence & Def.~\ref{def: m, nu-collision map} & $\mathcal{M}_{m,\nu}$ & $(m,\nu)$-collision map & Def.~\ref{def: m, nu-collision map} \\
%\hline
$\Gamma$ & Complexity parameter & Eq.~\eqref{eq:Gamma} & $\beta_S$ & System Hamiltonian weight & Sec.~\ref{sec:Lindbladian Dynamics Simulation} \\
%\hline
$\beta_{I_j}$ & Interaction Hamiltonian weight & Sec.~\ref{sec:Lindbladian Dynamics Simulation} & $\beta_{E_j}$ & Environment Hamiltonian weight & Sec.~\ref{sec:Lindbladian Dynamics Simulation} \\
%\hline
$\beta_{\max}$ & Maximum weight parameter & Eq.~\eqref{eq:beta-max} & $L_S$ & Terms in system Hamiltonian & Eq.~\eqref{eq:total-terms-L} \\
%\hline
$L_{I_j}$ & Terms in interaction Hamiltonian & Eq.~\eqref{eq:total-terms-L} & $M_z$ & Avg. transverse magnetization & Eq.~\eqref{eq: Magnetization operator} \\
%\hline
$J$ & Coupling strength in Ising model & Eq.~\eqref{eq: ising model} & $h$ & Transverse magnetic field strength & Eq.~\eqref{eq: ising model} \\
%\hline
$\gamma$ & Damping strength & Eq.~\eqref{eq:lindbladian-amplitude-interaction-ham} & $\mathcal{C}_{i,j}$ & CPTP channel between sub-environments & Eq.~\eqref{eq:non-markovian-partial-swap} \\
%\hline
$S_{j,j+1}$ & Swap operation between sub-environments & Eq.~\eqref{eq:non-markovian-partial-swap} & $p$ & Swap probability parameter & Eq.~\eqref{eq:non-markovian-partial-swap} \\
%\hline
$U_{S_j}$ & System-environment collision unitary & Eq.~\eqref{eq:nm-system-env-collision} & $\mathcal{V}_{j,j+1}$ & Environment-environment interaction & Eq.~\eqref{eq:nm-env-env-collision} \\
%\hline
$\Phi_j^{\mathcal{N}}$ & Non-Markovian collision map & Eq.~\eqref{eq:nm-collision-map} & $\mathcal{N}_K$ & Non-Markovian $K$-collision map & Eq.~\eqref{eq:nm-k-collision} \\
%\hline
$\widetilde{\Phi}_j^{\mathcal{N}}$ & Approx. non-Markovian collision map & Eq.~\eqref{eq: approximate non-markovian collision map} & $\widetilde{\mathcal{N}}_K$ & Approx. non-Markovian $K$-collision map & Eq.~\eqref{eq: approximate non-markovian K-collision map} \\
\hline
\end{tabular*}}
\end{table}
\newpage

\section{LCU decomposition of Unitaries}
\label{appendix: LCU decomposition of Unitaries}
We provide an LCU decomposition of the time evolution operator, which we incorporate into the simulation of the quantum collision models. This has been proven in Refs.~\cite{wan2022randomized, wang2023qubit, chakraborty2023implementing}, and we restate the result here for completeness. 

Consider any Hamiltonian $H$ which is a convex combination of Pauli operators, i.e.\, $H=\sum_{l=1}^{L}p_l P_l$, where $P_l$ is a sequence of Pauli operators, and $\sum_l p_l = 1$. We decompose $e^{-i\tau H}$ as an approximate linear combination of unitaries, for which we use ideas from the Truncated Taylor series method by Berry et al.\cite{berry2015simulating}, as well as the LCU decomposition of \cite{wan2022randomized}. We can write
\begin{align}
e^{-i\tau H}=\left(e^{-iH\tau/r}\right)^r,
\end{align}
where $r$ (to be selected later) is a parameter such that $r>t$. If each segment $S_r=e^{-iH\tau/r}$ has an (approximate) LCU decomposition $\sum_{m} c_m U_m$, such that $\norm{S_r-\sum_{m} c_m U_m}\leq \eps/r$ then,
\begin{align}
S=(S_r)^r &= \sum_{m_1 m_2\cdots m_r} c_{m_1} c_{m_2}\cdots c_{m_r} U_{m_1}U_{m_2}\cdots U_{m_r}=\sum_{j}\alpha_j W_j,
\end{align}
is $\eps$-close to $e^{-iHt}$, i.e.\ $\norm{e^{-iHt}-S}\leq \eps$. First note that by truncating the $S_r=e^{-i H\tau/r}$ to $q$ terms, we obtain
\begin{align}
\widetilde{S}_r=\sum_{k=0}^{q} \dfrac{(-i\tau H/r)^k}{k!}.
\end{align}
Then by choosing some 
\begin{align}
q=\bigO\left(\dfrac{\log(r/\eps)}{\log\log(r/\eps)}\right),
\end{align}
we ensure that $\norm{S_r - \widetilde{S}_r}\leq \eps/r$.

Now, we obtain the LCU decomposition of $\widetilde{S}_r$, similar in spirit to Ref.~\cite{wan2022randomized}. This gives us, 

\begin{align}
&\widetilde{S}_r=\sum_{k=0}^{q} \dfrac{(-i\tau H/r)^k}{k!}\\
&=\sum_{k=0,~k\in \mathrm{even}}^{q} \dfrac{1}{k!}(-i\tau H/r)^k \left(I- \dfrac{i\tau H/r}{k+1}\right)\\
&=\sum_{k=0,~k\in \mathrm{even}}^{q} \dfrac{1}{k!}\left(-i\tau /r \sum_{\ell=1}^{L} p_{\ell} P_{\ell}\right)^k \times \left(I- \dfrac{i\tau /r}{k+1}\left(\sum_{m=1}^{L} p_{m} P_{m}\right)\right)\\
&= \sum_{k=0,~k\in \mathrm{even}}^{q} \dfrac{(-i\tau /r)^k}{k!} \times \sum_{\ell_1,\ell_2,\cdots \ell_k = 1}^{L} p_{\ell_1}p_{\ell_2}\cdots p_{\ell_k} P_{\ell_1}P_{\ell_2}\cdots P_{\ell_k} \nonumber \times \sum_{m=1}^{L} p_m \left(I- \dfrac{i\tau P_m /r}{k+1}\right)\\
&=\sum_{k=0,~k\in \mathrm{even}}^{q} \dfrac{(-i\tau /r)^k}{k!}\sqrt{1+\left(\dfrac{\tau /r}{k+1}\right)^2}\times \sum_{\ell_1,\ell_2,\cdots \ell_k, m = 1}^{L} p_{\ell_1}p_{\ell_2}\cdots p_{\ell_k} p_m P_{\ell_1}P_{\ell_2}\cdots P_{\ell_k} e^{-i\theta_{m} P_m},
\end{align}

where $e^{-i\theta_{m}P_m}$ is a Pauli rotation operator, defined as follows:

\begin{align}
e^{i\theta_{m}P_m}=\dfrac{1}{\sqrt{1+\left(\dfrac{\tau /r}{k+1}\right)^2}}\left(I- \dfrac{i\tau P_m /r}{k+1}\right),
\end{align}
such that
\begin{align}
\theta_m = \arccos\left(\left[1+\left(\dfrac{\tau /r}{k+1}\right)^2\right]^{-1/2}\right).
\end{align}
Thus, $\widetilde{S}_r=\sum_{j\in M} c_j U_j$, where the index set $M$ can be defined as 
\begin{align}
M =\{(k, &\ell_1, \ell_2,\cdots \ell_k, m):\nonumber
\\&(i)\quad 0\leq k \leq K; \nonumber
\\&(ii)\quad \ell_1,\ell_2, \cdots \ell_k, m \in \{1,2,\cdots, L\}\}.
\end{align}

Also,
\begin{align}
c_j = \dfrac{(\tau /r)^k}{k!} \sqrt{1+\left(\dfrac{\tau /r}{k+1}\right)^2}p_{\ell_1}p_{\ell_2}\cdots p_{\ell_k} p_m,
\end{align}
while
\begin{align}
U_j=(-i)^k P_{\ell_1}P_{\ell_2}\cdots P_{\ell_k} e^{i\theta_{m} P_m}.
\end{align}

The sum of the coefficients
\begin{align}
&\sum_{j\in M} |c_j| = \sum_{k=0,~k\in \mathrm{even}}^{q}  \dfrac{(\tau /r)^k}{k!}\sqrt{1+\left(\dfrac{\tau /r}{k+1}\right)^2} \times \sum_{\ell_1,\ell_2,\cdots \ell_k, m = 1}^{L} p_{\ell_1}p_{\ell_2}\cdots p_{\ell_k} p_m\\
&=\sum_{k=0,~k\in \mathrm{even}}^{q}  \dfrac{(\tau /r)^k}{k!}\sqrt{1+\left(\dfrac{\tau /r}{k+1}\right)^2}\\
&\leq \sum_{k=0,~k\in \mathrm{even}}^{\infty}  \dfrac{(\tau /r)^k}{k!} \sqrt{1+\left(\dfrac{\tau /r}{k+1}\right)^2} \nonumber \\
&= \sum_{k=0}^{\infty}  \dfrac{(\tau /r)^{2k}}{(2k)!}\sqrt{1+\left(\dfrac{\tau /r}{2k+1}\right)^2}\\
&\leq \sum_{k=0}^{\infty}  \dfrac{(\tau /r)^{2k}}{k!} = e^{\tau ^2/r^2}.
\end{align}
Finally, in order to write down $S$ as an LCU, we write $S=\widetilde{S}_r^r$. That is,
\begin{align}
S&=\left(\sum_{j\in M} c_j U_j\right)^r=\sum_{j_1, j_2,\cdots j_r \in M} c_1 c_2 \cdots c_r U_{j_1} U_{j_2}\cdots U_{j_r} =\sum_{m}\alpha_m W_m,
\end{align}
where $|\alpha| = \sum_{m}|\alpha_m|=(\sum_{j\in M} |c_j|)^r\leq e^{\tau ^2/r}$.

\section{Distances between quantum states}

In this section, we prove some results concerning the distance between operators/ CPTP maps applied to quantum states. First, consider that there exist two operators $P$ and $Q$ such that $\norm{P-Q}\leq \gamma$. We demonstrate that the expectation value of $O$ with respect to $P\rho P^{\dag}$ is not far off from the expectation value of $O$ with respect to $Q\rho Q^{\dag}$, for any density matrix $\rho$. More precisely, we prove that
$$
\left|\Tr[O~P\rho P^{\dag}]-\Tr[O~Q\rho Q^{\dag}]\right|\leq 3\norm{P}\norm{O}\gamma.
$$
The result was proven in Refs.~\cite{chakraborty2023implementing, chakraborty2025quantum}, and we state this here for completeness. Let us recall the tracial version of H\"{o}lder's inequality, which is stated below for completeness:
~\\
\begin{lemma}[Tracial version of H\"{o}lder's inequality \cite{ruskai1972inequalities}]
\label{thm:holder}
Define two operators $A$ and $B$ and parameters $p,q\in [1,\infty]$ such that $1/p+1/q =1 $. Then the following holds:
$$
\Tr[A^{\dag}B]\leq \norm{A}_p \norm{B}_q.
$$
\end{lemma}
Here $\norm{X}_p$ corresponds to the Schatten $p$-norm of the operator $X$. For the special case of $p=\infty$ and $q=1$, the statement of Lemma \ref{thm:holder} can be rewritten as
\begin{equation}
\label{eq:holder-special-case}
\Tr[A^{\dag}B] = \|A^{\dag}B\|_1 \leq \norm{A}_{\infty} \norm{B}_1=\norm{A} \norm{B}_1.
\end{equation}
Now we are in a position to formally state the main result.
~\\
%%%%%%%%%%%%%%%%%%%%%%%%%%%%%%%%%%%%%%%%%%%%%%%%%
\begin{theorem}
\label{thm:distance-expectation}
Suppose $P$ and $Q$ are operators such that $\norm{P-Q}\leq \gamma$ for some $\gamma\in [0,1]$. Furthermore, let $\rho$ be any density matrix and $O$ be some Hermitian operator with spectral norm $\norm{O}$. Then, if $\norm{P}\geq 1$, the following holds:
$$
\left|\Tr[O~P\rho P^{\dag}]-\Tr[O~Q\rho Q^{\dag}]\right| \leq 3\norm{O}\norm{P}\gamma.
$$
\end{theorem}
\begin{proof}
Using Lemma \ref{thm:holder} with $p=\infty$ and $q=1$, we obtain

\begin{align}\label{eq:robustness-exp-inequality}
|\Tr[O~P\rho P^{\dag}]-&\Tr[O~Q\rho Q^{\dag}]|\leq \norm{O}\cdot \|P\rho P^{\dag}-Q\rho Q^{\dag}\|_1
\end{align}

For the second term in the RHS of the above equation, we can successively apply the tracial version of H\"{o}lder's inequality (Lemma \ref{thm:holder} with $p=\infty$ and $q=1$) the triangle inequality to obtain:

\begin{align}
\bigg\|P\rho P^{\dag}-Q\rho Q^{\dag}\bigg\|_1&=\norm{P\rho P^{\dag}-P\rho Q^{\dag}+P\rho Q^{\dag}-Q\rho Q^{\dag}}_1\\
                &\leq \norm{P\rho}_1\norm{P-Q}+\norm{P-Q}\norm{\rho Q}_1\\
                &\leq \norm{P}\norm{P-Q}+\norm{Q}\norm{P-Q}~~~~~~~~~~~~~~~~[\text{~As~}\norm{\rho}_1=1]\\
                &\leq \left(\|P\|+\|Q\|\right)\cdot \norm{P-Q}\\
                &\leq \left(\|P\|+\|Q-P+P\|\right)\cdot \norm{P-Q}\\
                &\leq \left(\|P\|+\norm{P-Q}+\norm{P}\right)\cdot \norm{P-Q}\\
                &\leq 2\norm{P}\norm{P-Q}+\norm{P-Q}^2.
\end{align}

Now, substituting this upper bound back in the RHS of Eq.~\eqref{eq:robustness-exp-inequality}, we obtain

\begin{align}
\left|\Tr[O~P\rho P^{\dag}] - \Tr[O~Q\rho Q^{\dag}]\right| &\leq \big\|O\big\| \big\|P-Q\big\|^2 + 2\big\|O\big\|\big\|P\big\|\big\|P-Q\big\|\\
&\leq \big\|O\big\| \big\|P-Q\big\|^2+2\big\|O\big\|\big\|P\big\|\big\|P-Q\big\|\nonumber\\
&\quad\leq \gamma^2\big\|O\big\|+2\big\|O\big\|\big\|P\big\|\gamma \nonumber\\
&\quad\leq 3\gamma\big\|O\big\|\big\|P\big\|
\end{align}
    
\end{proof}
Next, we bound the distance between two quantum states that have been transformed by a composition of two Completely Positive Trace Preserving (CPTP) maps. We have the following Lemma:
\begin{lemma}[Distance between quantum states obtained by applying a composition of CPTP maps]
\label{thm: Bounds on the distance between the composition of maps}
Let \( \{\mathcal{A}_i\}_{i=1}^K \) and \( \{\mathcal{B}_i\}_{i=1}^K \) be two sets of maps acting on any density operator \( \rho \), such that each \( \mathcal{A}_i \) and \( \mathcal{B}_i \) are CPTP maps. Assume that for all $i\in [1,K]$ , the following bound holds:
    \begin{equation}\label{eq: bound on composition initial assumption}
        \norm{\mathcal{A}_i[\rho] - \mathcal{B}_i[\rho]} _1\leq \eps.
    \end{equation}
Then, the compositions of these maps satisfy:
    \begin{equation}
        \norm{\bigcirc_{i=1}^K \mathcal{A}_i[\rho] - \bigcirc_{i=1}^K \mathcal{B}_i[\rho]}_1 \leq K \eps.
    \end{equation}
    \end{lemma}

\begin{proof}
Expanding the composition, we write:
\begin{align}
\norm{\bigcirc_{i=1}^K \mathcal{A}_i[\rho] - \bigcirc_{i=1}^K \mathcal{B}_i[\rho]}_1 =\norm{\mathcal{A}_K \mathcal{A}_{K-1} \dots \mathcal{A}_1[\rho] - \mathcal{B}_K \mathcal{B}_{K-1} \dots \mathcal{B}_1[\rho]}_1.
\end{align}
    
Adding and subtracting intermediate terms iteratively, and using the  triangle inequality, we obtain:
    \begin{align}
    \norm{\mathcal{A}_K \mathcal{A}_{K-1} \dots \mathcal{A}_1[\rho] - \mathcal{B}_K \mathcal{B}_{K-1} \dots \mathcal{B}_1[\rho]}_1 &\leq \norm{\mathcal{A}_K \mathcal{A}_{K-1} \dots \mathcal{A}_1[\rho] - \mathcal{B}_K \mathcal{A}_{K-1} \dots \mathcal{A}_1[\rho]}_1 + \nonumber \\
    &~~~~\norm{\mathcal{B}_K \mathcal{A}_{K-1} \dots \mathcal{A}_1[\rho] - \mathcal{B}_K \mathcal{B}_{K-1}\dots \mathcal{A}_1[\rho]}_1
    +\dots + \nonumber \\
    &~~~~\norm{\mathcal{B}_K\dots \mathcal{B}_2 \mathcal{A}_1[\rho] - \mathcal{B}_K\dots \mathcal{B}_2\mathcal{B}_1[\rho]}_1. \label{eq: bounds on the composition intermediate terms}
    \end{align}
    \vspace{5pt}
Now, let $ \rho^{\mathcal{A}}_k = \bigcirc_{i=1}^k \mathcal{A}_i[\rho] $ and $ \rho^{\mathcal{B}}_k = \bigcirc_{i=1}^k \mathcal{B}_i[\rho] $, representing the intermediate states obtained after the application of $k$ maps. Then due to the initial assumption Eq.~\eqref{eq: bound on composition initial assumption}
    \begin{equation}
        \norm{\mathcal{A}_{j}[\rho^{\mathcal{A}}_{j-1}] - \mathcal{B}_{j}[\rho^{\mathcal{B}}_{j-1}]}_1 \leq \eps
    \end{equation}
Using the contractivity of completely positive maps, 
    \begin{equation}
    \|\mathcal{B}_{j+1} \mathcal{A}_{j}[\rho^{\mathcal{A}}_{j-1}] - \mathcal{B}_{j+1}\mathcal{B}_{j}[\rho^\mathcal{B}_{j-1}]\|_1\leq \eps,
    \end{equation}
Using this repeatedly, we bound each term on the right side of the Eq.~\eqref{eq: bounds on the composition intermediate terms} by $\eps$. Thus, by summing the contributions across all maps, we obtain:
    \begin{equation}
    \left\|\bigcirc_{i=1}^K \mathcal{A}_i[\rho] - \bigcirc_{i=1}^K \mathcal{B}_i[\rho]\right\|_1 \leq K \eps.
    \end{equation}
\end{proof}

Now, consider two copies of the quantum state $\rho$, such that the CPTP map $\mathcal{A}$ has been applied to one copy, while another CPTP map $\mathcal{B}$ has been applied to the second copy to obtain $\mathcal{A}[\rho]$ and $\mathcal{B}[\rho]$, respectively. Now consider unitaries $U$ and $\widetilde{U}$ such that they are close (in spectral norm). Then, we find the distance ($1$-norm) between the quantum states obtained by applying $U$ to $\mathcal{A}[\rho]$, and $\widetilde{U}$ to $\mathcal{A}[\rho]$, via the following theorem:

\begin{theorem}[Distance between quantum states]
\label{thm:dist-unitary-cptp}
Suppose $U$ and $\widetilde{U}$ are unitary while $\mathcal{A}$, $\mathcal{B}$ are CPTP maps. Then for any density operator $\rho$,
$$
\norm{U~\mathcal{A}[\rho]U^{\dag} - \widetilde{U}~\mathcal{B}[\rho] \widetilde{U}^{\dag}}_1 \leq 2\norm{U-\widetilde{U}}+\norm{\mathcal{A}[\rho]-\mathcal{B}[\rho]}_1.
$$
~\\
\begin{proof}
    Using triangle inequality, and Tracial version of H\"{o}lder's inequality (Lemma \ref{thm:holder}), we obtain
    \begin{align}
        \norm{U~\mathcal{A}[\rho]U^{\dag} - \widetilde{U}~\mathcal{B}[\rho] \widetilde{U}^{\dag}}_1 &\leq \norm{U~\mathcal{A}[\rho]U^{\dag} - U~\mathcal{A}[\rho] \widetilde{U}^{\dag}}_1 + \norm{U~\mathcal{A}[\rho] \widetilde{U}^{\dag} - \widetilde{U}~\mathcal{B}[\rho] \widetilde{U}^{\dag}}_1\\
        &\leq\norm{U~\mathcal{A}[\rho]U^{\dag} - U~\mathcal{A}[\rho] \widetilde{U}^{\dag}}_1+\norm{U~\mathcal{A}[\rho] - \widetilde{U}~\mathcal{B}[\rho]}_1\cdot \norm{\widetilde{U}}\\
        &=\norm{U~\mathcal{A}[\rho]U^{\dag} - U~\mathcal{A}[\rho] \widetilde{U}^{\dag}}_1+\norm{U~\mathcal{A}[\rho] - \widetilde{U}~\mathcal{B}[\rho]}_1
    \end{align}
For the first term in the RHS we can use Lemma \ref{thm:holder} to obtain
\begin{align}
    \norm{U~\mathcal{A}[\rho]U^{\dag} - U~\mathcal{A}[\rho] \widetilde{U}^{\dag}}_1 &\leq \norm{U~\mathcal{A}[\rho]}_1\cdot \norm{U-\widetilde{U}}\\
    &\leq \norm{U}\cdot \norm{\mathcal{A}[\rho]}_1\cdot \norm{U-\widetilde{U}}\\
    &\leq \norm{U-\widetilde{U}}.
\end{align}
On the other hand, for the second term,
\begin{align}
    \norm{U~\mathcal{A}[\rho] - \widetilde{U}~\mathcal{B}[\rho]}_1=\norm{U~\mathcal{A}[\rho]-U~ \mathcal{B}[\rho]+U~\mathcal{B}[\rho]-\widetilde{U}\mathcal{B}[\rho]}_1.
\end{align}
By using triangle inequality once again to the RHS, followed by Lemma \ref{thm:holder}, we obtain
\begin{align}
    \norm{U~\mathcal{A}[\rho] - \widetilde{U}~\mathcal{B}[\rho]}_1&\leq\norm{U}\cdot\norm{\mathcal{A}[\rho]- \mathcal{B}[\rho]}_1+\norm{U-\widetilde{U}}\cdot \norm{\mathcal{B}[\rho]}_1\\
    &\leq \norm{\mathcal{A}[\rho]-\mathcal{B}[\rho]}_1+\norm{U-\widetilde{U}}.
\end{align}
So, overall, we have
$$
\norm{U~\mathcal{A}[\rho]U^{\dag} - \widetilde{U}~\mathcal{B}[\rho] \widetilde{U}^{\dag}}_1 \leq 2\norm{U-\widetilde{U}}+\norm{\mathcal{A}[\rho]-\mathcal{B}[\rho]}_1.
$$
This completes the proof.
\end{proof}
    
\end{theorem}

\section{Correctness of Algorithm \ref{algo: nm-dynamics}}
\label{sec:app-correctness-algo-nm}

In this section, we formally prove the correctness of the algorithm (Algorithm \ref{algo: nm-dynamics}) for simulating non-Markovian collisions. 

\begin{restatable}[]{theorem}{}
\label{thm: Algorithm 3 proof}
Let $\varepsilon, \delta \in (0,1)$. Then, for $\eps'=\varepsilon/(6K\norm{O})$, 
with probability at least $1 - \delta$, Algorithm \ref{algo: nm-dynamics} outputs $\mu$, such that
$$
\left|\mu-\Tr[O~\mathcal{N}_K[\rho_S]]\right|\leq \varepsilon,
$$
using $T$ repetitions of the circuit shown in Figure \ref{fig: single ancilla non markovian collision model circuit}, where
\begin{align}
       T = \bigO\left( \frac{\|O\|^2 \ln(2/\delta)}{\eps^2} \right).
\end{align} 
Each such coherent run has a circuit depth of 
\begin{equation}
    \tau_d = \bigO\left(\beta ^2K^2\Delta t^2 \dfrac{\log(\beta  K \norm{O}\Delta t/\eps)}{\log\log(\beta  K \norm{O}\Delta t/\eps)}+K\tau_{\rho_E}\right)
\end{equation}
where, $\beta = \max_j \beta_j$, and $\tau_{\rho_E}=\max_{j}\tau_{\rho_{E_j}}$, where $\tau_{\rho_{E_j}}$ is the circuit depth of the unitary preparing the sub-environment in the state $\rho_{E_j}$.
\end{restatable}
\begin{proof}
The proof is similar to Theorem \ref{thm: Algorithm 1 proof}. We first initialize the system and ancilla registers. Subsequently, we prepare the first environment register and apply the operations $X^{(c)}_1$ and $Y^{(a)}_1$ obtained by sampling $X_1$ and $Y_1$ from $\mathcal{D}_1$. Additionally, we also initialize another environment register that interacts with the first environment register via the partial swap operation represented by the map $\mathcal{V}_{1, 2}$. After this interaction, the first environment register is traced out, leaving the combined system, ancilla and the second environment registers in a state ready for subsequent interactions.

We define the map $\Phi^{\mathcal{N}~(PQ)}_j$ as:
\begin{equation}
    \widetilde{\Phi}^{\mathcal{N}^{(PQ)}}_j[.]\equiv \begin{cases}
\Tr_{E_{j}}\left[\mathcal{V}_{j,j+1}\big[P(.\otimes . \otimes \rho_{E_{j+1}})Q^{\dagger}\big]\right],~$j$~\text{is odd}\\~\\
\Tr_{E_{j}}\left[\mathcal{V}_{j,j+1}\big[P(.\otimes \rho_{E_{j+1}} \otimes .)Q^{\dagger}\big]\right],~ $j$~\text{is even}
\end{cases}     
\end{equation}
Here, the map $\widetilde{\Phi}_j^{\mathcal{N}^{(PQ)}}[.]$ represents applying operator $P$ from the left and $Q^{\dag}$ from the right, then applying the interaction between environment register followed by tracing out of the first (second) environment register $E_j$ if $j$ is odd (even). Thus, the state of the combined system-ancilla register after the initial collision can succinctly be expressed as:
\begin{align}
    \rho_1 = &\frac{1}{2} \bigg[
    |0\rangle\langle0| \otimes \Phi^{\mathcal{N}^{(Y_1Y_1)}}_{1}[\rho_S ]
    \nonumber + |0\rangle\langle1| \otimes \Phi^{\mathcal{N}^{(Y_1X_1)}}_{1}[\rho_S ] + |1\rangle\langle0| \otimes \Phi^{\mathcal{N}^{(X_1Y_1)}}_{1}[\rho_S ]
    \nonumber + |1\rangle\langle1| \otimes \Phi^{\mathcal{N}^{(X_1X_1)}}_{1}[\rho_S ]
    \bigg].
    \label{eq: expanded form of rho 1 in Algo 3}
\end{align}
We use the definition of the controlled and anti-controlled operator to simplify the combined state of the system and ancilla. 

To continue the process, we perform the next collision step analogously. We apply the next set of unitaries (obtained by sampling from $\mathcal{D}_2$), then we simulate the intra-environmental interaction by first initializing the next environment register in the state $\rho_{E_3}$ and interacting it with the previous environment and then tracing out the previous environment. At this stage, the cross terms involving different operators, such as $\Phi^{\mathcal{N}^{(Y_2 Y_2)}}_{2}\Phi^{\mathcal{N}^{(X_1X_1)}}_{1}$ vanish. Thus, after tracing out the appropriate environment register, the state simplifies neatly to:
    \begin{align}
    \rho_2 = &\frac{1}{2} 
    \bigg[
        |0\rangle\langle0| \otimes \Phi^{\mathcal{N}^{(Y_2Y_2)}}_{2}\Phi^{\mathcal{N}^{(Y_1Y_1)}}_{1} \left[\rho_S \right] + |0\rangle\langle1| \otimes \Phi^{\mathcal{N}^{(Y_2X_2)}}_{2}\Phi^{\mathcal{N}^{(Y_1X_1)}}_{1} \left[\rho_S \right] \nonumber \\      
        &+|1\rangle\langle0| \otimes \Phi^{\mathcal{N}^{(X_2Y_2)}}_{2}\Phi^{\mathcal{N}^{(X_1Y_1)}}_{1} \left[\rho_S \right] + |1\rangle\langle1| \otimes \Phi^{\mathcal{N}^{(X_2X_2)}}_{2}\Phi^{\mathcal{N}^{(X_1X_1)}}_{1} \left[\rho_S \right] \bigg]
    \end{align}

We continue this $K-1$ times, where in each step the mismatched terms will cancel out and we will be left with $K-1$ composition of the map as follows: 
    \begin{align}
    \label{eq:rho-K-one-run-nm}
        \rho_{K-1} 
        = &\frac{1}{2} 
        \bigg[
            |0\rangle\langle0| \otimes \bigcirc_{j=1}^{K-1} \Phi^{\mathcal{N}^{(Y_jY_j)}}_j \left[\rho_S \right]+ 
            |0\rangle\langle1| \otimes \bigcirc_{j=1}^{K-1} \Phi^{\mathcal{N}^{(Y_jX_j)}}_j \left[\rho_S \right] \nonumber \\&+  
            |1\rangle\langle0| \otimes \bigcirc_{j=1}^{K-1} \Phi^{\mathcal{N}^{(X_jY_j)}}_j \left[\rho_S \right]+
            |1\rangle\langle1| \otimes \bigcirc_{j=1}^{K-1} \Phi^{\mathcal{N}^{(X_jX_j)}}_j \left[\rho_S \right]
        \bigg]
    \end{align}

For the last collision we apply the control and anti-control operators $X_K^{(c)}$ and $Y_K^{(a)}$ and trace out the first (second) environment register if $K$ is odd (even). Resulting in the final state to be 
\begin{align}
    \label{eq:rho-K-one-run-nm}
        \rho_{K} 
        = &\frac{1}{2} 
        \bigg[
            |0\rangle\langle0| \otimes \Tr_{E_K}\left[Y_K\big(\bigcirc_{j=1}^{K-1} \Phi^{\mathcal{N}^{(Y_jY_j)}}_j \left[\rho_S \right]\big)Y_K^{\dag} \right]+ 
            |0\rangle\langle1| \otimes \Tr_{E_K}\left[Y_K\big(\bigcirc_{j=1}^{K-1} \Phi^{\mathcal{N}^{(Y_jX_j)}}_j \left[\rho_S \right]\big)X_K^{\dag} \right] \nonumber \\&+  
            |1\rangle\langle0| \otimes \Tr_{E_K}\left[X_K\big(\bigcirc_{j=1}^{K-1} \Phi^{\mathcal{N}^{(X_jY_j)}}_j \left[\rho_S \right]\big)Y_K^{\dag} \right]+
            |1\rangle\langle1| \otimes \Tr_{E_K}\left[X_K\big(\bigcirc_{j=1}^{K-1} \Phi^{\mathcal{N}^{(X_jX_j)}}_j \left[\rho_S \right]\big)X_K^{\dag} \right]
        \bigg]
    \end{align}
Finally, we measure the ancilla and system register with the observable $\sigma^x \otimes O$. This constitutes one run of Algorithm \ref{algo: nm-dynamics}. Now measuring the ancilla on $\sigma^x$, the first and last terms of Eq.~\eqref{eq:rho-K-one-run-nm} disappear, and so, the output of the $k$-${\mathrm{th}}$ run, 
    \begin{align}
    \mu_k = \frac{1}{2} \Tr &\Bigg[
        O~\bigg[\Tr_{E_K}\left[Y_K\big(\bigcirc_{j=1}^{K-1} \Phi^{\mathcal{N}^{(Y_jX_j)}}_j \left[\rho_S \right]\big)X_K^{\dag} \right] + \Tr_{E_K}\left[Y_K\big(\bigcirc_{j=1}^{K-1} \Phi^{\mathcal{N}^{(Y_jX_j)}}_j \left[\rho_S \right]\big)X_K^{\dag} \right]
        \bigg]
    \Bigg]
    \end{align}
Then, by the linearity of expectation, we have 
        \begin{align}
        \mathbb{E}[\mu_k]&=\frac{1}{\zeta^2} \left[ \Tr_{E_K}\left[\widetilde{U}_{S_K}\big(\bigcirc_{j=1}^{K-1} \widetilde{\Phi}^{\mathcal{N}}_j \left[\rho_S \right]\big)\widetilde{U}_{S_K}^{\dag} \right] \right]= \frac{1}{\zeta^2} \Tr\left[O 
                \widetilde{\mathcal{N}}_K \left[ \rho_S \right]\right],
        \end{align}
where $\zeta$ is as defined in Algorithm \ref{algo: nm-dynamics}. Thus, the outcome of each run is a random variable that in expectation value estimates the desired quantity (upto a multiplicative factor of $1/\zeta^2$). 

Since the observable $O$ has eigenvalues bounded within $[-\norm{O},\norm{O}]$, each individual outcome $\mu_k$ satisfies:
    \begin{align}
        -{\norm{O}\zeta^2 \leq \zeta^2\mu_k\leq \norm{O}\zeta^2}.
    \end{align}
After performing the experiment for $T$ independent runs, we have a collection of random variables $\{\mu_k\}_{k=1}^T$. Then, from Hoeffding's inequality,
$$
\mu=\dfrac{\zeta^2}{T}\sum_{k=1}^{T}\mu_k,
$$
satisfies
 \begin{align*}
        \Pr\Bigg[\bigg|\mu - \Tr\big[O~\widetilde{\mathcal{N}}_K[ \rho_S ] \big]\bigg| \geq \eps/2\Bigg] \leq 2\exp\left[-\dfrac{T\eps^2}{8 \zeta^4\norm{O}^2}\right].
        \end{align*}       
Thus, with probability at least $1-\delta$,
        \begin{equation} \label{eq:error in expectation and imprecise map-nm}
        \left|\mu - \Tr\left[O\widetilde{\mathcal{N}}_K[\rho_S ] \right]\right| 
        \leq \eps/2,
        \end{equation}
for 
\begin{equation*}
        \label {eq:number-of-repititions-nm}
        T\geq \dfrac{8\norm{O}^2\ln(2/\delta)\zeta^4}{\eps^2}.
\end{equation*}

Now from the statement of the Lemma, for any $j\in [1,K]$.
$$
\norm{U_j-\widetilde{U}_j}\leq \eps'= \dfrac{\varepsilon}{6K\norm{O}}.
$$
Then using Lemma \ref{lemma: nm-approximate K-collision map distance bound} and the triangle inequality, we obtain
\begin{align} 
\label{eq:error in expectation and true map-nm}
        \bigg|\mu - \Tr\big[O~\mathcal{N}_K [\rho_S] \big] \bigg|
        &\leq \bigg|\mu - \Tr\big[O\widetilde{\mathcal{N}}_K[\rho_S ] \big]\bigg|+ \bigg|\Tr\big[O\widetilde{\mathcal{N}}_K[\rho_S ] \big]
        - \Tr\big[O\mathcal{N}_K[\rho_S ] \big]\bigg|\\
        &\leq \eps/2+\eps/2=\eps.
        \end{align}
In order to estimate the circuit depth of Algorithm~\ref{algo: nm-dynamics} and the number of classical repetitions $T$, we need to find $\zeta$. Analogous to the Theorem \ref{thm: Algorithm 1 proof}, we use the Lemma \ref{lemma: lcu of time evolution} for each collision unitary and choose the maximum repetitions to be $r = \bigO(K\beta^2 \Delta t^2)$, which ensure $\zeta = \bigO(1)$ 
Consequently, the number of classical repetitions needed is
    \begin{equation}
        T =  \bigO\left( \dfrac{\norm{O}^2\log(1/\delta)}{\eps^2}\right).
    \end{equation}
Similarly, analogous to Theorem \ref{thm: Algorithm 1 proof}, for the appropriate choices of $r$ and $q$, we have the overall circuit depth per coherent run as
    \begin{equation}
        \tau_d = \bigO\left(\beta ^2K^2\Delta t^2 \dfrac{\log(\beta  K \norm{O}\Delta t/\eps)}{\log\log(\beta  K \norm{O}\Delta t/\eps)} + K\tau_{\rho_E}\right),
    \end{equation}
where the additive term $K \tau_{\rho_E}$ again appears as in each run of the circuit, on account of preparing the individual sub-environments, a total of $K$ times. This completes the proof.
\end{proof}
\end{document}